\documentclass{llncs}
\usepackage{amssymb,amsmath}
\usepackage{graphicx,latexsym,float}
\usepackage{algorithm}
\usepackage[noend]{algorithmic}
\usepackage{versions}
\usepackage{xifthen}

\usepackage[usenames,dvipsnames]{color}
\usepackage[dvipsnames,svgnames,table]{xcolor}

\usepackage{bm}

\usepackage{xspace}
\usepackage{afterpage}

\usepackage{microtype}

\excludeversion{FULL}
\includeversion{SHORT}
\newcommand{\nvsp}{\begin{SHORT}\vspace{-0.85 mm}\end{SHORT}}
\includeversion{FORPROOFSONLYEXCLUDETHIS}
\excludeversion{FORPROOFSONLYINCLUDETHIS}

\excludeversion{CUSTOM}

\newboolean{proof-lemonetodelta}
\setboolean{proof-lemonetodelta}{true}

\newboolean{proof-thmlowerbounddemandpaging}
\setboolean{proof-thmlowerbounddemandpaging}{true}

\newboolean{proof-thmlowerboundcompetitive}
\setboolean{proof-thmlowerboundcompetitive}{true}

\newboolean{proof-thmoptsmoothness}
\setboolean{proof-thmoptsmoothness}{true}

\newboolean{proof-thmcompetitivesmoothness}
\setboolean{proof-thmcompetitivesmoothness}{true}

\newboolean{proof-thmboundedmemorysmooth}
\setboolean{proof-thmboundedmemorysmooth}{true}

\newboolean{proof-thmlrusmoothness}
\setboolean{proof-thmlrusmoothness}{true}

\newboolean{proof-lemsuffix}
\setboolean{proof-lemsuffix}{true}

\newboolean{proof-lemphases}
\setboolean{proof-lemphases}{true}

\newboolean{proof-thmfwfsmoothness}
\setboolean{proof-thmfwfsmoothness}{true}

\newboolean{proof-thmfifosmoothness}
\setboolean{proof-thmfifosmoothness}{true}

\newboolean{proof-thmlowerboundrandomizeddemangpaging}
\setboolean{proof-thmlowerboundrandomizeddemangpaging}{true}

\newboolean{proof-thmlowerboundstronglycompetitive}
\setboolean{proof-thmlowerboundstronglycompetitive}{true}

\newboolean{proof-thmlowerboundpartitionequitable}
\setboolean{proof-thmlowerboundpartitionequitable}{true}

\newboolean{proof-thmlowerboundmark}
\setboolean{proof-thmlowerboundmark}{true}

\newboolean{proof-thmrandomsmoothness}
\setboolean{proof-thmrandomsmoothness}{true}

\newboolean{proof-thmeoasmoothness}
\setboolean{proof-thmeoasmoothness}{true}

\newboolean{proof-thmsmoothlrusmoothness}
\setboolean{proof-thmsmoothlrusmoothness}{true}

\newboolean{proof-lemcompetitivenesssmoothlru}
\setboolean{proof-lemcompetitivenesssmoothlru}{true}

\newboolean{proof-lemdecompositionsmoothlru}
\setboolean{proof-lemdecompositionsmoothlru}{true}

\newboolean{proof-thmsteplrumixed}
\setboolean{proof-thmsteplrumixed}{true}

\newboolean{proof-corsmoothlrumixed}
\setboolean{proof-corsmoothlrumixed}{true}

\newboolean{proof-thmcompetitivenesslrurandom}
\setboolean{proof-thmcompetitivenesslrurandom}{true}

\newboolean{proof-thmlrurandomsmoothness}
\setboolean{proof-thmlrurandomsmoothness}{true}

\includeversion{FULL}									%include the following four lines for the full version
\includeversion{CUSTOM}
\excludeversion{SHORT}
\renewcommand{\nvsp}{}
\usepackage{fullpage}

\setboolean{proof-lemonetodelta}{true}

\setboolean{proof-thmlowerbounddemandpaging}{true}

\setboolean{proof-thmlowerboundcompetitive}{true}

\setboolean{proof-thmoptsmoothness}{true}

\setboolean{proof-thmcompetitivesmoothness}{true}

\setboolean{proof-thmboundedmemorysmooth}{true}

\setboolean{proof-thmlrusmoothness}{true}

\setboolean{proof-lemsuffix}{false}

\setboolean{proof-lemphases}{false}

\setboolean{proof-thmfwfsmoothness}{true}

\setboolean{proof-thmfifosmoothness}{true}

\setboolean{proof-thmlowerboundrandomizeddemangpaging}{true}

\setboolean{proof-thmlowerboundstronglycompetitive}{false}

\setboolean{proof-thmlowerboundpartitionequitable}{false}

\setboolean{proof-thmlowerboundmark}{false}

\setboolean{proof-thmrandomsmoothness}{true}

\setboolean{proof-thmeoasmoothness}{false}

\setboolean{proof-thmsmoothlrusmoothness}{true}

\setboolean{proof-lemcompetitivenesssmoothlru}{true}

\setboolean{proof-lemdecompositionsmoothlru}{false}

\setboolean{proof-thmsteplrumixed}{false}

\setboolean{proof-corsmoothlrumixed}{true}

\setboolean{proof-thmcompetitivenesslrurandom}{true}

\setboolean{proof-thmlrurandomsmoothness}{false}

\newboolean{atleastoneomittedproof}
\ifthenelse{
\boolean{proof-lemonetodelta} \AND 
\boolean{proof-thmlowerbounddemandpaging} \AND
\boolean{proof-thmlowerboundcompetitive} \AND
\boolean{proof-thmoptsmoothness} \AND
\boolean{proof-thmcompetitivesmoothness} \AND
\boolean{proof-thmboundedmemorysmooth} \AND
\boolean{proof-thmlrusmoothness} \AND
\boolean{proof-lemsuffix} \AND
\boolean{proof-lemphases} \AND
\boolean{proof-thmfwfsmoothness} \AND
\boolean{proof-thmfifosmoothness} \AND
\boolean{proof-thmlowerboundrandomizeddemangpaging} \AND
\boolean{proof-thmlowerboundstronglycompetitive} \AND
\boolean{proof-thmlowerboundpartitionequitable} \AND
\boolean{proof-thmlowerboundmark} \AND
\boolean{proof-thmrandomsmoothness} \AND
\boolean{proof-thmeoasmoothness} \AND
\boolean{proof-thmsmoothlrusmoothness} \AND
\boolean{proof-lemcompetitivenesssmoothlru} \AND
\boolean{proof-lemdecompositionsmoothlru} \AND
\boolean{proof-thmsteplrumixed} \AND
\boolean{proof-corsmoothlrumixed} \AND
\boolean{proof-thmcompetitivenesslrurandom} \AND
\boolean{proof-thmlrurandomsmoothness}
}
{\setboolean{atleastoneomittedproof}{false}}
{\setboolean{atleastoneomittedproof}{true}}

\newboolean{allproofsomitted}
\ifthenelse{
\boolean{proof-lemonetodelta} \OR 
\boolean{proof-thmlowerbounddemandpaging} \OR
\boolean{proof-thmlowerboundcompetitive} \OR
\boolean{proof-thmoptsmoothness} \OR
\boolean{proof-thmcompetitivesmoothness} \OR
\boolean{proof-thmboundedmemorysmooth} \OR
\boolean{proof-thmlrusmoothness} \OR
\boolean{proof-lemsuffix} \OR
\boolean{proof-lemphases} \OR
\boolean{proof-thmfwfsmoothness} \OR
\boolean{proof-thmfifosmoothness} \OR
\boolean{proof-thmlowerboundrandomizeddemangpaging} \OR
\boolean{proof-thmlowerboundstronglycompetitive} \OR
\boolean{proof-thmlowerboundpartitionequitable} \OR
\boolean{proof-thmlowerboundmark} \OR
\boolean{proof-thmrandomsmoothness} \OR
\boolean{proof-thmeoasmoothness} \OR
\boolean{proof-thmsmoothlrusmoothness} \OR
\boolean{proof-lemcompetitivenesssmoothlru} \OR
\boolean{proof-lemdecompositionsmoothlru} \OR
\boolean{proof-thmsteplrumixed} \OR
\boolean{proof-corsmoothlrumixed} \OR
\boolean{proof-thmcompetitivenesslrurandom} \OR
\boolean{proof-thmlrurandomsmoothness}
}
{\setboolean{allproofsomitted}{false}}
{\setboolean{allproofsomitted}{true}}
\usepackage{thm-restate}

\usepackage{paralist}

\usepackage{tikz}
\usetikzlibrary{arrows,backgrounds,calc,chains,decorations,decorations.pathreplacing,decorations.pathmorphing,decorations.shapes,decorations.text,fit,matrix,patterns,positioning,shadows,shapes}
\usepackage{pgfplots} 
\usetikzlibrary{plotmarks}
\pgfplotsset{width=7cm,compat=1.3}
\usetikzlibrary{calc}

\begin{FORPROOFSONLYINCLUDETHIS}	
\end{FORPROOFSONLYINCLUDETHIS}

\newcommand\TikzHitProbabilitiesFig[4]{
	\begin{tikzpicture}
	\begin{axis}[#1, width=#2, height=#3,every axis legend/.append style={fill opacity=0.8,nodes={right}, xlabel={Age of requested page}, ylabel={Hit probability}},xtick={0,1,...,15}] 
		#4
	\end{axis}
	\end{tikzpicture}
}

\newcommand\AddCurve[4]{% ensure that the filename has a suffix, e.g. .txt
	\addplot[mark=#3, color=#2] plot file {curves/#1}; 
	\addlegendentry{#4};
}

\definecolor{color1}{RGB}{198, 219, 239}	%blue
\definecolor{color2}{RGB}{252, 146, 114}	%red
\definecolor{color3}{RGB}{116, 196, 118}	%green
\definecolor{color4}{RGB}{231, 41, 138}		%pink
\definecolor{color5}{RGB}{106, 81, 163}		%violett
\definecolor{color6}{RGB}{189, 0, 38}			%red
\definecolor{color7}{RGB}{8, 48, 107}			%blue
\definecolor{color8}{RGB}{0,0,0}					%black

\definecolor{myred}{rgb}{0.89,0.1,0.11}
\definecolor{mygreen}{rgb}{0.30,0.69,0.29}
\definecolor{myblue}{rgb}{0.22,0.49,0.72}

\newcommand{\polname}[1]{\textsc{#1}}
\newcommand{\equitable}{{\polname{Equitable}}\xspace}
\newcommand{\onlinemin}{{\polname{OnlineMin}}\xspace}
\newcommand{\partition}{{\polname{Partition}}\xspace}
\newcommand{\FWF}{{\polname{FWF}}\xspace}
\newcommand{\LRU}{{\polname{LRU}}\xspace}
\newcommand{\FIFO}{{\polname{FIFO}}\xspace}
\newcommand{\PARTITION}{{\polname{Partition}}\xspace}
\newcommand{\MARK}{{\polname{Mark}}\xspace}
\newcommand{\EQUITABLE}{{\polname{Equitable}}\xspace}
\newcommand{\RANDOM}{{\polname{Random}}\xspace}
\newcommand{\RAND}{{\polname{Random}}\xspace}
\newcommand{\lrurandom}{{\polname{LRU-Random}}\xspace}
\newcommand{\lrur}{{\polname{LRUR}}\xspace}
\newcommand{\smoothlru}{{\polname{Smoothed-LRU}}\xspace}
\newcommand{\steplru}{{\polname{Step-LRU}}\xspace}
\newcommand{\detsteplru}{{\polname{Det-Step-LRU}}\xspace}
\newcommand{\opt}{{\polname{OPT}}\xspace}
\newcommand{\adv}{{\polname{ADV}}\xspace}
\newcommand{\EOA}{{\polname{EOA}}\xspace}

\newcommand{\todo}[1]{{{\textcolor{Red}{(Todo: #1)}}}}
\renewcommand{\todo}[1]{}

\newcommand{\etal}{\textit{et~al.}\xspace}

\DeclareMathOperator*{\argmax}{arg\,max}

\tikzset{annotation/.style={draw=gray!50,fill=gray!20,line width=1pt}}
\tikzset{specialannotation/.style={draw=mygreen,fill=mygreen,rounded corners=0.5pt}}
\tikzset{klammer/.style={decorate,decoration={brace,amplitude=6pt}, line width=1.5pt}}

\newcommand{\crosswi}{2.3pt}
\newcommand{\wi}{4pt}
\newcommand{\wwi}{6.5pt}
\newcommand{\wwwi}{10.5pt}

\newcommand{\markx}[2]{\draw (#1 cm,2pt) -- (#1 cm,-2pt) node[anchor=north,fill=white] {#2}; \draw[dotted] (#1 cm, 2pt) -- (#1 cm, \yMax cm);}
\newcommand{\marky}[2]{\draw (2pt, #1 cm) -- (-2pt, #1 cm) node[anchor=east,fill=white] {#2}; \draw[dotted] (2pt, #1 cm) -- (\xMax cm, #1 cm);}

\newcommand{\drawcross}[2]{\draw[line width=1.5pt] (#1 cm+\crosswi, #2 cm-\crosswi) -- (#1 cm-\crosswi, #2 cm+\crosswi); \draw[line width=1.5pt] (#1 cm-\crosswi, #2 cm-\crosswi) -- (#1 cm+\crosswi, #2 cm + \crosswi);}
\newcommand{\drawcircle}[2]{\draw[line width=1.5pt, fill=gray!50] (#1 cm, #2 cm) circle (\crosswi);}

\newcommand{\markxy}[3]{\draw[annotation] (#1 cm, #2 cm) -- +(\wwi, \wwi) -- ++(\wi, \wi) node[anchor=south west,annotation] {#3}; \drawcross{#1}{#2}}
\newcommand{\markxylefta}[3]{\draw[annotation] (#1 cm, #2 cm) -- +(-\wwi, \wwi) -- ++(-\wi, \wi) node[anchor=south east,annotation] {#3}; \drawcross{#1}{#2}}
\newcommand{\markxybelow}[3]{\draw[annotation] (#1 cm, #2 cm) -- (#1 cm + \wwi, #2 cm - \wwi) -- (#1 cm + \wi, #2 cm - \wi) node[anchor=north west,annotation] {#3}; \drawcross{#1}{#2}}
\newcommand{\markxybelowcircle}[3]{\draw[annotation] (#1 cm, #2 cm) -- (#1 cm + \wwi, #2 cm - \wwi) -- (#1 cm + \wi, #2 cm - \wi) node[anchor=north west,annotation] {#3}; \drawcircle{#1}{#2}}
\newcommand{\markxyabove}[3]{\draw[annotation] (#1 cm, #2 cm) -- +(60:\wwwi) -- +(60:\wwi) node[anchor=south,annotation] {#3}; \drawcircle{#1}{#2}}

\newcommand{\markxyright}[3]{\draw[annotation,specialannotation] (#1 cm, #2 cm) -- +(0:\wwwi) -- ++(0:\wwi) node[anchor=west,fill=gray!20,draw=mygreen] {#3}; }
\newcommand{\markxyleft}[3]{\draw[annotation,specialannotation] (#1 cm, #2 cm) -- +(180:\wwwi) -- ++(180:\wwi) node[anchor=east,fill=gray!20,draw=mygreen] {#3}; }

\newcommand{\markxyinterval}[3]{\fill[specialannotation] (#1 cm-0.7mm, #2 cm+0.7mm) rectangle (#1 cm+0.7mm, #3 cm-0.7mm);}

\newcommand{\separationy}[3]{%\draw[color=green, line width=1.8pt] (0, #1 cm) -- (\xMax cm, #1 cm); 
\draw[color=violet,klammer] (#2, #1 cm) -- (#2, 0 cm) node[anchor=west,midway,right,fill=white,xshift=2.5mm, text width=2.3cm,yshift=2.1mm] {\small #3};}
\newcommand{\separationybelow}[3]{%\draw[color=red, line width=1.8pt] (0, #1 cm) -- (\xMax cm, #1 cm); 
\draw[color=red,klammer] (#2, \yMax cm) -- (#2, #1 cm) node[anchor=west,midway,right,fill=white,xshift=2.5mm, text width=2.3cm, yshift=3.95mm] {\small #3};}

\newcommand{\separationx}[4]{%\draw[color=blue, line width=1.8pt] (#1 cm, 0) -- (#1 cm, \yMax cm); 
\draw[color=blue,klammer] (#1 cm, #3 cm) -- (#2 cm, #3 cm)  node[midway, above,fill=white, xshift=0mm, yshift=2mm] {\small #4};}

\newcommand{\xOPT}{1}
\newcommand{\xHk}{1.75}
\newcommand{\xHkTwo}{2.65}
\newcommand{\xK}{4.1}

\newcommand{\xGamma}{6.65}
\newcommand{\xinfty}{7.65}
\newcommand{\xAnnot}{9.075}
\newcommand{\xAnnotTwo}{8.925}
\newcommand{\xMax}{8.85}

\newcommand{\ySLow}{0.4}
\newcommand{\ySTwo}{1.15}
\newcommand{\ySTwoHigh}{1.4}
\newcommand{\ySHk}{1.9}
\newcommand{\yLRURandom}{2.58}
\newcommand{\ySK}{2.95}
\newcommand{\ySKTwo}{3.35}
\newcommand{\ySAnnot}{3.7}
\newcommand{\ySc}{4.5}
\newcommand{\yHk}{5.5}
\newcommand{\yK}{6.5}
\newcommand{\yInfty}{7.5}

\newcommand{\yMax}{7.9}
\newcommand{\yAnnot}{7.95}
\newcommand{\yAnnotTwo}{5}

\newcommand{\robustnessvscompetitiveness}{
\begin{tikzpicture}[yscale=0.6, text height=1.5ex,text depth=.25ex,>=stealth]
	\draw[line width=1pt] (0,0) -- (0, \yMax) -- node[below=-10pt, xshift=-14.5mm, fill=white] {\textbf{Smoothness}} (0, \yMax);
	\draw[line width=1pt] (-0.5pt,0) -- (\xMax, 0) -- node[left=-65pt,yshift=-3.3mm,fill=white] {\textbf{Competitiveness}} (\xMax, 0);
	
	\markx{\xOPT}{$1$} %OPT
	\markx{\xHk}{$H_k$} %randomized strongly-competitive
	\markx{\xHkTwo}{$2H_k{-}1$} %randomized strongly-competitive
	\markx{\xK}{$k$} %randomized strongly-competitive
	\markx{\xGamma}{$\gamma$} %randomized strongly-competitive
	\markx{\xinfty}{$\infty$} %randomized strongly-competitive
	
	\marky{\ySTwo}{$(1, 2\delta)$}
	\marky{\ySLow}{$(1, (1+\frac{k}{2k-1})\delta))$}
	\marky{\ySHk}{$(1, \delta H_k)$}
	\marky{\ySK}{$(1, \delta (k+1))$}
	\marky{\ySc}{$(1, \delta \beta)$}
	\marky{\yHk}{$({\cal O}(H_k), \gamma_2)$}
	\marky{\yK}{$(k, \gamma_1)$}
	\marky{\yInfty}{$\infty$}

	\separationx{\xK}{\xGamma}{\yAnnot}{Bounded mem. $(\alpha, \beta)$-smooth det. demand paging}

	\markxyleft{\xHk}{\ySAnnot}{Randomized strongly comp.}
	\markxyinterval{\xHk}{\yHk}{\ySHk}
	\markxyright{\xK}{\yAnnotTwo}{Deterministic strongly comp.}
	\markxyinterval{\xK}{\yK}{\ySK}

	\separationybelow{\ySK}{\xAnnot}{{Deterministic demand paging or competitive}}
	\separationy{\ySc}{\xAnnotTwo}{{Smooth\\algorithms}}

	\markxylefta{\xK}{\ySKTwo}{\FWF}
	\markxy{\xK}{\ySK}{\RANDOM, \LRU}
	\markxybelowcircle{\xK}{\yLRURandom}{$\lrurandom_{k=2}$}
	\markxybelow{\xOPT}{\ySTwo}{\opt}
	\markxybelow{\xK}{\ySTwoHigh}{$\smoothlru_{k+i,i}$}
	
	\markxy{\xK}{\yK}{\FIFO}
	
	\markxyabove{\xHk}{\yHk}{\PARTITION, \EQUITABLE}
	\markxybelowcircle{\xHkTwo}{\yHk}{\MARK}

	\markxy{\xinfty}{\ySLow}{\EOA}
	
\end{tikzpicture}}

\begin{document}  

\begin{FORPROOFSONLYINCLUDETHIS}
\section*{Appendix}
\end{FORPROOFSONLYINCLUDETHIS}
\begin{FORPROOFSONLYEXCLUDETHIS}

\title{On the Smoothness of Paging Algorithms}
\titlerunning{On the Smoothness of Paging Algorithms}  
\author{Jan Reineke\inst{1} \and Alejandro Salinger\inst{2}\thanks{Most of the reported work was carried out while this author was a postdoctoral researcher at Saarland~University.}}
%\author{Jan Reineke \and Alejandro Salinger}

\authorrunning{Reineke and Salinger}
\tocauthor{Jan Reineke, Alejandro Salinger}

\institute{Department of Computer Science, Saarland University, Saarbr\"ucken, Germany\\
%\email{\{reineke,salinger\}@cs.uni-saarland.de}
\email{reineke@cs.uni-saarland.de}
\and
SAP SE, Walldorf, Germany\\
%\email{salinger@cs.uni-saarland.de}
\email{alejandro.salinger@sap.com}
}
\maketitle

\begin{abstract}
We study the smoothness of paging algorithms.
How much can the number of page faults increase due to a perturbation of the request sequence?
We call a paging algorithm \emph{smooth} if the maximal increase in page faults is proportional to the number of changes in the request sequence.
We also introduce quantitative smoothness notions that measure the smoothness of an algorithm.

We derive lower and upper bounds on the smoothness of deterministic and randomized demand-paging and competitive algorithms. 
Among strongly-competitive deterministic algorithms \LRU matches the lower bound, while \FIFO matches the upper bound. 

Well-known randomized algorithms like \partition, \equitable, or \MARK are shown not to be smooth.
We introduce two new randomized algorithms, called \smoothlru and \lrurandom.
\smoothlru allows to sacrifice competitiveness for smoothness, where the trade-off is controlled by a parameter.
\lrurandom is at least as competitive as any deterministic algorithm while smoother.
\end{abstract}

%\begin{SHORT}\newpage
%\setcounter{page}{1}
%\end{SHORT}
\section{Introduction}

Due to their strong influence on system performance, paging algorithms have been studied extensively since the 1960s.
Early studies were based on probabilistic request models~\cite{Belady66,Mattson70,Aho71}.
In their seminal work, Sleator and Tarjan~\cite{Sleator85} introduced the notion of competitiveness, which relates the performance of an online algorithm to that of the optimal offline algorithm.
By now, the competitiveness of well-known deterministic and randomized paging algorithms is well understood, and various optimal online algorithms~\cite{McGeoch91,Achlioptas00} have been identified.

In this paper, we study the \emph{smoothness} of paging algorithms.
We seek to answer the following question:
How strongly may the performance of a paging algorithm change when the sequence of memory requests is slightly perturbed?
This question is relevant in various domains:
Can the cache performance of an algorithm suffer significantly due to the occasional execution of interrupt handling code?
Can the execution time of a safety-critical real-time application be safely and tightly bounded in the presence of interference on the cache?
Can secret-dependent memory requests have a significant influence on the number of cache misses of a cryptographic protocol and thus give rise to a timing side-channel attack?\looseness=-1

We formalize the notion of smoothness by identifying the performance of a paging algorithm with the number of page faults and the magnitude of a perturbation with the edit distance between two request sequences.

We show that for any deterministic, demand-paging or competitive algorithm, a single additional memory request may cause $k+1$ additional faults, where $k$ is the size of the cache.
Least-recently-used (\LRU) matches this lower bound, indicating that there is no trade-off between competitiveness and smoothness for deterministic algorithms.
In contrast, First-in first-out (\FIFO) is shown to be least smooth among all strongly-competitive deterministic algorithms.

Randomized algorithms have been shown to be more competitive than deterministic ones.
We derive lower bounds for the smoothness of randomized, demand-paging and randomized strongly-competitive algorithms that indicate that randomization might also help with smoothness.
However, we show that none of the well-known randomized algorithms \MARK, \EQUITABLE, and \PARTITION is smooth.
The simple randomized algorithm that evicts one of the cached pages uniformly at random is shown to be as smooth as~\LRU, but not more.

We then introduce a new parameterized randomized algorithm, \smoothlru, that allows to sacrifice competitiveness for smoothness.
For some parameter values \smoothlru is smoother than any randomized strongly-competitive algorithm can possibly be, indicating a trade-off between smoothness and competitiveness for randomized algorithms.
This leaves the question whether there is a randomized algorithms that is smoother than any deterministic algorithm without sacrificing competitiveness.
We answer this question in the affirmative by introducing \lrurandom, a randomized version of \LRU that evicts older pages with a higher probability than younger ones.
We show that \lrurandom is smoother than any deterministic algorithm for $k=2$. While we conjecture that this is the case as well for general $k$, this remains an open problem.

The notion of smoothness we present is not meant to be an alternative to competitive analysis for the evaluation of the performance of a paging algorithm; rather, it is a complementary quantitative measure that provides guarantees about the performance of an algorithm under uncertainty of the input. 
In general, smoothness is useful in both testing and verification:
\begin{compactitem}
	\item In testing: if a system is smooth, then a successful test run is indicative of the system's correct behavior not only on the particular test input, but also in its neighborhood.
	\item In verification, systems are shown to behave correctly under some assumption on their environment. Due to incomplete environment specifications, operator errors, faulty implementations, or other causes, the environment assumption may not always hold completely. In such a case, if the system is smooth, ``small'' violations of the environment assumptions will, in the worst case, result in ``small'' deviations from correct behavior.
\end{compactitem}
An example of the latter case that motivates our present work appears in safety-critical real-time systems, where static analyses are employed to derive guarantees on the worst-case execution time (WCET) of a program on a particular microarchitecture~\cite{Wilhelm08}.
While state-of-the-art WCET analyses are able to derive fairly precise bounds on execution times,
they usually only hold for the uninterrupted execution of a single program with no interference from the environment whatsoever.
These assumptions are increasingly difficult to satisfy with the adoption of preemptive scheduling or even multi-core architectures, which may introduce interference on shared resources such as caches and buses.
Given a smooth cache hierarchy, it is possible to separately analyze the effects of interference on the cache, e.g. due to interrupts, preemptions, or even co-running programs on other cores.
Our results may thus inform the design and analysis of microarchitectures for real-time systems~\cite{Axer14}.\looseness=-1

Interestingly, our model shows a significant difference between \LRU and \FIFO, two algorithms whose theoretical performance has proven difficult to separate.

Our results are summarized in Table~\ref{tab:results}. 
An algorithm $A$ is $(\alpha, \beta, \delta)$-smooth, if the number of page faults $A(\sigma')$ of $A$ on request sequence $\sigma'$ is bounded by $\alpha\cdot A(\sigma) + \beta$ whenever $\sigma$ can be transformed into $\sigma'$ by at most $\delta$ insertions, deletions, or substitutions of individual requests.
Often, our results apply to a generic value of $\delta$.
In such cases, we express the smoothness of a paging algorithm by a pair $(\alpha,\beta)$, where $\alpha$ and $\beta$ are functions of $\delta$, and $A$ is $(\alpha(\delta), \beta(\delta), \delta)$-smooth for every $\delta$.
Usually, the smoothness of an algorithm depends on the size of the cache, which we denote by $k$.
As an example, under \LRU the number of faults may increase by at most $\delta(k+1)$, where $\delta$ is the number of changes in the sequence.
A precise definition of these notions is given in Section~\ref{sec:smoothnesssec}.

 %\todo{We should briefly mention here or earlier what $(\alpha,\beta)$ means here to give a more concrete description of what we mean by smoothness and better understanding of the results in the table}

\begin{CUSTOM}
\ifthenelse{\boolean{atleastoneomittedproof}}
{For readability, we place \ifthenelse{\boolean{allproofsomitted}}{all}{some of} the proofs of our results in the appendix.}
{}
\end{CUSTOM}

\begin{table}
\begin{SHORT}\vspace{-5mm}\end{SHORT}
\begin{center}
\caption{Upper and lower bounds on the smoothness of paging algorithms. In the table, $k$ is the size of the cache, $\delta$ is the distance between input sequences, $H_k$ denotes the $k^{th}$ harmonic number, and $\gamma$ is an arbitrary constant.\label{tab:results}}
\begin{SHORT}\vspace{-2mm}\end{SHORT}
\begin{tabular}{l | c | c}
\textbf{Algorithm} & \textbf{Lower bound} & \textbf{Upper bound} \\\hline\hline
Deterministic, demand-paging & $(1,\delta(k+1))$ & $\infty$ \\
Det. $c$-competitive with additive constant $\beta$ & $(1,\delta(k+1))$ & $(c,2\delta c+\beta)$ \\
Deterministic, strongly-competitive & $(1,\delta(k+1))$ & $(k,2\delta k)$ \\\hline
Optimal offline & $(1,2\delta)$ & $(1,2\delta)$\\ \hline
\LRU & $(1,\delta(k+1))$ & $(1,\delta(k+1))$\\ 
\FWF & $(1,2\delta k)$ & $(1,2\delta k)$\\
\FIFO & $(k, \gamma, 1)$ & $(k,2\delta k)$\\\hline\hline
Randomized, demand-paging & $(1,H_k+\frac{1}{k},1)$ & $\infty$ \\
Randomized, strongly-competitive & $(1,\delta H_k)$ & $(H_k,2\delta H_k)$  \\\hline
\equitable, \partition & $(1+\epsilon,\gamma,1)$ & $(H_k,2\delta H_k)$ \\
\MARK & $(\Omega(H_k), \gamma, 1)$ & $(2H_k-1,\delta (4H_k-2))$ \\
%\MARK & $(\max_{1 < \ell < k}\left\{\frac{\ell(1+H_k-H_\ell)}{\ell-1+H_k-H_{\ell-1}}\right\}, \gamma, 1)$ & $(2H_k-1,\delta (4H_k-2))$ \\
\RAND & $(1,\delta(k+1))$ & $(1,\delta(k+1))$ \\
Evict-On-Access & $(1,\delta(1+\frac{k}{2k-1}))$ & $(1,\delta(1+\frac{k}{2k-1}))$ \\
$\smoothlru_{k,i}$ & $(1,\delta(\frac{k+i}{2i+1}+1))$ & ${(1,\delta(\frac{k+i}{2i+1}+1))}$ \\
\todo{add lru-random}
\end{tabular}
\end{center}
\begin{SHORT}\vspace{-15mm}\end{SHORT}
\end{table}

\section{Related Work}
\begin{SHORT}\vspace{-2mm}\end{SHORT}
\subsection{Notions of Smoothness}
\begin{SHORT}\vspace{-1.5mm}\end{SHORT}

%In engineering, smoothness has always been an important design goal.
Robust control is a branch of control theory that explicitly deals with uncertainty in its approach to controller design. %\looseness=-1 
Informally, a controller designed for a particular set of parameters is said to be robust if it would also work well under a slightly different set of assumptions.
In computer science, the focus has long been on the binary property of correctness, as well as on average- and worst-case performance. 
Lately, however, various notions of smoothness have received increasing attention:
Chaudhuri~\etal~\cite{Chaudhuri12} develop analysis techniques to determine whether a given program computes a \emph{Lipschitz-continuous} function.
Lipschitz continuity is a special case of our notion of smoothness.
Continuity is also strongly related to differential privacy~\cite{Dwork06}, where the result of a query may not depend strongly on the information about any particular individual.
Differential privacy proofs with respect to cache side channels~\cite{Doychev15} may be achievable in a compositional manner for caches employing smooth paging algorithms.

Doyen~\etal~\cite{Doyen10} consider the robustness of sequential circuits.
They determine how long into the future a single disturbance in the inputs of a sequential circuit may affect the circuit's outputs.
Much earlier, but in a similar vein, Kleene~\cite{Kleene56}, Perles, Rabin, Shamir~\cite{Perles63}, and Liu~\cite{Liu63} developed the theory of definite events and definite automata.
The outputs of a definite automaton are determined by a fixed-length suffix of its inputs.
Definiteness is a sufficient condition for smoothness.\looseness=-1

The work of Reineke and Grund~\cite{Reineke13} is closest to ours: they study the maximal difference in the number of page faults on the \emph{same} request sequence starting from two different initial states for various deterministic paging algorithms.
In contrast, here, we study the effect of differences in the request sequences on the number of faults.
Also, in addition to only studying particular deterministic algorithms as in~\cite{Reineke13}, in this paper we determine smoothness properties that apply to classes of algorithms, such as all demand-paging or strongly-competitive ones, as well as to randomized algorithms.
One motivation to consider randomized algorithms in this work are recent efforts to employ randomized caches in the context of hard real-time systems~\cite{Cazorla13}.

\begin{SHORT}\vspace{-2mm}\end{SHORT}
\subsection{The Paging Problem}
\begin{SHORT}\vspace{-1mm}\end{SHORT}

Paging models a two-level memory system with a small fast memory known as cache, and a large but slow memory, usually referred to simply as memory. During a program's execution, data is transferred between the cache and memory in units of data known as pages. The size of the cache in pages is usually referred to as $k$. The size of the memory can be assumed to be infinite. The input to the paging problem is a sequence of page requests which must be made available in the cache as they arrive. When a request for a page arrives and this page is already in the cache, then no action is required. This is known as a \emph{hit}. Otherwise, the page must be brought from memory to the cache, possibly requiring the eviction of another page from the cache. This is known as a \emph{page fault} or \emph{miss}. %\footnote{We use the terms fault, page fault, and miss interchangeably throughout the paper.}. 
A paging algorithm must decide which pages to keep in the cache in order to minimize the number of faults.

A paging algorithm is said to be \emph{demand paging} if it only evicts a page from the cache upon a fault with a full cache. Any non-demand paging algorithm can be made to be demand paging without sacrificing performance~\cite{borodin98}. %Hence, demand paging algorithms are also known as cache eviction policies, since the algorithm is fully determined by the eviction policy on a page fault. 

In general, paging algorithms must make decisions as requests arrive, with no knowledge of future requests. That is, paging is an online problem. The most prevalent way to analyze online algorithms is competitive analysis~\cite{Sleator85}. In this framework, the performance of an online algorithm is measured against an algorithm with full knowledge of the input sequence, known as optimal offline or \opt. We denote by $A(\sigma)$ the number of misses of an algorithm when processing the request sequence $\sigma$. A paging algorithm $A$ is said to be $c$-competitive if for all sequences $\sigma$, $A(\sigma) \le c\cdot\opt(\sigma) + \beta$, where $\beta$ is a constant independent of $\sigma$. The \emph{competitive ratio} of an algorithm is the infimum over all possible values of $c$ satisfying the inequality above. An algorithm is called competitive if it has a constant competitive ratio and \emph{strongly competitive} if its competitive ratio is the best possible~\cite{McGeoch91}.\looseness=-1

Traditional paging algorithm are Least-recently-used (\LRU)---evict the page in the cache that has been requested least recently--- and First-in first-out (\FIFO)---evict the page in the cache that was brought into cache the earliest. Another simple algorithm often considered is Flush-when-full (\FWF)---empty the cache if the cache is full and a fault occurs. These algorithms are $k$-competitive, which is the best ratio that can be achieved for deterministic online algorithms~\cite{Sleator85}. An optimal offline algorithm for paging is Furthest-in-the-future, also known as Longest-forward-distance and Belady's algorithm~\cite{Belady66}. This algorithm evicts the page in the cache that will be requested at the latest time in the future.

%A random paging algorithm is $c$-competitive if there exists $\beta$ such that for all sequences $\sigma$, $E[A(\sigma)- c\cdot\opt(\sigma)] \le \beta$. The definition and results differ depending on the type of adversary considered: an \emph{oblivious} adversary must choose an input sequence in advance (independent of the choices of the online algorithm) but can serve it offline. An \emph{adaptive online} adversary can pick each request after learning the decision of the online algorithm, but must also serve the sequence online. An \emph{adaptive offline} adversary can choose the sequence as the adaptive online, but can also serve the requests offline. If an algorithm has competitive ratios $c_1, c_2, c_3$ under the oblivious, adaptive online, and adaptive offline adversaries, then $c1\le c_2 \le c_3$~\cite{borodin98}.

%Lower bounds under weak adversaries carry to more powerful adversaries, whereas upper bound carry in the other direction. 
A competitive ratio less than $k$ can be achieved by the use of randomization.
Important randomized paging algorithms are \RAND---evict a page chosen uniformly at random--- and \MARK~\cite{Fiat91}---mark a page when it is unmarked and requested, and upon a fault evict a page chosen uniformly at random among unmarked pages (unmarking all pages first if no unmarked pages remain). \RAND achieves a competitive ratio of $k$, while \MARK's competitive ratio is $2H_k-1$, where $H_k = \sum_{i=1}^k \frac{1}{i}$ is the $k^{th}$ harmonic number.  The strongly-competitive algorithms \partition~\cite{McGeoch91} and \equitable~\cite{Achlioptas00} achieve the optimal ratio of $H_k$.\looseness=-1

%Competitive analysis has been criticized for being too pessimistic and failing to provide a clear distinction among online algorihtms. Several alternatives forms of measuring the performance of online algorithms have been proposed (see~\cite{dorrigiv05,reza_thesis10} for a survey).

\begin{SHORT}\vspace{-2mm}\end{SHORT}
\section{Smoothness of Paging Algorithms}\label{sec:smoothnesssec}
\begin{SHORT}\vspace{-2mm}\end{SHORT}

%\enlargethispage{2\baselineskip}

We now formalize the notion of smoothness of paging algorithms.  
We are interested in answering the following question:
How does the number of misses of a paging algorithm vary as its inputs vary? 
We quantify the similarity of two request sequences by their edit~distance:

\begin{SHORT}\vspace{-0.6mm}\end{SHORT}
\begin{definition}[Distance]
Let $\sigma=x_1,\ldots,x_n$ and $\sigma'=x'_1,\ldots,x'_m$  be two request sequences. Then we denote by $\Delta(\sigma,\sigma')$ their edit distance, defined as the minimum number of substitutions, insertions, or deletions to transform $\sigma$ into~$\sigma'$. 
\end{definition}
\begin{SHORT}\vspace{-0.6mm}\end{SHORT}
This is also referred to as the \emph{Levenshtein distance}.
%
%\todo{consider/discuss other distances. e.g. allowing only additions (relevant in multi-core and preemption case) or only substitutions (possibly relevant in case of unknown requests)}
%
Based on this notion of distance we define $(\alpha, \beta, \delta)$-smoothness:
\begin{SHORT}\vspace{-0.6mm}\end{SHORT}
\begin{definition}[$(\alpha,\beta,\delta)$-smoothness]
Given a paging algorithm $A$, we say that $A$ is $(\alpha,\beta,\delta)$-smooth, if for all pairs of sequences $\sigma,\sigma'$ with $\Delta(\sigma,\sigma') \leq \delta$, 
\[A(\sigma')\le \alpha \cdot A(\sigma)+\beta\]
\end{definition}
\begin{SHORT}\vspace{-0.6mm}\end{SHORT}
For randomized algorithms, $A(\sigma)$ denotes the algorithm's \emph{expected} number of faults when serving~$\sigma$.

An algorithm that is $(\alpha, \beta, \delta)$-smooth may also be $(\alpha', \beta', \delta)$-smooth for $\alpha' > \alpha$ and $\beta' < \beta$.
As the multiplicative factor $\alpha$ dominates the additive constant $\beta$ in the long run, when analyzing the smoothness of an algorithm, we first look for the minimal $\alpha$ such that the algorithm is $(\alpha, \beta, \delta)$-smooth for any $\beta$.
%Upper and lower bounds on the smoothness of an algorithm determined in this paper are according to the lexicographical order. % $(\alpha, \beta)$.
%\enlargethispage{\baselineskip}
%when analyzing upper bounds, we first minimize $\alpha$ and then determine the corresponding $\beta$.

We say that an algorithm is \emph{smooth} if it is $(1,\beta,1)$-smooth for some~$\beta$.
In this case, the maximal increase in the number of page faults is proportional to the number of changes in the request sequence.
This is called \emph{Lipschitz continuity} in mathematical analysis. 
For smooth algorithms, we also analyze the Lipschitz constant, i.e, the additive part $\beta$ in detail, otherwise we concentrate the analysis on the multiplicative factor $\alpha$.

We use the above notation when referring to a specific distance $\delta$. 
For a generic value of $\delta$ we omit this parameter and express the smoothness of a paging algorithm with a pair $(\alpha, \beta)$, where both $\alpha$ and $\beta$ are functions of $\delta$.

\nvsp
\begin{definition}[$(\alpha,\beta)$-smoothness]
Given a paging algorithm $A$, we say that $A$ is $(\alpha,\beta)$-smooth, if for all pairs of sequences $\sigma,\sigma'$, 
\nvsp
\[A(\sigma')\le \alpha(\delta) \cdot A(\sigma)+\beta(\delta),\]
\nvsp
\nvsp
where  $\alpha$ and $\beta$ are functions, and $\delta = \Delta(\sigma,\sigma')$.
\end{definition}

Often, it is enough to determine the effects of one change in the inputs to characterize the smoothness of an algorithm~$A$.
\nvsp
\end{FORPROOFSONLYEXCLUDETHIS}
\begin{restatable}[]{lemma}{lemonetodelta}
If $A$ is $(\alpha,\beta,1)$-smooth, then $A$ is  $(\alpha^\delta,\beta\sum_{i=0}^{\delta-1}\alpha^i)$-smooth. 
\end{restatable}
\begin{SHORT}
All proofs can be found in the full version of this paper~\cite{smoothness_full}.
\end{SHORT}

\nvsp
\begin{FULL}
\ifthenelse{\boolean{proof-lemonetodelta}}{%% proof of lemonetodelta
\begin{proof}
By induction on $\delta$. The case $\delta=1$ is trivial. Assume the hypothesis is true for $1< \delta \le h$.  Let $\sigma_{h+1}$ and $\sigma$ be any pair of sequences such that $\Delta(\sigma,\sigma_{h+1})=h+1$. Then there exists a sequence $\sigma_h$ such that $\Delta(\sigma,\sigma_h)=h$ and  $\Delta(\sigma_h,\sigma_{h+1})=1$. Since $A$ is  $(\alpha,\beta,1)$-smooth, then 
$A(\sigma_{h+1})\le \alpha A(\sigma_h)+\beta$.  By the inductive hypothesis,
 $A(\sigma_h)\le \alpha^h A(\sigma)+\beta\sum_{i=0}^{h-1}\alpha^i$. Therefore,  $A(\sigma_{h+1})\le \alpha(\alpha^h A(\sigma)+\beta\sum_{i=0}^{h-1}\alpha^i)+\beta=\alpha^{h+1} A(\sigma)+\beta\sum_{i=0}^{h}\alpha^i$, and thus $A$ is  $(\alpha^\delta,\beta\sum_{i=0}^{\delta-1}\alpha^i)$-smooth.\looseness=-1\qed
%For any pair of sequences $\sigma$ and $\sigma_\delta$, if $\Delta(\sigma, \sigma_\delta)=\delta$ then there exist sequences $\sigma_1,\sigma_2,\ldots,\sigma_{\delta-1}$ such that $\Delta(\sigma,\sigma_{1})=1$, $\Delta(\sigma_{\delta-1},\sigma_{\delta})=1$, and  $\Delta(\sigma_i,\sigma_{i+1})=1$ for all $1\le i < \delta-1$. Since $A$ is  $(\alpha,\beta,1)$-smooth then by composition of the inequalities 
\end{proof}}{}
\end{FULL}
\begin{restatable}[]{corollary}{lemcoronetodelta}
\label{cor:1delta}
If $A$ is $(1,\beta,1)$-smooth, then $A$ is  $(1,\delta\beta)$-smooth. 
\end{restatable}

\begin{FORPROOFSONLYEXCLUDETHIS}

\begin{SHORT}\vspace{-4mm}\end{SHORT}
\section{Smoothness of Deterministic Paging Algorithms}
\begin{SHORT}\vspace{-2mm}\end{SHORT}

\subsection{Bounds on the Smoothness of Deterministic Paging Algorithms}
\nvsp

Before considering particular deterministic online algorithms, we determine upper and lower bounds for several important classes of algorithms.
Many natural algorithms are demand paging. 
\nvsp
\end{FORPROOFSONLYEXCLUDETHIS}
\begin{restatable}[Lower bound for deterministic, demand-paging algorithms]{theorem}{thmlowerbounddemandpaging}%
\label{thm:lowerbounddemandpaging}%
No deterministic,\hspace{1pt}demand-paging algorithm is $(1{,}\hspace{0.025pt}\delta (k{+}1{-}\epsilon))$-smooth for any~$\epsilon\!>\!0$.
%No deterministic, demand-paging algorithm is $(1, \delta (k+1-\epsilon), \delta)$-smooth for any $\epsilon > 0$ and $\delta \geq 1$.
\end{restatable}
\nvsp
\begin{SHORT}% \todo{this should also be included in custom without this proof}
The idea of the proof is to first construct a sequence of length $k+\delta(k+1)$ containing $k+1$ distinct pages, such that $A$ faults on every request.
This sequence can then be transformed into a sequence containing only $k$ distinct pages by removing all requests to the page that occurs least often, which reduces the overall number of misses to $k$, while requiring at most $\delta$ changes.
\end{SHORT}
\begin{FULL}
\ifthenelse{\boolean{proof-thmlowerbounddemandpaging}}{%% proof of thmlowerbounddemandpaging
\begin{proof}
Let $A$ be any deterministic, demand-paging algorithm. 
Using $k+1$ distinct pages, we can construct a sequence $\sigma_A(\delta)$ of length $k+\delta(k+1)$ such that $A$ faults on every request: first request the $k+1$ distinct pages in any order; then arbitrarily extend the sequence by requesting the page that $A$ has just evicted. 
Let $p$ be the page that occurs least frequently in $\sigma_A(\delta)$.
By removing all requests to $p$ from $\sigma_A(\delta)$, we obtain a sequence $\sigma'_A(\delta)$ that consists of $k$ distinct pages only.
By assumption $A$ is demand paging. 
Thus, $A$ incurs only $k$ page faults on the entire sequence.
Assume for a contradiction that $A$ is $(1, \delta(k+1-\epsilon))$-smooth for some $\epsilon > 0$.
Then, we have by definition:
\[A(\sigma_A(\delta)) \leq 1\cdot A(\sigma'_A(\delta)) + \Delta(\sigma'_A(\delta), \sigma_A(\delta))\cdot(k+1-\epsilon).\]
Clearly, $p$ occurs at most $\left\lfloor \frac{k+\delta(k+1)}{k+1} \right\rfloor = \delta$ times in $\sigma_A(\delta)$.
So $\Delta(\sigma'_A(\delta), \sigma_A(\delta)) \leq \delta$, and we get:
\begin{align*}
 k+\delta(k+1) &\leq  k + \delta\cdot(k+1-\epsilon)\\
\Leftrightarrow \epsilon &\leq 0, 
\end{align*}
which contradicts the assumption that $\epsilon > 0$.\qed
\end{proof}}
{
The idea of the proof is to first construct a sequence of length $k+\delta(k+1)$ containing $k+1$ distinct pages, such that $A$ faults on every request.
This sequence can then be transformed into a sequence containing only $k$ distinct pages by removing all requests to the page that occurs least often, which reduces the overall number of misses to $k$, while requiring at most $\delta$ changes.
}
\end{FULL}
\begin{FORPROOFSONLYEXCLUDETHIS}
While most algorithms are demand paging, it is not a necessary condition for an algorithm to be competitive, as demonstrated by \FWF.
However, we obtain the same lower bound for competitive algorithms as for demand-paging ones.
\nvsp
\end{FORPROOFSONLYEXCLUDETHIS}
\begin{restatable}[Lower bound for deterministic, competitive paging algorithms]{theorem}{thmlowerboundcompetitive}
\label{thm:lowerboundcompetitive}
No deterministic, competitive paging algorithm is $(1, \delta(k+1-\epsilon))$-smooth for any $\epsilon > 0$.
\end{restatable}
\nvsp
\begin{FULL}
\ifthenelse{\boolean{proof-thmlowerboundcompetitive}}{%% proof of thmlowerboundcompetitive
\begin{proof}
Let $A$ be any $c$-competitive deterministic online paging algorithm. The proof is essentially the same as the one for Theorem~\ref{thm:lowerbounddemandpaging} with the number of faults of $A$ on $\sigma'_A(\delta)$ being at most $ck + \beta$, for some constant $\beta$. This follows from the competitiveness of $A$ and the fact that $OPT$ makes at most $k$ faults on $\sigma'_A(\delta)$.
	Assuming for a contradiction that $A$ is $(1, \delta(k+1-\epsilon))$-smooth for some $\epsilon$, we get
\begin{align*}
	A(\sigma_A(\delta)) & \leq 1\cdot A(\sigma'_A(\delta)) + \Delta(\sigma'_A(\delta), \sigma_A(\delta))\cdot(k+1-\epsilon)\\
	\Rightarrow k+\delta(k+1) & \leq ck + \beta + \delta\cdot(k+1-\epsilon)\\
	\Rightarrow \delta\epsilon & \leq (c-1)k+\beta
\end{align*}
For any $\epsilon > 0$, there is a $\delta_m$ such that any $\delta > \delta_m$ contradicts the above inequality.
By a slight generalization of Corollary~\ref{cor:1delta} this implies that the algorithm is not $(1,\delta(k+1-\epsilon),\delta)$-smooth for any $\delta$ and $\epsilon > 0$.
\qed
\end{proof}}{}
\end{FULL}
\begin{FORPROOFSONLYEXCLUDETHIS}
By contraposition of Corollary~\ref{cor:1delta}, the two previous theorems show that no deterministic, demand-paging or competitive algorithm is $(1, k+1-\epsilon, 1)$-smooth for any $\epsilon > 0$.
%\pagebreak

%\begin{SHORT}
%The proof is essentially the same as the one for Theorem~\ref{thm:lowerbounddemandpaging}. 
%We apply the competitiveness to obtain an upper bound on the number of misses for $\sigma_A'(\delta)$. The proof of this corollary and all other missing proofs can be found in the appendix.
%\end{SHORT}

Intuitively, the optimal offline algorithm should be very smooth, and this is indeed the case as we show next:
\end{FORPROOFSONLYEXCLUDETHIS}
\begin{restatable}[Smoothness of OPT]{theorem}{thmoptsmoothness}
\label{thm:optsmoothness}
\opt is $(1,2\delta)$-smooth. This is tight.%
\end{restatable}
\begin{FULL}
\ifthenelse{\boolean{proof-thmoptsmoothness}}{\begin{proof}
For the lower bound consider the following two sequences $\sigma_\delta$ and $\sigma'_\delta$ with $\Delta(\sigma_\delta,\sigma'_\delta)~=~\delta$: $\sigma_\delta = (1, \dots, k, 1, \dots, k)^{\delta+1}$ and $\sigma_\delta' = (1, \dots, k, 1, \dots, k, x)^\delta(1, \dots, k, 1, \dots, k)$, where $\sigma^l$ denotes the concatenation of $l$ copies of $\sigma$ and $x \not\in \{1, \dots, k\}$.
Clearly, $\opt(\sigma_\delta) = k$ as $\sigma_\delta$ contains only $k$ distinct pages.
Further, under optimal replacement every request to $x$ faults in $\sigma'_\delta$ and it replaces one of the pages $1, \dots, k$, which results in an additional fault later on.
So, $\opt(\sigma'_\delta) \geq k + 2\delta$.

We show that $\opt$ is $(1,2,1)$-smooth, which implies the theorem by Corollary~\ref{cor:1delta}. Let $\sigma$ and $\sigma'$ be two sequences such that $\Delta(\sigma,\sigma')=1$.
We will show that there exists an algorithm $A$ such that $A(\sigma')\le OPT(\sigma)+2$, from which the theorem follows since $\opt(\sigma')\le A(\sigma')$.
Let $A$ be an offline paging algorithm serving $\sigma'$. On the equal prefix of $\sigma'$ and $\sigma$, A will act exactly as $\opt$ does on $\sigma$. This implies that right before the difference the caches $C_{\opt}$ of $\opt$ and $C_{A}$ of $A$ have the same pages. No matter what the difference between the sequences is, the different request can only make the caches of $A$ and $\opt$ differ by at most one page: if it is an insertion of $p$ in $\sigma'$ that results in a hit for $A$ then $C_{\opt}=C_{A}$, otherwise $p$ evicts a page $q$ and fetches $p$. If the difference is a deletion of $p$ from $\sigma'$ then if $p$ is a hit $C_{\opt}=C_{A}$, and otherwise $\opt$ evicts some page $q$ and fetches $p$. If the difference is a substitution of $p$ in $\sigma$ by $r$ in $\sigma'$, then if either of the requests is a hit this is equivalent to the cases above, and if they are both misses, $A$ evicts what $\opt$ evicts. In all cases, after the difference either  $C_{A}=C_{\opt}$, or $C_{A}=(C_{\opt}\setminus \{q\}) \cup \{p\}$ for some  pages $p$ and $q$, with $p\ne q$. 
At this point the number of faults of $A$ exceeds those of $\opt$ by at most one. We now show that $A$ can manage to incur at most one more fault than $\opt$ in the rest of the sequence.
 
 Let $\rho$ be the suffix of $\sigma$ and $\sigma'$ after the difference and let $A_i(\rho)$ and $\opt_i(\rho)$ denote the number of faults of the algorithms on the suffix up to request $i$.
We now claim that $A$ can be such that after every request $\rho_i$ either (1) $A_i(\rho)\le \opt_i(\rho)$ and either $C_A=C_\opt$ or $C_{A}=(C_{\opt}\setminus \{q\}) \cup \{p\}$, for some pages $p$ and $q$ with $p\ne q$, or (2) $A_i(\rho)= \opt_i(\rho)+1$ and $C_{A}=C_{\opt}$. If at any point $C_{A}=C_{\opt}$, then $A$ acts like $\opt$ for the rest of the suffix and the claim is true.

We show that this invariant holds after every request $\rho_i$, which implies that $A(\rho)\le \opt(\rho)+1$ and hence that $A(\sigma')\le \opt(\sigma)+2$.  

Initially $A_i(\rho)= \opt_i(\rho)$ and assume that $C_{A}=(C_{\opt}\setminus \{q\}) \cup \{p\}$. $A$ acts as follows on $\rho_i$:
\begin{itemize}
\item If $\rho_i\in C_A$ and $\rho_i\in C_{\opt}$, $A$ does nothing and the invariant holds.
\item If $\rho_i\in C_A$ and $\rho_i\notin C_{\opt}$, then $\rho_i=p$ and $A$ does nothing. It holds that $A_i(\rho)\le \opt_i(\rho)$. If $\opt$ evicts a page $q'\ne q$, then $C_{A}=(C_{\opt}\setminus \{q\}) \cup \{q'\}$ and the invariant holds. If $\opt$ evicts $q$ then both caches are equal and the claim is true.
\item If $\rho_i\notin C_A$ and $\rho_i\notin C_{\opt}$, $A_i(\rho)\le \opt_i(\rho)$ still holds. If $\opt$ evicts $q$, $A$ evicts $p$, and the caches are equal. Otherwise $A$ evicts the same page as $\opt$ and $C_{A}=(C_{\opt}\setminus \{q\}) \cup \{p\}$.
\item If $\rho_i\notin C_A$ and $\rho_i\in C_{\opt}$, then $\rho_i=q$. $A_i(\rho)=\opt_i(\rho)+1$, $A$ evicts $p$ and $C_{A}=C_{\opt}$. \qed
\end{itemize}
\end{proof}}{}
\end{FULL}
\begin{FORPROOFSONLYEXCLUDETHIS}
With Theorem~\ref{thm:optsmoothness} it is easy to show the following upper bound on the smoothness of any competitive algorithm:
\end{FORPROOFSONLYEXCLUDETHIS}
\begin{restatable}[Smoothness of competitive algorithms]{theorem}{thmcompetitivesmoothness}
\label{thm:competitivesmoothness}
Let $A$ be any paging algorithm such that for all sequences $\sigma$, $A(\sigma) \le c\cdot \opt(\sigma) +\beta$. Then $A$ is $(c,2\delta c+\beta)$-smooth. 
\end{restatable}
\begin{FULL}
\ifthenelse{\boolean{proof-thmcompetitivesmoothness}}{\begin{proof}
Let $\sigma'$ be a sequence such that $\Delta(\sigma,\sigma')=\delta$. By Theorem~\ref{thm:optsmoothness}, $\opt(\sigma')\le \opt(\sigma)+2\delta$. Therefore, $A(\sigma') \le c\cdot (\opt(\sigma) +2\delta) +\beta \le  c\cdot A(\sigma) +2\delta c +\beta$.\qed
\end{proof}}{}
\end{FULL}

\begin{FORPROOFSONLYEXCLUDETHIS}
Note that the above theorem applies to both deterministic and randomized algorithms.
Given that every competitive algorithm is $(\alpha, \beta)$-smooth for some $\alpha$ and $\beta$, the natural question to ask is whether the converse also holds.
Below, we answer this question in the affirmative for deterministic bounded-memory, demand-paging algorithms.
By \emph{bounded memory} we mean algorithms that, in addition to the contents of their fast memory, only have a finite amount of additional state.
For a more formal definition consult \cite[page 93]{borodin98}.
Paging algorithms implemented in hardware caches are bounded memory.
\end{FORPROOFSONLYEXCLUDETHIS}
\begin{FULL}
\ifthenelse{\boolean{proof-thmboundedmemorysmooth}}{\begin{FORPROOFSONLYEXCLUDETHIS}
Our proof requires the notion of a $k$-phase partition:
\end{FORPROOFSONLYEXCLUDETHIS}
\begin{definition}[k-phase partition]
\label{def:phasepartition}
The $k$-phase partition of a sequence $\sigma$ is a partition of $\sigma$ into contiguous subsequences called $k$-phases, or simply phases, such that the first phase starts with the first request of $\sigma$ and a new phase starts when $(k+1)$ distinct pages have been requested since the beginning of the previous phase.
\end{definition}}{}
\end{FULL}

\begin{restatable}[Competitiveness of smooth algorithms]{theorem}{thmboundedmemorysmooth}
\label{thm:boundedmemorysmooth}
If algorithm $A$ is deterministic bounded-memory, demand-paging, and $(\alpha, \beta)$-smooth for some $\alpha$~and~$\beta$, then $A$ is also competitive.
\end{restatable}
\begin{FULL}
\ifthenelse{\boolean{proof-thmboundedmemorysmooth}}{%% proof of thmboundedmemorysmooth
\begin{proof}
Assume algorithm~$A$ is non-competitive.
Then, there is no bound on the number of misses in a single $k$-phase for $A$:  otherwise, if $r$ is a bound on the number of misses of $A$ in every phase, then $A$ is competitive with competitive ratio $c \leq r$.

Let $n$ be the number of states of $A$.
Within a $k$-phase, a demand-paging algorithm can reach at most $(2k)^k$ different configurations: each of the $k$ slots can either contain one of the $k$ ``old'' pages cached at the start of the $k$-phase, or one of the up to $k$ ``new'' pages requested within the phase. \todo{an exact bound assumes no other information held by the cache, but we say that there might be some more finite amount of additional state}
Let $\sigma$ be a sequence of minimal length that ends on a phase in which $A$ misses more than $n\cdot (2k)^k$ times.
 By the pigeon-hole principle, $A$ must assume the same state and configuration pair twice within that phase.
 Due to the minimality of the sequence, $A$ must fault at least once between those two occurrences. 
 By repeating the sequence of requests between the two occurrences, we can thus pump up the sequence and the number of faults arbitrarily without increasing the number of phases.
By removing the finite prefix of the sequence that comprises all but the final phase, we can construct a sequence $\sigma'$ containing at most $k$ distinct pages.
Any demand-paging algorithm, in particular $A$, will fault at most $k$ times on this sequence. 
The edit distance between $\sigma'$ and $\sigma$ is finite, but the difference in faults is unbounded.
This shows that $A$ is not $(\alpha, \beta)$-smooth for any $\alpha$ and $\beta$.\looseness=-1\qed
\end{proof}}{}
\end{FULL}
\todo{We conjecture that this also holds without the \emph{bounded memory} assumption.}

\begin{FORPROOFSONLYEXCLUDETHIS}
\subsection{Smoothness of Particular Deterministic Algorithms}
\label{sec:smoothnessdeterministic}

Now let us turn to the analysis of three well-known deterministic algorithms: \LRU, \FWF, and~\FIFO. We show that both \LRU and \FWF are smooth. On the other hand, \FIFO is not smooth, as a single change in the request sequence may increase the number of misses by a factor of $k$.
\end{FORPROOFSONLYEXCLUDETHIS}

\begin{restatable}[Smoothness of Least-recently-used]{theorem}{thmlrusmoothness}
\label{thm:lrusmoothness}~\\
\indent$\LRU$ is $(1,\delta(k+1))$-smooth. This is tight.
\end{restatable}
\begin{FULL}
\ifthenelse{\boolean{proof-thmlrusmoothness}}{\begin{proof}
We show that $\LRU$ is $(1,k+1,1)$-smooth. 
Corollary~\ref{cor:1delta} then immediately implies that $\LRU$ is $(1,\delta(k+1))$-smooth. 
Tightness follows from Theorem~\ref{thm:lowerbounddemandpaging} as $\LRU$ is demand paging.
%% proof of thmlrusmoothness
\todo{write this down more elegantly}
To analyze $\LRU$, it is convenient to introduce the notion of \emph{age}. 
The age of page $p$ is the number of distinct pages that have been requested since the previous request to $p$. 
Before their first request, all pages have age $\infty$.
A request to page $p$ results in a fault if and only if $p$'s age is greater than or equal to $k$, the size of the cache.
Finite ages are unique, i.e., no two pages have the same age less than $\infty$.
At any time at most $k$ pages are cached, and at most $k$ pages have an age less than $k$.

Let us now consider how the insertion of one request may affect ages and the expected number of faults.
By definition, the age of any page is only affected from the point of insertion up to its next request.
Only the next request to a page may thus turn from a hit into a miss. 
At any time at most $k$ pages have an age less than $k$. 
So at most $k$ requests may turn from hits into misses.
As the inserted request itself may also introduce a fault, the overall number of faults may thus increase by at most $k+1$.

Substitutions are similar to insertions: they turn at most $k$ succeeding hits into misses, and the substituted request itself may introduce one additional fault.
The deletion of a request to page $p$ does not increase the ages of other pages.
Only the next request to $p$ may turn from a hit into a miss.\looseness=-1 \qed
\end{proof}}{}
\end{FULL}

\begin{FORPROOFSONLYEXCLUDETHIS}
So \LRU matches the lower bound for both demand-paging and competitive paging algorithms. We now show that \FWF is also smooth, with a factor that is almost twice that of~\LRU. The smoothness of \FWF follows from the fact that it always misses $k$ times per phase, and the number of phases can only change marginally when perturbing a sequence\begin{FULL}, as we show in Lemma~\ref{prop:phases}.

For a sequence $\sigma$, let $\Phi(\sigma)$ denote the number of phases in its $k$-phase partition.\end{FULL}\begin{SHORT}.\end{SHORT}\looseness=-1
%\ifthenelse{\boolean{proof-lemsuffix} \AND \boolean{proof-lemphases}}{Propositions~\ref{prop:suffix} and~\ref{prop:phases}}{the appendix}
\end{FORPROOFSONLYEXCLUDETHIS}

\begin{FULL}
\ifthenelse{\boolean{proof-lemsuffix}}{
\begin{restatable}{proposition}{lemsuffix}
\label{prop:suffix}
For a sequence $\sigma$, let $\Phi(\sigma)$ denote the number of phases in its $k$-phase partition.
Let $\sigma$ be a sequence, let $\rho$ be a suffix of $\sigma$, and let $\ell$ and $\ell'$ denote the number of distinct pages in the last phase of $\sigma$ and~$\rho$, respectively. Then $\Phi(\rho)\le \Phi(\sigma)$. Furthermore, if $\Phi(\rho)= \Phi(\sigma)$ then $\ell'\le \ell$. 
\end{restatable}
%% proof of lemsuffix
\begin{proof}
Let $i_j$ and $i_j'$ denote the indices in $\sigma$ of the first request of the $j^{th}$ phase in $\sigma$ and $\rho$, respectively, 
%with $i_{\Phi(\sigma)+1}-1$ and $i_{\Phi(\rho)+1}-1$ being the indices of the last pages of $\sigma$ and $\rho$. 
with $i_j=|\sigma|+1$ for $j>\Phi(\sigma)$ and $i_j'=|\rho'|+1$ for $j>\Phi(\rho')$.
Then, for all $j$ it holds that $i_j\le i_j'$. We prove this by induction on $j$. The case $j=1$ is trivially true as $i_1=1$ and $i_1'\ge 1$ since $\rho$ is a suffix of $\sigma$. Suppose that the hypothesis holds for $1 < j \le n$. It is easy to see that it holds for $j=n+1$: since there are at most $k$ distinct pages between $i_j$ and $i_{j+1}-1$, inclusive, and by the inductive hypothesis $i_j'\ge i_j$, then the request that ends the $j^{th}$ phase in $\rho$ cannot be earlier than $i_{j+1}$, and hence $i_{j+1}'\ge i_{j+1}$. Since this is true for all phases including the last one, then $\Phi(\rho)\le \Phi(\sigma)$.
Now, assume that $\Phi(\rho)=\Phi(\sigma)=m$. Then $i_{m}' < |\rho'|+1$ and by the proof above $i_{m}\le i_{m}'$, which implies that $\ell'\le \ell$.\qed
 \end{proof}}{}

\begin{restatable}{lemma}{lemphases}
\label{prop:phases}
%For a sequence $\sigma$, let $\Phi_k(\sigma)$ denote the number of phases in its $k$-phase partition. 
Let $\sigma$ and $\sigma'$ be two sequences such that $\Delta(\sigma,\sigma')=1$. Then $\Phi(\sigma')\le \Phi(\sigma)+2$. Furthermore, let $\ell$ and $\ell'$ be the number of distinct pages in the last phase of $\sigma$ and $\sigma'$, respectively. If $\Phi(\sigma')= \Phi(\sigma)+2$, then $\ell' \le \ell$.
\end{restatable}
\ifthenelse{\boolean{proof-lemphases}}{%% proof of lemphases
\begin{proof}
%To simplify notation we omit the subindex $k$ in $\Phi_k(\sigma)$, as $k$ is fixed.
%
%Let $i_j$ be the index of the request that marks the first page of the $j^{th}$ phase in $\sigma$, with $i_j=|\sigma|+1$ for $j>\Phi(\sigma)$.
Let $i_j$ and $i_j'$ be the indices of the requests that mark the first page of the $j^{th}$ phase in $\sigma$ and $\sigma'$, respectively, with $i_j=|\sigma|+1$ for $j>\Phi(\sigma)$ and $i_j'=|\sigma'|+1$ for $j>\Phi(\sigma')$. Let $\Phi(\sigma,j)$ denote the number of phases of $\sigma$ starting from the $j^{th}$ phase (with $\Phi(\sigma,j)=0$ if $j>\Phi(\sigma)$). Let $h-1$ be the phase in $\sigma'$ where the difference between both sequences occurs. %Then, $i_{h-1}=i_{h-1}'$. 
For simplicity, assume that if the difference is an insertion (deletion) on $\sigma'_i$, then $i$ refers to an empty page in $\sigma$ ($\sigma'$), i.e., unaffected requests have equal indices in both sequences. If the difference is a deletion, then $i_{h}\le i'_{h}$\todo{check if it's the first page of the phase, I think it's the same} and by Proposition~\ref{prop:suffix}, $\Phi(\sigma',h)\le \Phi(\sigma,h)$, which implies the lemma. If it is a substitution, suppose that $q$ in $\sigma$  is changed to $p$ in $\sigma'$. Then consider $\sigma''$ resulting from the deletion of $q$ from $\sigma$. By the argument above, $\Phi(\sigma'')\le \Phi(\sigma)$. Hence, showing that $\Phi(\sigma')\le \Phi(\sigma'')+2$ implies as well that $\Phi(\sigma')\le \Phi(\sigma)+2$. Since $\sigma'$ is the result of inserting $p$ into $\sigma''$, it suffices to consider the insertion case (we argue later that if $\Phi(\sigma')= \Phi(\sigma'')+2$, then $\ell''<\ell'$ also holds).

Let $p$ be the page that is added to $\sigma$ to make $\sigma'$. We analyze $\Phi(\sigma)$ in terms of $\Phi(\sigma')$. We have the following cases:

\begin{itemize}
\item $[$$p$ is not the first page of phase $h-1$$]$. If $p$ occurs again in the phase then $\Phi(\sigma)=\Phi(\sigma')$. This is also the case if $h-1$ is the last phase of $\sigma'$. Otherwise, $i_{h}' < i_{h} \le i_{h+1}'$ ($i_h'$ cannot be larger than $i_{h+1}$ as in this case phase $h-1$ in $\sigma'$ would include the $k+1$ distinct pages in $\sigma[i_h..i_{h+1}]$). Then by Proposition~\ref{prop:suffix}, $\Phi(\sigma',h+1)\le \Phi(\sigma,h)$, and therefore $\Phi(\sigma')\le \Phi(\sigma)+1$.
\item $[$$p$ is the first page of phase $h-1$$]$. Then $i_{h-1}>i_{h-1}'$. We have two cases:
\begin{itemize}
  \item If $i_{h-1} \le i_h'$ then we have the same case as above but with  $i_{h-1}$ and $i_h'$. Thus, $\Phi(\sigma')\le \Phi(\sigma)+1$. 
	\item If $i_h' < i_{h-1} \le i_{h+1}'$ (again, $i_{h-1}$ cannot be greater than $i'_{h+1}$ as in this case the $(h-2)^{th}$ phase of $\sigma$ would include all $k+1$ distinct pages in $\sigma'[i'_h..i'_{h+1}]$), then by Proposition~\ref{prop:suffix}, $\Phi(\sigma',h+1)\le \Phi(\sigma,h-1)$. If $\Phi(\sigma',h+1)= \Phi(\sigma,h-1)$, then $\ell'\le \ell$ and $\Phi(\sigma')= \Phi(\sigma)+2$. Otherwise, $\Phi(\sigma',h+1)< \Phi(\sigma,h-1)$ and $\Phi(\sigma')\le \Phi(\sigma)+1$.
\end{itemize}
\end{itemize}

In all cases above either $\Phi(\sigma')\le \Phi(\sigma)+1$ or $\Phi(\sigma')= \Phi(\sigma)+2$ with $\ell'\le \ell$. 
%Note that $\Phi(\sigma')= \Phi(\sigma)+2$ only happens if $p$ is the first page of the phase and $i_h' < i_{h-1} \le i_{h+1}'$.  
If the difference is a substitution, let $q$ be the page in $\sigma'$ that replaces $p$ in $\sigma'$. Note that the case $\Phi(\sigma')= \Phi(\sigma)+2$ can only happen if $q$ is requested earlier in the same phase in $\sigma$. Then, removing the request to $q$ would not change the $k$-phase partition of $\sigma$, and hence the same analysis above for an insertion applies and thus $\ell'\le \ell$ as well.\qed
\end{proof}}{}
\end{FULL}

\begin{restatable}[Smoothness of Flush-when-full]{theorem}{thmfwfsmoothness}~\\
\indent$\FWF$ is $(1,2\delta k)$-smooth. This is tight.
\end{restatable}

\begin{FULL}
\ifthenelse{\boolean{proof-thmfwfsmoothness}}{%% proof of thmfwfsmoothness
\begin{proof}
Let $\sigma$ and $\sigma'$ be two sequences such that $\Delta(\sigma,\sigma')=1$. Let $\Phi(\sigma)$ (resp. $\Phi(\sigma')$) be the number of phases in the $k$-phase partition of $\sigma$ (resp. $\sigma'$), and let $\ell$ (resp. $\ell'$) be the number distinct pages in the last phase of the partition of $\sigma$ (resp. $\sigma'$). 
FWF misses exactly $k$ times in any phase of a sequence, except possibly for the last one, in which it misses a number of times equal to the number of distinct pages in the phase. Then, $\FWF(\sigma)=k \cdot (\Phi(\sigma)-1) + \ell$, and  $\FWF(\sigma')=k \cdot (\Phi(\sigma')-1) + \ell'$.
By Lemma~\ref{prop:phases}, if $\Phi(\sigma')=\Phi(\sigma)+2$, $\ell'\le \ell$, and thus $\FWF(\sigma')\le k(\Phi(\sigma)+2-1) + \ell=\FWF(\sigma)+2k$. Otherwise, if $\Phi(\sigma')\le \Phi(\sigma)+1$, then $\FWF(\sigma')\le k(\Phi(\sigma)+1-1) + \ell'=\FWF(\sigma)-\ell+k+\ell'$. Since $\ell\ge 0$ and $\ell'\le k$, $\FWF(\sigma')\le \FWF(\sigma)+2k$. The upper bound in the lemma follows by Corollary~\ref{cor:1delta}.
To see that this upper bound is tight, let $\sigma = (x_1,\ldots,x_k)^{2\delta+1}$, where $x_i\ne x_j$ for all $i\ne j$. Let $\sigma' = x_1,\ldots,x_k(x_{k+1},x_1,\ldots,x_k,x_1,\ldots,x_k)^\delta$, where $x_{k+1}\ne x_i$ for all $i\le k$. Thus, $\Delta(\sigma,\sigma')=\delta$. Clearly $\FWF(\sigma)=k$, while $\FWF(\sigma')=k+2\delta k$, and hence $\FWF(\sigma')=\FWF(\sigma)+2\delta k$.\qed
\end{proof}}{}
\end{FULL}

\begin{FORPROOFSONLYEXCLUDETHIS}
We now show that \FIFO is not smooth. 
In fact, we show that with only a single difference in the sequences, the number of misses of \FIFO can be $k$ times higher than the number of misses in the original sequence. On the other hand, since \FIFO is strongly competitive, the multiplicative factor $k$ is also an upper bound for \FIFO's smoothness.
\end{FORPROOFSONLYEXCLUDETHIS}

\begin{restatable}[Smoothness of First-in first-out]{theorem}{thmfifosmoothness}~\\
\FIFO is $(k, 2\delta k)$-smooth. \FIFO is not $(k-\epsilon,\gamma,1)$-smooth for any $\epsilon > 0$ and $\gamma$.
\end{restatable}
\begin{FULL}
\ifthenelse{\boolean{proof-thmfifosmoothness}}{%% proof of thmfifosmoothness
\begin{proof}
The upper bound follows from the competitiveness of \FIFO and Theorem~\ref{thm:competitivesmoothness}.

For the lower bound, we show how to construct two sequences $\sigma_k$ and $\sigma'_k$ for each cache size $k$, such that $\Delta(\sigma'_k, \sigma_k) = 1$, that yield\todo{clarify terms configuration, state, etc. currently the following use of the term configurations conflicts with its definition} configurations $c = [1, \dots, k]$ and $c' = [k, \dots, 1]$, where pages are sorted from last-in to first-in from left to right.
Then, the sequence $0, 1, 2, \dots, k-1$ yields $k$ misses starting from configuration $c'$ and only one miss starting from $c$.
The resulting configurations are $[0, \dots, k-1]$ and $[k-1, \dots, 0]$, which are equal to $c$ and $c'$ up to renaming.
So we can construct an arbitrarily long sequence that yields $k$ times as many misses starting from configuration $c'$ as it does from configuration $c$.

For $k=2$, $\sigma_2 = 2, 1$ and $\sigma'_2 = 1,2,1 = 1 \circ \sigma_2$ have edit distance $\Delta(\sigma'_2, \sigma_2) = 1$ and yield configurations $[1, 2]$ and $[2, 1]$, respectively.
For $k=3$, $\sigma_3 = 2, 3, 1, 4, 2, 1, 5, 1, 4$ and $\sigma'_3 = 1 \circ \sigma_3$ yield configurations $[4, 1, 5]$ and $[5,1,4]$, respectively, which are equal up to renaming to $[1,2,3]$ and $[3,2,1]$.

For $k > 3$, we present a recursive construction of $\sigma_k$ and $\sigma'_k$ based on $\sigma_{k-1}$ and $\sigma'_{k-1}$.
Notice, that $\sigma'_2 = 1 \circ \sigma_2$ and $\sigma'_3 = 1 \circ \sigma_3$.
We will maintain that $\sigma'_k = 1 \circ \sigma_k$ in the recursive construction.

As $\sigma_{k-1}$ and $\sigma'_{k-1}$ are constructed for a cache of size $k-1$, they will behave differently on a larger cache of size $k$.
However, we can pad $\sigma_{k-1}$ and $\sigma'_{k-1}$ with requests to one additional page $x$ that fills up the additional space in the cache.
This can be achieved as follows:
Add a request to $x$ at the start of the two sequences (following the request to $1$ in $\sigma'_k$).
Also, whenever $x$ is evicted in either of the two sequences, in a cache of size $k$, add a request to $x$ in both sequences.
By construction, the additional requests do not increase the edit distance between the two sequences.
Further, the additional requests ensure that every request that belongs to the original sequences faults in the new sequence on a cache of size $k$ if and only if it faults in the original sequence on a cache of size $k-1$.
In this way, we obtain $\sigma_{k, pre}$ and $\sigma_{k, pre}'$.
The two sequences yield configurations $c = [1, \dots, i', x, i'+1, \dots, k-1]$ and $c' = [k-1, \dots, j'+1, x, j', \dots 1]$, respectively, which, unless $i'=j'$, are almost solutions to the original problem.

Observe that $c = [1, \dots, i', x, i'+1, \dots, k-1]$ and $c' = [k-1, \dots, j'+1, x, j', \dots 1]$ are equal up to renaming to $d = [1, \dots, k]$ and $d' = [k	, \dots, j+1, i, j \dots, i+1, i-1, \dots 1]$ for some $i, j$ with $1 \leq i < j \leq k$.
We distinguish five cases depending on the values of $i$ and $j$:

Case 1: $1 < i < j < k$. Below we build a suffix $\sigma_{k, post}$ that finishes the construction:
\[\small\begin{array}{c|cc}
	\parbox[c]{2.5cm}{Requests}	& \parbox[c]{3.5cm}{State for prefix $\sigma'_{k, pre}$}	& \parbox[c]{3.5cm}{State for prefix $\sigma_{k, pre}$}\\\hline
   							 					& [1, \dots, k] 					& [k, \dots, j+1,i,j,\dots, i+1,i-1,\dots, 1]\\
	\xrightarrow{v} 							& [v,1, \dots, k-1] 				& [v,k ,\dots ,j+1,i,j, \dots, i+1, i-1, \dots, 2]\\ %assuming i >1
	\xrightarrow{k, \dots, i+1} 				& [i+1,\dots, k,v,1, \dots, i-1]  & [v,k, \dots, j+1,i,j,\dots, i+1, i-1, \dots, 2]\\
	\xrightarrow{1, \dots, i-1} 				& [i+1,\dots, k,v,1, \dots, i-1]	& [i-1,\dots, 1,v,k, \dots, j+1,i,j,\dots, i+2]\\
	\xrightarrow{w} 							& [w,i+1,\dots, k,v,1, \dots, i-2]	& [w,i-1,\dots, 1,v,k, \dots, j+1,i,j,\dots, i+3]\\
	\xrightarrow{i+1, \dots, k,v,1, \dots, i-2} 	& [w,i+1,\dots, k,v,1, \dots, i-2]	& [i-2,\dots, 1,v,k, \dots, i+1,w]
\end{array}\]
\normalsize

Case 2: $1 = i < j < k-1$. Consider the following suffix:
\[\small\begin{array}{c|cc}
	\parbox[c]{2.5cm}{Requests}	& \parbox[c]{3.5cm}{State for prefix $\sigma'_{k, pre}$}	& \parbox[c]{3.5cm}{State for prefix $\sigma_{k, pre}$}\\\hline
   							 					& [1, \dots, k] 					& [k, \dots, j+1,1,j, \dots, 2]\\
	\xrightarrow{y} 									& [y,1, \dots, k-1] 				& [y,k ,\dots ,j+1,1,j, \dots, 3]\\ 
	\xrightarrow{2, \dots, j} 							& [y,1, \dots, k-1] 				& [j, \dots, 2, y,k ,\dots, j+1]\\
	\xrightarrow{1} 									& [y,1, \dots, k-1] 				& [1, j, \dots, 2, y, k, \dots, j+2]\\
	\xrightarrow{j+1, \dots, k-1} 						& [y,1, \dots, k-1] 				& [k-1, \dots, j+1, 1, j,\dots, 2,y]
\end{array}\]
\normalsize
The final pair of states is equal up to renaming to the pair $d = [1, \dots, k]$ and $d' = [k, \dots, j+2, 2, j+1, \dots, 3, 1]$, and so it fulfills the conditions under which the suffix $\sigma_{k, post}$ constructed in Case 1 finishes the construction.

Case 3: $1 = i < j = k-1$. Consider the following suffix:
\[\small\begin{array}{c|cc}
	\parbox[c]{2.5cm}{Requests}	& \parbox[c]{3.5cm}{State for prefix $\sigma'_{k, pre}$}	& \parbox[c]{3.5cm}{State for prefix $\sigma_{k, pre}$}\\\hline
   							 					& [1, \dots, k] 					& [k, 1, k-1, \dots, 2]\\
	\xrightarrow{x} 									& [x,1, \dots, k-1] 				& [x, k, 1, k-1, \dots, 3]\\ 
	\xrightarrow{2, \dots, k-1} 							& [x,1, \dots, k-1] 				& [k-1, \dots, 2, x, k]\\
	\xrightarrow{y, z} 								& [z, y, x, 1, \dots, k-3] 				& [z, y, k-1, \dots, 2]\\
	\xrightarrow{x, 1, \dots, k-3} 						& [z, y, x, 1, \dots, k-3] 				& [k-3, \dots, 1, x, z, y]
\end{array}\]
\normalsize
The final pair of states is equal up to renaming to the pair $d = [1, \dots, k]$ and $d' = [k, \dots, 3, 1, 2]$, and so it fulfills the conditions under which Case 2 continues the construction.

Case 4: $1 < i < j = k$. Exchanging $\sigma_{k, pre}$ and $\sigma_{k, pre}'$ yields states that fulfill the conditions of either Case 2 or Case 3.

Case 5: $1 = i < j = k$. Consider the following suffix:
\[\small\begin{array}{c|cc}
	\parbox[c]{2.5cm}{Requests}	& \parbox[c]{3.5cm}{State for prefix $\sigma'_{k, pre}$}	& \parbox[c]{3.5cm}{State for prefix $\sigma_{k, pre}$}\\\hline
   							 					& [1, \dots, k] 					& [1, k, \dots, 2]\\
	\xrightarrow{x} 									& [x,1, \dots, k-1] 				& [x,1, k ,\dots, 3]\\ 
	\xrightarrow{2, \dots, k-1} 							& [x,1, \dots, k-1] 				& [k-1, \dots, 2, x, 1]
\end{array}\]
\normalsize
The resulting pair of states is equal up to renaming to $[1, \dots, k]$ and $[k, \dots, 3, 1, 2]$, which corresponds to Case 2.\qed

\end{proof}}{}
\end{FULL}

\begin{FORPROOFSONLYEXCLUDETHIS}
\FIFO matches the upper bound for strongly-competitive deterministic paging algorithms.
With the result for \LRU, this demonstrates that the upper and lower bounds for the smoothness of strongly-competitive algorithms are tight.

\begin{SHORT}\vspace{-2mm}\end{SHORT}
\section{Smoothness of Randomized Paging Algorithms}
\begin{SHORT}\vspace{-2mm}\end{SHORT}

\subsection{Bounds on the Smoothness of Randomized Paging Algorithms}
\nvsp

Similarly to deterministic algorithms, we can show a lower bound on the smoothness of any randomized demand-paging algorithm. 
%Here $H_k = \sum_{i=1}^k \frac 1i$ denotes the $k^{th}$ harmonic number.
\begin{SHORT}
The proof is strongly inspired by the proof of a lower bound for the competitiveness of randomized algorithms by Fiat~\etal~\cite{Fiat91}.
The high-level idea is to construct a sequence using $k+1$ distinct pages on which any randomized algorithm faults at least $k+H_k + \frac{1}{k}$ times that can be converted into a sequence containing only $k$ distinct pages by deleting a single request.\end{SHORT} \todo{this should be included if not including the proof}
Notice that the lower bound only applies to $\delta=1$ and so additional disturbances might have a smaller effect than the first one.

\nvsp
\end{FORPROOFSONLYEXCLUDETHIS}
\begin{restatable}[Lower bound for randomized, demand-paging algorithms]{theorem}{thmlowerboundrandomizeddemangpaging}\label{thm:lowerboundrandomizeddemandpaging}No randomized, demand-paging algorithm is $(1, H_k{+}\frac{1}{k}{-}\epsilon, 1)$-smooth for any $\epsilon{>}0$.\looseness=-1
\end{restatable}

\begin{FULL}
\ifthenelse{\boolean{proof-thmlowerboundrandomizeddemangpaging}}{%% proof of thmlowerboundrandomizeddemangpaging
\begin{proof}
For a given randomized, demand-paging algorithm $A$, we show how an oblivious adversary can construct two sequences, a ``bad'' sequence $\sigma'_A$ and a ``good'' sequence $\sigma_A$, with edit distance~1, such that $A(\sigma')$ is at least $k+H_k+\frac{1}{k}$ and $A(\sigma)$ is exactly $k$.
The existence of such sequences immediately implies the theorem.
The construction is inspired by the nemesis sequence devised by Fiat~\etal~\cite{Fiat91}%
\begin{FORPROOFSONLYINCLUDETHIS}%
\footnote{All references cited in the appendix appear at the end of the document.}
\end{FORPROOFSONLYINCLUDETHIS}
 in their proof of a lower bound for the competitiveness of randomized algorithms.
 
The sequence $\sigma'_A$ consists of requests to $n = k+1$ distinct pages.
During the construction of the sequence, the adversary maintains for each of the $n$ pages its probability $p_i$ of \emph{not} being in the cache.
This is possible, because the adversary knows the probability distribution used by $A$. 
We have $\sum_i p_i \geq 1$, as only $k$ of the $n=k+1$ pages can be in the fast memory.

The ``bad'' sequence $\sigma'_A$ begins by $n$ requests, a single request to each of the $n$ pages in an arbitrary order.
Initially, the fast memory is empty, and so these requests will result in $k+1$ faults.
After those requests, as $p_n = 0$, there will be at least one page $i$ with $p_i \geq \frac{1}{k}$.
The next request in $\sigma'_A$ is to such a page. 
We will later refer to this page as $m$.
The remainder of $\sigma'_A$ is composed of $k-1$ subphases, the $i^{th}$ subphase of which will contribute an expected $\frac{1}{k-i+1}$ page faults.
By linearity of expectation, we can sum up the expected faults on the entire sequence, and obtain 
$A(\sigma') \geq k+1 + \frac{1}{k} + \sum_{i=1}^{k-1} \frac{1}{k-i+1} = k+1+\frac{1}{k} + H_k-1 = k + H_k + \frac{1}{k}$.
It remains to show how to construct the remaining $k-1$ subphases and the ``good'' sequence~$\sigma_A$.

Each of the $k-1$ subphases consists of zero or more requests to \emph{marked} pages followed by exactly one request to an \emph{unmarked} page.
A page is \emph{marked} at the start of subphase $i$ if it is page $m$ or if it has been requested in at least one of the preceding subphases $1 \leq j < i$.
Let $M$ be the set of marked pages at the start of the $j^{th}$ subphase.
Then the number of marked pages is $|M| = j$ and the number of unmarked pages is $u = k+1-j$.
Let $p_M = \sum_{i \in M} p_i$.
If $p_M = 0$, then there must be an unmarked page $n$ with $p_n \geq \frac{1}{u}$ and the adversary can pick this page to end the subphase.
Otherwise, if $p_M > 0$ there must be a marked page $l$ with $p_l > 0$.
The first request of subphase $j$ is to page $l$.
Let $\epsilon = p_l$. 
The adversary can now generate requests to marked pages using the following loop:\looseness=-1

While the expected number of faults in subphase $i$ is less than $\frac{1}{u}$, and while $p_M > \epsilon$, request page $l$ such that $l = \argmax_{i \in M} p_i$.

Note that the loop must terminate, as each iteration will contribute $p_l \geq \frac{p_M}{|M|} > \frac{\epsilon}{|M|}$ expected faults.
If the loop terminates due to the first condition, the adversary can request an arbitrary unmarked page to end the subphase.
Otherwise, the adversary requests the unmarked page $i$ with the highest probability values. 
Clearly, $p_i \geq \frac{1-p_M}{u} > \frac{1-\epsilon}{u}$.
The total expected number of faults of the subphase is then $\epsilon + p_i > \frac{1}{u}$.
This concludes the construction of $\sigma'_A$.

Notice that there is one unmarked page that has only been requested in the initial $n$ requests of~$\sigma'_A$.
We obtain the ``good'' sequence $\sigma_A$ be deleting the request to this unmarked page from~$\sigma'_A$.
By construction, $\sigma_A$ contains requests to only $k$ distinct pages.
As $A$ is by assumption demand paging, $\sigma_A$ will thus incur $k$ page faults only.\qed
\end{proof}}{}
\end{FULL}

\nvsp
\begin{FORPROOFSONLYEXCLUDETHIS}
For strongly-competitive randomized algorithms we can show a similar statement using a similar yet more complex construction:
\end{FORPROOFSONLYEXCLUDETHIS}

\begin{restatable}[Lower bound for strongly-competitive randomized paging algorithms]{theorem}{thmlowerboundstronglycompetitive}
\label{thm:lowerboundrandomizedstronglycompetitive}
No strongly-competitive, randomized paging algorithm is $(1, \delta (H_k-\epsilon))$-smooth for any $\epsilon > 0$.
\end{restatable}
\nvsp
\begin{FULL}
\ifthenelse{\boolean{proof-thmlowerboundstronglycompetitive}}{%% proof of thmlowerboundstronglycompetitive
\begin{proof}
For any $H_k$-competitive algorithm $A$ and any $\epsilon > 0$, we can construct two sequences $\sigma_A'$ and $\sigma_A$, such that $A(\sigma_A')-A(\sigma_A) > \Delta(\sigma_A', \sigma_A)\cdot (H_k-\epsilon)$, which proves the theorem.

Fiat~\etal~\cite{Fiat91} show how to construct a ``bad'' sequence consisting of an arbitrary number of $k$-phases, in each of which any randomized algorithm incurs at least $H_k$ misses, while the optimal offline algorithm incurs only one miss.
To follow this proof, it is helpful to be familiar with the proof of Theorem~4 in \cite{Fiat91}.
We adapt their construction in the following way:
whenever it is possible to incur a cost of $H_k-\frac{\epsilon}{2}$ in a $k$-phase by requesting only \emph{unmarked} pages, we do so.
We call such phases \emph{type I}.
Whenever, this is not possible, we can construct a \emph{type II} phase that results in more than $H_k + \delta$ misses, where $\delta > 0$ depends on $\epsilon$.
Type~II phases may only happen a finite number of times for each occurrence of a type~I phase: otherwise, the sequence would be a counterexample to the $H_k$-competitiveness of~$A$.
Thus, for any $d$ we can construct a sequence $\sigma_A'(d)$ that includes exactly~$d$ type~I phases.

From the resulting ``bad'' sequence $\sigma_A'(d)$ we obtain the ``good'' sequence $\sigma_A(d)$ by deleting one request from each of the type~I phases.
By construction, each page is only requested once in a type~I phase. 
By deleting the one request to the page that is not requested in the following phase, we reduce the number of $k$-phases in the sequence by one.
As a consequence, the resulting ``good'' sequence contains $n-d$ phases, where $n$ is the number of phases of $\sigma_A'(d)$.
Including compulsory misses, the optimal offline algorithm incurs $(k+n-1)-d$ misses.
As~$A$ is strongly competitive, it incurs at most $H_k\cdot (k+n-1-d)$ misses on $\sigma_A$.
On the other hand, by construction, the number of expected misses on $\sigma_A'$ is at least $k+(n-1)\cdot H_k-d\cdot \frac{\epsilon}{2}$, and we get:
\[A(\sigma_A'(d)) - A(\sigma_A(d)) \geq k\cdot(1-H_k) + d\cdot\left(H_k - \frac{\epsilon}{2}\right),\]
which is greater than $d\cdot(H_k - \epsilon)$ for a large enough value of $d$.
As, by construction $d = \Delta(\sigma_A', \sigma_A)$ this proves the theorem.

It remains to show how to construct type~I and type~II phases with the properties discussed above.
Let us first discuss how to adapt the construction of the $i^{th}$ subphase.
We introduce an additional parameter $\epsilon_i$ that controls the reduction in expected faults we are willing to pay for a type~I subphase.

Let $M$ be the set of marked pages at the start of the subphase and $p_M = \sum_{i \in M} p_i$ the probability that a marked page is not cached.
Let $u = k+1-i$ denote the number of unmarked pages.
If there is an unmarked page $j$ with $p_j \geq \frac{1}{u} - \epsilon_i$ the adversary requests this page to end the subphase.
We call such a subphase type~I, analogously to the convention for phases.
Other subphases are called type~II.

Note, that in the first subphase there is always such an unmarked page, as $p_m = 0$ for the only marked $m$ page, as it has just been requested.
Otherwise, $p_M > 0$, and the adversary requests the marked page $l$ with $l = \argmax_{i \in M} p_i$.
Let $\mu = p_l$.
As there are $i$ marked pages and $p_M > \epsilon_i\cdot u$, $\mu$ must be at least $\frac{\epsilon_i\cdot u}{i} = \frac{\epsilon_i\cdot (k+1-i)}{i} \geq \frac{\epsilon_i}{k}$ for $1 < i \leq k$.
The adversary can now generate requests to marked pages using the following loop:

While the expected number of faults in subphase $i$ is less than $\frac{1}{u}+\frac{\epsilon_i}{2k}$, and while $p_M > \frac{\epsilon_i u}{2 k}$, request a marked page $l$ such that $l = \argmax_{i \in M} p_i$.

This loop is guaranteed to terminate, as each iteration adds at least $\frac{\epsilon_i u}{2 k  i}$ to the expected number faults in the subphase.
If the total expected number of faults ends up exceeding $\frac{1}{u}+\frac{\epsilon_i}{2k}$ an arbitrary request is made to an unmarked page.
Otherwise, the page with the highest probability value is requested. 
Its fault probability is at least $\left(1-\frac{\epsilon_i\cdot u}{2k}\right)\cdot\frac{1}{u} = \frac{1}{u} - \frac{\epsilon_i}{2k}$.
Taking into account the initial request to a marked page that contributed at least $\frac{\epsilon_i}{k}$ faults, the subphase has  an expected number of at least $\frac{1}{u} + \frac{\epsilon_i}{2k}$ faults.

For a given $\epsilon$, we can choose the parameters $\epsilon_i$ to the subphases in a way that if all subphases end up as type~I, and thus the phase itself ends up as type~I, the expected number of faults is at least $H_k - \epsilon$:
the sum of the $\epsilon_i$ needs to be less than or equal to $\epsilon$.
Further, by picking $\epsilon_i$ such that $\frac{\epsilon_i}{2k} > \delta + \sum_{1 < j < i} \epsilon_i$ we can make sure that a type~II phase will have an expected number of misses greater than $H_k + \delta$ for some $\delta > 0$.\qed
\end{proof}}{}
\end{FULL}
\nvsp
\begin{FORPROOFSONLYEXCLUDETHIS}
In contrast to the deterministic case, this lower bound only applies to strongly-competitive algorithms, as opposed to simply competitive.
So with randomization there might be a trade-off between competitiveness and smoothness. 
\begin{FULL}
There might be competitive algorithms that are smoother than all strongly-competitive ones.
\end{FULL}

\nvsp
\nvsp
\subsection{Smoothness of Particular Randomized Algorithms}
\nvsp
\nvsp

\begin{FULL}
Two known strongly-competitive randomized paging algorithms are \partition, introduced by McGeoch and Sleator~\cite{McGeoch91} and \equitable, introduced by Achlioptas, Chrobak, and Noga~\cite{Achlioptas00}.
We show that neither of the two algorithms is smooth.
\end{FULL}
\begin{SHORT}
Two known strongly-competitive randomized paging algorithms are \partition~\cite{McGeoch91} and \equitable~\cite{Achlioptas00}.
We show that neither of the two algorithms is smooth.
\end{SHORT}
\nvsp
\end{FORPROOFSONLYEXCLUDETHIS}
\begin{restatable}[Smoothness of \partition and \equitable]{theorem}{thmlowerboundpartitionequitable}
For any cache size $k \geq 2$, there is an $\epsilon > 0$, such that neither \partition nor \equitable is $(1+\epsilon, \gamma, 1)$-smooth for any $\gamma$.
Also, \partition and \equitable are $(H_k,2\delta H_k)$-smooth.
\end{restatable}
\nvsp
\begin{FULL}
\ifthenelse{\boolean{proof-thmlowerboundpartitionequitable}}{%% proof of thmlowerboundpartitionequitable
\begin{proof}
The fact that the two policies are $(H_k,2\delta H_k)$-smooth follows immediately from Theorem~\ref{thm:competitivesmoothness} and the fact that the two policies are $H_k$-competitive with additive constant~$0$.

\newcommand{\hits}[2]{\left<#1,#2\right>}
Koutsoupias and Papadimitriou~\cite{Koutsoupias00} introduced the \emph{layer representation}, a sequence $(L_1, \dots, L_k)$ of $k$ sets of pages that compactly represents the current work function.
For both \partition and \equitable, the probability of being in a particular configuration is determined by the current work function and thus its layer representation.
See Achlioptas~\etal~\cite{Achlioptas00} for more details. 

Let $\omega = (L_1, \dots, L_k)$ and let $p$ be the next page to be requested.
Then the layer representation of the work function is updated as follows:
\[\omega^p = \begin{cases}
	(p, L_1, \dots, L_{j-1}, L_j \cup L_{j+1} - \{p\}, L_{j+2}, \dots, L_k)	& $if $ p \in L_j \wedge j < k\\
	(p, L_1, \dots, L_{k-1})		& $if $ p \in L_k\\
	(p, L_1 \cup L_2, L_3, \dots, L_k) & $if $ p \not\in \bigcup_{i=1}^k L_i 
	\end{cases}\]
In the following, to save space, we will omit braces in the representation of the sets $L_i$, i.e.
$(abc, d, e) = (\{a,b,c\}, \{d\}, \{e\})$.

Starting from an empty cache, and the corresponding empty layer representation, the sequence $\sigma = y, x, k-1, k-2, \dots, 1, 0, k-2, k-3, \dots, 1, 0$ yields the following layers:
\[(0, 1, \dots, k-2, (k-1)xy).\]
The sequence $\sigma' = y, x, k-1, k-2, \dots, 1, 0, k-1, k-2, \dots, 1, 0$, on the other hand, with $\Delta(\sigma', \sigma) = 1$ yields the following layers, in which all pages are revealed:
\[(0, 1, \dots, k-2, k-1).\]

In the following, we show how to extend these two sequences to yield an unbounded difference in the expected number of faults.
An arrow $\xrightarrow[]{p : \hits{h_1}{h_2}}$ indicates a request to page $p$ with a hit probability of $h_1$ under prefix $\sigma'$ and a hit probability of $h_2$ under prefix $\sigma$.

\[\small\begin{array}{c|cc}
	\parbox{2.5cm}{Request and hit probabilities}				& \parbox[c]{3.5cm}{Layers for prefix $\sigma'$}	& \parbox[c]{3.5cm}{Layers for prefix $\sigma$}\\\hline
   							 					& (0, 1, 2, \dots, k-1) 					& (0, 1, 2, \dots, (k-1)xy)\\
	\xrightarrow[]{x : \hits{0}{\frac{1}{3}}} 	& (x, 01, 2, 3, \dots, k-1) 				& (x, 0, 1, \dots,	k-2)\\
	\xrightarrow[]{0 : \hits{\frac{k-1}{k}}{1}} 	& (0, x, 12, 3, \dots, k-1)						& (0, x, 1, 2, \dots, k-2)\\
	\xrightarrow[]{1 : \hits{\frac{k-2}{k-1}}{1}} 	& (1, 0, x, 23, 4, \dots, k-1)				& (1, 0, x, 2, \dots, k-2)\\
		 & \hfill\vdots \\
	\xrightarrow[]{i : \hits{\frac{k-i-1}{k-i}}{1}} 	& (i, i-1, \dots, 1, 0, x, (i+1)(i+2), \dots, k-1)	& (i, i-1, \dots, 1, 0, x, i+1, \dots, k-2)\\
		 & \hfill\vdots \\
	\xrightarrow[]{k-3 : \hits{\frac{2}{3}}{1}} 	& (k-3, k-4, \dots, 1, 0, x, (k-2)(k-1))		& (k-3, k-4, \dots, 1, 0, x, k-2)\\
	\xrightarrow[]{y : \hits{0}{0}} 	& (y, (k-3)(k-4), \dots, 1, 0, x, (k-2)(k-1))		& (y, (k-3)(k-4), \dots, 1, 0, x, k-2)\\
	\xrightarrow[]{k-1 : \hits{\frac{1}{2}\cdot\frac{k-1}{k}}{0}} 	& (k-1, y, (k-3)(k-4), \dots, 1, 0, x)		& (k-1, y(k-3)(k-4), \dots, 1, 0, x, k-2)\\
	\xrightarrow[]{k-3 : \hits{\frac{k-2}{k-1}}{\frac{k-1}{k+1}}} 	& (k-3, k-1, y, (k-4)(k-5), \dots, 1, 0, x)		& (k-3, k-1, y(k-4)(k-5), \dots, 1, 0, x, k-2)\\
		& \hfill\vdots \\
	\xrightarrow[]{i : \hits{\frac{i+1}{i+2}}{\frac{i+2}{i+4}}} 	& (i, \dots, k-3, k-1, y, (i-1)(i-2), \dots, 1, 0, x)	& (i, \dots, k-3, k-1, y(i-1)(i-2), \dots, 1, 0, x, k-2)\\
		& \hfill\vdots \\
	\xrightarrow[]{0 : \hits{\frac{1}{2}}{\frac{2}{4}}} 	& (0, 1, \dots, k-3, k-1, y)	& (0, 1, \dots, k-3, k-1, yx(k-2))
\end{array}\]

Observe that the layers at the end of the above sequence are equal to the layers at the beginning of the sequence up to renaming.
So we can extend the sequence arbitrarily achieving the same hit probabilities on both sides.

It remains to compute the expected number of hits and misses on the sequence above.
Summing up the hit probabilities for prefix $\sigma'$, we get:
\begin{align*}
  & 0 + \sum_{i=0}^{k-3} \frac{k-i-1}{k-i} + 0 + \frac{1}{2}\cdot\frac{k-1}{k} + \sum_{i=0}^{k-3} \frac{i+1}{i+2}\\
=~& \left(k-H_k-\frac{1}{2}\right) + \frac{1}{2}\cdot\frac{k-1}{k} + \left(k-H_{k-1}-1\right)\\
=~& 2k - 2H_k + \frac{1}{k} + \frac{1}{2}\cdot\frac{k-1}{k} - \frac{3}{2}\\
=~& 2k - 2H_k - 1 + \frac{1}{2k}.
\end{align*}
As there are $2k-1$ requests on the sequence, the expected number of misses for prefix $\sigma'$ is 
\[2H_k - \frac{1}{2k}.\]

Summing up the hit probabilities for prefix $\sigma$, we get:
\begin{align*}
  & \frac{1}{3} + \sum_{i=0}^{k-3} 1 + 0 + 0 + \sum_{i=0}^{k-3} \frac{i+2}{i+4}\\
=~& \frac{1}{3} + k-2 + \left(k-2H_{k+1}+\frac{5}{3}\right)\\
=~& 2k - 2H_{k+1},
\end{align*}
which yields an expected number of misses for prefix $\sigma$ of
\[2H_{k+1}-1 = 2H_k - \frac{k-1}{k+1}.\]

The hit probabilities on the sequence above are the same under \equitable as under \partition, except for the request to page~$k-1$.
Under \equitable this request has a hit probability of $\frac{1}{2}-\frac{1}{k(k+1)}$.
So the number of expected misses under \equitable for prefix $\sigma'$ is 
\begin{align*}
  & 2H_k-\frac{1}{2k} + \frac{k-1}{2k} - \left(\frac{1}{2}-\frac{1}{k(k+1)}\right)\\
=~& 2H_k-\frac{1}{k+1}.
\end{align*}

For prefix $\sigma$, \equitable and \partition behave exactly the same.
Observe that in both cases, the number of expected misses for prefix $\sigma$ is lower than for prefix $\sigma'$.
As we can arbitrarily extend the sequence presented above, for any cache size $k \geq 3$, there is an $\epsilon$, such that neither $\equitable$ nor $\partition$ are $(1+\epsilon, \gamma, 1)$-smooth for any $\gamma$.

For a cache size of $2$, consider the sequences $\sigma' = 4,3,2,1,0$ and $\sigma = 3,2,1,0$, with $\Delta(\sigma', \sigma)=1$, which yield layers $(0, 1234)$ and $(0, 123)$, respectively.
The sequence $4, 5, 6, 7, 1, 4, 7, 0, 2$ yields the following layers and hit probabilities under both \equitable and \partition:
\[\begin{array}{c|cc}
	\parbox[c]{5.5cm}{Requests and hit probabilities}	& \parbox[c]{3.5cm}{Layers for prefix $\sigma'$}	& \parbox[c]{3.5cm}{Layers for prefix $\sigma$}\\\hline
   							 							& (0, 1234) 					& (0, 123)\\
	\xrightarrow{4 : \hits{\frac{1}{4}}{0}} 			& (4, 0) 						& (4, 0123)\\
	\xrightarrow{5, 6, 7 : \hits{0}{0}} 				& (7, 0456)						& (7, 0123456)\\
	\xrightarrow{1 : \hits{0}{\frac{1}{7}}} 			& (1, 04567)					& (1, 7)\\
	\xrightarrow{4 : \hits{\frac{1}{5}}{0}} 			& (4, 1)						& (4, 17)\\
	\xrightarrow{7 : \hits{0}{\frac{1}{2}}} 			& (7, 14)						& (7, 4)\\
	\xrightarrow{0, 2 : \hits{0}{0}} 					& (2, 0147)						& (2, 047)
\end{array}\]
The final layers are equal up to renaming to the initial layers, so we can extend the sequence arbitrarily achieving the same hit probabilities on both sides.
As the expected number of hits differs depending on the two prefixes, $\frac{9}{20}$ for $\sigma'$ versus $\frac{9}{14}$ for $\sigma$, this proves the theorem for $k=2$.\qed
\end{proof}}{}
\end{FULL}

\begin{FORPROOFSONLYEXCLUDETHIS}
The lower bound in the theorem above is not tight, but it shows that neither of the two algorithms matches the lower bound from Theorem~\ref{thm:lowerboundrandomizedstronglycompetitive}.
This leaves open the question whether the lower bound from Theorem~\ref{thm:lowerboundrandomizedstronglycompetitive} is tight.

\begin{FULL}
Note that the lower bound for \equitable applies equally to \onlinemin~\cite{Brodal15}, as \onlinemin has the same expected number of faults as \equitable on all request sequences.
\end{FULL}

\MARK~\cite{Fiat91} is a simpler randomized algorithm that is $(2H_k-1)$-competitive.
We show that it is not smooth either.
\end{FORPROOFSONLYEXCLUDETHIS}
\begin{restatable}[Smoothness of \MARK]{theorem}{thmlowerboundmark}
Let $\alpha=\max_{1< \ell \leq k}\left\{\frac{\ell(1+H_k-H_\ell)}{\ell-1+H_k-H_{\ell-1}}\right\}=\Omega(H_k)$, where $k$ is the cache size. 
	$\MARK$ is not $(\alpha-\epsilon,\gamma,1)$-smooth for any $\epsilon>0$ and any~$\gamma$. Also, $\MARK$ is $(2H_k-1,\delta (4H_k-2))$-smooth.\label{thm:mark}
\end{restatable}
\begin{FULL}
\ifthenelse{\boolean{proof-thmlowerboundmark}}{%% proof of thmlowerboundmark
\begin{proof}
The upper bound follows from the fact that $\MARK$ is $2H_k-1$ competitive with additive constant~$0$~\cite{Achlioptas00} and Theorem~\ref{thm:competitivesmoothness}. 
For the lower bound, we construct sequences $\sigma$ and $\sigma'$ with $\Delta(\sigma,\sigma')=1$ such that $\MARK(\sigma')\ge \alpha\cdot \MARK(\sigma)+\beta'$. Consider the $k$-phase partition of $\sigma$ and let $\sigma_h$ denote the $h^{th}$ phase of $\sigma$. The sequence $\sigma'$ will be such that its phases will be shifted to the left with respect to those of $\sigma$. More specifically, if $i_h$ and $i'_h$ are the indices of the first page of $\sigma_h$ and $\sigma'_h$, respectively, then for all $h\ge 3$, $i_h = i_h'+\ell-1$, where $\ell$ is a parameter satisfying $2\le \ell \le k$. After an initial setup that includes the difference between sequences, each phase of $\sigma'$ and $\sigma$ will consist of exactly $k$ pages. For a given phase  and sequence we call a page \emph{old} if it was requested in the previous phase and \emph{new} otherwise. Since there are $k$ requests per phase, no pages are repeated during a phase. A phase for $\sigma'$ will consist of $\ell$ new pages followed by $k-\ell$ old ones. In turn, a phase for $\sigma$ starts with one new page, followed by $k-\ell+1$ old pages, and ends with $\ell-2$ new pages (see Figure~\ref{fig:markphases}).

\begin{figure}[t]
\begin{center}
\includegraphics{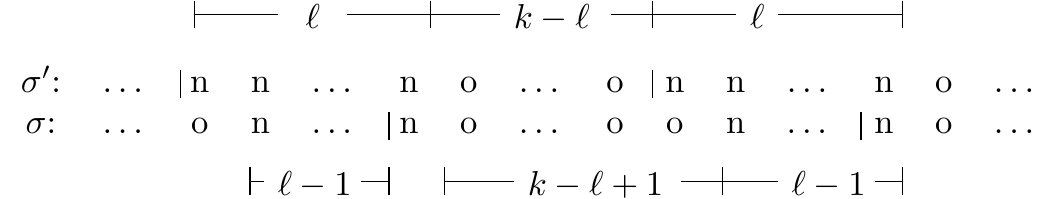}
\end{center}
\caption{New (n) and old (o) pages in each phase in $\sigma$ and $\sigma'$ for the lower bound of $\MARK$. Vertical bars between pages indicate the boundaries of phases.}
\label{fig:markphases}
\end{figure}

Since at the end of a phase $h$ all pages of the phase are in $\MARK$'s cache, a new page in phase $h+1$ is always a fault.  On the other hand, a request for the $j^{th}$ old page in the phase has a fault probability of $n_j/(k-j+1)$, where $n_j$ is the number of new pages in the phase before the request to this page~\cite{borodin98}. Then, the expected number of misses in a phase of $\sigma'$ is $\ell+\sum_{j=1}^{k-\ell}\frac{\ell}{k-j+1}=\ell(1+H_k-H_\ell)$, while the expected number of misses in a phase of $\sigma$ is $1+\sum_{j=1}^{k-\ell+1}\frac{1}{k-j+1} + \ell-2=\ell-1+H_k-H_{\ell-1}$. 

We now show that given a phase $\sigma'_h$ as described above, with $\sigma$ satisfying $i_h = i_h'+\ell-1$, we can construct the structure of the $\sigma_h$ and moreover we can maintain the configuration for the next phase. Since $i_h = i_h'+\ell-1$, the request in $i_h$ coincides with the last new page of $\sigma'_h$. This page can be any new page not requested so far, and hence is also new for $\sigma_h$. Now $k-\ell$ old pages follow in $\sigma'_h$ that are different from the first $\ell$ of the phase. Note that the first $\ell-1$ pages of $\sigma_h'$ belong to phase $\sigma_{h-1}$. This leaves $k-\ell+1$ of possible pages in $\sigma_{h-1}$ that can be old for $\sigma$. We request those pages next. The first $k-\ell$ are old in $\sigma'_h$ and the last one is the one that starts phase $\sigma'_{h+1}$. From then, we request $\ell-1$ new pages that have not been requested before and hence are new for both $\sigma$ and $\sigma'$. Since there have been already $k-\ell+2+\ell-1=k+1$ distinct pages requested in $\sigma_h$, the last of these $\ell-1$ pages marks the start of the $(h+1)^{th}$ phase of $\sigma$ and $i_{h+1}=i_{h+1}'+\ell-1$.

Finally, we show that the configuration of phases described above can be reached with one difference between $\sigma$ and $\sigma'$ and a constant number of misses on both sequences. An example sequence is shown below for any $k$ and $l$. The structure of the phases described above is satisfied starting from the third phase of the sequences.

\todo{commas or no commas between pages}

\begin{center}
\small
\begin{tabular}{*{28}{c}}
 &   &  &&&&&&&&&&&&&&n&\ldots &n&n& o&\ldots&o&n&n&\ldots & n & n\\  
$\sigma':$ & $x_0$   & $x_1$ & $x_2$ & $\ldots$ & $x_{k-1}$ & $|x_k$ &$x_1$ & $x_2$ & $\ldots$ & $x_{\ell-1}$ & $y_{1}$ & $x_1$ & $y_{2}$  & $\ldots$ & $y_{k-\ell}$& $|y_{k-\ell+1}$ & $\ldots$ & $y_{k-1}$ & $y_{k}$ &  $y_{1}$ &  $\ldots$ & $y_{k-\ell}$ & $|x_1$ & $x_2$ & $\ldots$ & $x_{\ell-1}$& $x_{\ell}$\\
%&\\
 $\sigma:$ &        & $x_1$ & $x_2$ & $\ldots$ & $x_{k-1}$ & $x_k$ &$x_1$ & $x_2$ & $\ldots$ & $x_{\ell-1}$ & $|y_{1}$ & $x_1$ & $y_{2}$  & $\ldots$ & $y_{k-\ell}$& $y_{k-\ell+1}$ & $\ldots$ & $y_{k-1}$ & $|y_{k}$ &  $y_{1}$ &  $\ldots$ & $y_{k-\ell}$ & $x_1$ & $x_2$ & $\ldots$ & $x_{\ell-1}$& $|x_{\ell}$\\
&&  &&&&&&&&&&&&&&& &&n& o&\ldots&o&o&n&\ldots & n & n\\ 
\end{tabular}
\normalsize
\end{center}

For each phase $h$ the ratio between faults in $\sigma_h'$ and $\sigma_h$ is $\frac{\ell(1+H_k-H_\ell)}{\ell-1+H_k-H_{\ell-1}}$. Taking $\ell=\lceil H_k \rceil$ it is easy to show that this ratio is $\Theta(H_k)$. By the argument above the sequences can be extended to an arbitrary length maintaining the phase configurations and the theorem follows.\qed
\end{proof}}{}
\end{FULL}

\begin{FORPROOFSONLYEXCLUDETHIS}
We conjecture that the lower bound for \MARK is tight, i.e., that \MARK is $(\alpha, \beta)$-smooth for $\alpha$ as defined in Theorem~\ref{thm:mark} and some $\beta$.

We now prove that \RANDOM achieves the same bounds for smoothness as \LRU and the best possible for any deterministic, demand-paging or competitive algorithm. \begin{SHORT}For simplicity, we prove the theorem for a non-demand-paging definition of \RANDOM in which each page gets evicted upon a miss with probability $1/k$ even if the cache is not yet full.\end{SHORT} Intuitively, the  additive term $k+1$ in the smoothness of \RANDOM is explained by the fact that a single difference between two sequences can make the caches of both executions differ by one page $p$. Since $\RANDOM$ evicts a page with probability $1/k$, the expected number of faults until $p$ is evicted is $k$. 
%The proof, however, requires a detailed case analysis with a potential function based on the difference between the probability distribution of the cache states in both executions.
\end{FORPROOFSONLYEXCLUDETHIS}

\begin{restatable}[Smoothness of \RAND]{theorem}{thmrandomsmoothness}
\label{thm:rand}~\\
\indent $\RAND$ is $(1,\delta(k+1))$-smooth. This is tight.
\end{restatable}

\begin{FULL}
\ifthenelse{\boolean{proof-thmrandomsmoothness}}{%% proof of thmrandomsmoothness
\begin{proof}
For this theorem we use a non-demand paging definition of $\RAND$ that, upon a fault, it evicts a page with probability $1/k$, even if the cache is not yet full.
%If the cache is full, pages are evicted upon faults as in the standard definition of \RAND. 
%Note that the non-demand paging version is still $k$-competitive. 
This modification with respect to the demand-paging version does not change the competitiveness of the algorithm and it allows us to avoid the analysis of special cases when proving properties about smoothness. In fact, for the non-demand paging version of the algorithm, given any pair of sequences $\sigma,\sigma'$, it is possible to construct two sequences $\rho$ and $\rho'$ with $\Delta(\rho,\rho')=\Delta(\sigma,\sigma')$ such that the number of faults of $\sigma$ and $\sigma'$ starting from an empty cache equals the number of faults of $\rho$ and $\rho'$ starting with a full cache containing an arbitrary set of pages. This can be achieved by renaming in $\sigma$ and $\sigma'$ any occurrences of the pages in the initial cache so that these pages do not appear in the rest of the sequences. This implies that any property derived on the smoothness of the algorithm starting with an empty cache can also be achieved when the cache is assumed to be initially full. Note that the same property holds for $\lrurandom$, which is introduced in Section~\ref{sec:lrurandom}.

For the lower bound, we use a similar construction as the one used for the lower bound of \RAND's competitiveness in~\cite{DBLP:conf/icalp/RaghavanS89}. Consider the sequences 
$\sigma=\sigma_{1...k}\cdot\sigma_{1...k}$
and
$\sigma'=\sigma_{1...k}\cdot x_{k+1}\cdot\sigma_{1...k}$
with
$\sigma_{1...k} = (x_1,x_2,\ldots, x_k)^n$. 
The sequences are identical but for the insertion of $x_{k+1}$ in $\sigma'$ and thus $\Delta(\sigma,\sigma')=1$. 
For any $\epsilon > 0$, the number of faults in the second half of $\sigma$ is less than $\epsilon$ for a sufficiently large $n$.
On $\sigma'$, $\RAND$ faults on $x_{k+1}$ and evicts one of $x_1,\ldots,x_k$. Then, on each of the $n$ subsequences of $k$ requests in the second part of $\sigma'$, and while $x_{k+1}$ is still in its cache, $\RAND$ will incur a fault. If in one of these faults $\RAND$ evicts $x_{k+1}$, then it does not incur any faults for the rest of the sequence. Since on every fault $\RAND$ evicts $x_{k+1}$ with probability $1/k$, the expected number of faults until this happens exceeds $k-\epsilon$ for any $\epsilon > 0$ for sufficiently large $n$. This, plus the initial request to $x_{k+1}$ yield $\RAND(\sigma')\ge \RAND(\sigma)+k+1-2\epsilon$ for any $\epsilon$ and sufficiently large $n$. Now, for general $\delta$ we follow the same idea: instead of one, we have $\delta$ subsequences $(x_1\ldots x_k)^n$ in $\sigma$ and $\delta$  subsequences $y_{i}(x_1\ldots x_k)^n$, where $y_{i}$ ($1 \le i \le \delta)$ is a new page not requested so far, and it is distinct in every repetition. The number of expected faults in each repetition in $\sigma'$ is at least $k+1-\epsilon$ for any $\epsilon$ and sufficiently large $n$, while $\RAND$ does not incur extra faults. Thus,  $\RAND(\sigma')\ge \RAND(\sigma)+\delta(k+1)-\epsilon(\delta+1)$. %Adding up over all repetitions we obtain the lower bound.

In order to prove the upper bound we look at the state distributions of \RAND~when serving two sequences $\sigma$ and $\sigma'$ with $\Delta(\sigma,\sigma')=1$. We use a potential function defined as the distance between two state distributions. For this distance, we define a version of the earth mover's distance. Let $D$ and $D'$ be two probability distribution of cache states. We define the distance between $D$ and $D'$ as the minimum cost of transforming $D$ into $D'$ by means of transferring probability mass from the states of $D$ to the states of $D'$. 

Let $s$ and $s'$ be two cache states in $D$ and $D'$ with probabilities $p_s$ and $p_{s'}$, respectively. Let $\alpha$ be a function that denotes the amount of probability mass to be transferred from states in $D$ to states in $D'$. The earth mover's distance between $D$ and $D'$ is defined as

\[\Delta(D, D') := \min_{\alpha} \sum_{s,s'} \alpha(s, s')\cdot d(s,s'),\]
where for all $s$, $\sum_{s'} \alpha(s,s')=p_s$, for all $s'$, $\sum_{s} \alpha(s,s')=p_{s'}$, and $d(s,s')$ is the distance between states $s$ and~$s'$. We define $d(s,s')=k \cdot H_{c(s,s')}$, where $c(s,s')=\max\{|s\setminus s'|,|s'\setminus s|\}$, and $H_\ell$ is the $\ell^{th}$ harmonic number. Note that $|s\setminus s'|$ might not equal $|s'\setminus s|$ if either state does not represent a full cache. 
For convenience we let $H_0=0$. We now prove the following claim:

\begin{claim}
Let $D$ and $D'$ be two probability distributions over cache states. Let $\sigma$ be any request sequence and let $M_D(\sigma)$ and $M_{D'}(\sigma)$ be two random variables equal to the number of misses on $\sigma$ by \RAND~when starting from distributions $D$ and $D'$, respectively. Then, $E[M_D'(\sigma)] - E[M_{D}(\sigma)]  \le \Delta(D,D')$. 
\end{claim}

Let us assume that the claim is true. Then, we prove the theorem by considering two sequences $\rho$ and $\rho'$ such that $\delta=\Delta(\rho,\rho')=1$ and arguing that $\Delta(D,D')\le k$ for any pair of distributions $D$ and $D'$ that can be reached, respectively, by serving prefixes of $\rho$ and $\rho'$ starting from an empty cache. If this prefix includes the single difference between both sequences, then the theorem follows by applying the claim above to the maximal suffix $\sigma$ shared by both sequences. 

Let $j$ be the minimum $j$ such that $\rho[(j+1)..|\rho|]=\rho'[(j+1)..|\rho'|]=\sigma$. Then, $\rho_{j}\ne\rho'_{j}$ (one of the two might be empty) and $\rho[1..(j-1)]=\rho'[1..(j-1)]$. Since $\RAND(\rho)$ and $\RAND(\rho')$ start both with an empty cache, their distributions and expected misses before serving $\rho_j$ and $\rho_j'$ coincide. We now argue that after serving $\rho_j$ and $\rho_j'$ the distance between the resulting distributions $D$ and $D'$ is at most $k$.

Let $F$ be the state distribution of both executions before serving $\rho_j$ and $\rho_j'$. 
Suppose first that $\rho_j'$ is empty and thus $D'=F$ (the case when $\rho_j$ is empty is symmetric). We look at the minimum cost to transfer the probability mass from each state from $F$ to $D$. Let $s_i$ be a state in $F$ with probability $p_i$. If $\rho_j \in s_i$ then $s_i$ has probability at least $p_i$ in $D$ and hence we can transfer $p_i$ mass between these states in $F$ and $D$ at cost zero. Otherwise, if  $\rho_j \notin s_i$, $D$ contains $k$ states $s_{i_1}',\ldots,s_{i_k}'$ resulting from the eviction of each of the $k$ pages of $s_i$, with $c(s_i,s_{i_r}')=1$ and hence $d(s_i,s_{i_r}')=kH_1=k$ for all $1\le r \le k$. Moreover, the probability of these states is at least $p_i/k$ and hence we can transfer all the mass of $s_i$ to these states at a total cost of $\sum_{r=1}^k k(p_i/k)= kp_i$. Adding up over all states $s_i \in F$, we can transfer all probability mass of $F$ to $D$ at a cost of at most $k\sum {p_i}=k$, since $\sum {p_i}=1$. Since the distance between $D$ and $F$ is the minimum cost of transferring the probability mass from $F$ to $D$, this cost is at most $k$.
For the case when  $\rho_j\ne  \rho_j'$ and neither is empty, we apply a similar argument. Let $s_i$ be a state in $F$ with probability $p_i$. If both  $\rho_j$ and $\rho_j'$ are in $s_i$, then this state is also in $D$ and $D'$, an we can transfer $p_i$ from $D$ to $D'$ at cost zero. Assume that $\rho_j \in s_i$ but $\rho_j'\notin s_i$. Then, as we argued above, in $D'$ there are $k$ states with probability at least $p_i/k$ with distance 1 to $s_i$. Since $s_i\in D$, we can transfer a mass of $p_i$ to these states at a cost of $kp_i$. Now, if $\rho_j \notin s_i$ but $\rho_j'\in s_i$, $s_i$ is in $D'$ and there are $k$ states in $D$ with distance 1 to $s_i$. We can transfer $p_i/k$ mass from each of these states in $D$ to $s_i$ in $D'$ at a cost of $kp_i$. Finally, if $\rho_j \notin s_i$ and $\rho_j'\notin s_i$, then there are $k$ pairs of states $(s,s')$ with $s\in D$ and $s'\in D'$ resulting from the replacement of the same page in $s_i$ by $\rho_j$ and $\rho_j'$, respectively, and thus $c(s,s')=1$. In the distance $\Delta(D,D')$ we can transfer $p_i/k$ from $s$ to $s'$ at a cost of $k$. Since there are $k$ such such pairs for each $s_i$, the total cost contributed by these pairs is $p_ik$. Since for all cases the cost contributed by a state $s_i \in F$ when transferring mass from $D$ to $D'$ is at most $kp_i$, the distance  $\Delta(D,D')$ is at most $k\sum {p_i}=k$.

Since serving $\rho_j$ and $\rho_j'$ can add at most 1 to the difference in expected misses, and by the claim above the difference in expected misses in the suffix $\sigma$ is at most $\Delta(D,D')=k$, it follows that $E[\RAND(\rho')- \RAND(\rho)] \le k+1$. The theorem follows by Corollary~\ref{cor:1delta}.

\todo{define substrings of sequences and pages within sequences somewhere}
We now prove the claim. Let $M_{D}(\sigma_i)$ denote the number of misses of \RAND~when $\sigma_i$ is requested and when the state distribution of \RAND~is $D$. Let $D_0=D$ and $D_0'=D'$. Then, it is sufficient to prove that for every request $\sigma_i \in \sigma$, for $1\le i \le |\sigma|$, 
\begin{equation}
E[M_{D_{i-1}}(\sigma_i)]-E[M_{D'_{i-1}}(\sigma_i)] \le \Delta(D_{i-1},D'_{i-1})-\Delta(D_{i},D'_{i}). 
\label{eq:potential}
\end{equation}

This implies that $E[M_D(\sigma)] - E[M_{D'}(\sigma)] \le  \Delta(D_{0},D'_{0}) - \Delta(D_{f},D'_{f}) \le \Delta(D_{0},D'_{0}) = \Delta(D,D')$, since $\Delta(\cdot,\cdot)\ge 0$ for any pair of distributions.

Let $D_{i-1}$ and $D_{i-1}'$ be the distributions before the request to $\sigma_i$. \[\Delta(D_{i-1},D'_{i-1})=\sum_{s_u,s_v}\alpha(s_u,s_v)d(s_u,s_v),\] 
where $\alpha(s_u,s_v)$ is the amount of mass transferred from $s_u$ to $s_v$ (which could be zero). We look at two states $s_u \in D$ with probability $p_u$ and $s_v\in D'$ with probability $p_v$ and construct a valid assignment $\alpha'$ after the request to $\sigma_i$ for the distance $\Delta(D_{i},D'_{i})$.

 We separate the analysis in the following cases:

\begin{enumerate}
\item $[\sigma_i \in s_u, s_v]$ In this case $s_u \in D_i$ with probability at least $p_u$ and $s_v \in D'_i$ with probability at least $p_v$. Hence, since $\alpha(s_u,s_v)\le p_u, p_v$ we can make $\alpha'(s_u,s_v)=\alpha(s_u,s_v)$. The contribution of this pair of states to   $\Delta(D_{i},D'_{i})$ is $\alpha(s_u,s_v)d(s_u,s_v)= \alpha(s_u,s_v)kH_{c(s_u,s_v)}$.
\item $[\sigma_i \notin s_u, s_v]$ There are $k$ states $r=\{r_1,\ldots,r_k\}$ in $D_i$ and $t=\{t_1,\ldots,t_k\}$ in $D_i'$ resulting from the eviction of each page of $s_u$ and $s_v$, respectively. The probability of each state of $r$ and $t$ is at least $p_u/k$ and $p_v/k$, respectively. Let $c=c(s_u,s_v)$ and $\alpha=\alpha(s_u,s_v)$. If $c=0$, then we pair states in $r$ and $t$ such that $r_{j_1}=t_{j_2}$ and we make $\alpha'(r_{j_1},t_{j_2})=\alpha/k$. Otherwise, there are $c$ pages that $s_u$ and $s_v$ do not have in common. We sort the states in $r$ and $s$ such that the first $c$ states are those that result from evicting a page from $s_u$ that is not in $s_v$ and vice versa, while the rest of the states are the ones resulting from evicting a common page. We pair the states in order and set $\alpha'(r_j,t_j)=\alpha/k$. Note that $c(r_j,t_j)=c-1$ for all $j\le c$ and $c(r_j,t_j)= c$ for all $j> c$. The  contribution of this pair of states to $\Delta(D_{i},D'_{i})$ is at most

$$(\alpha/k) (ckH_{c-1}+(k-c)kH_{c})=
\alpha(kH_{c}+c(H_{c-1}-H_c))=
\alpha(kH_{c}-1)
$$

\item $[\sigma_i \in s_u, \sigma_i \notin s_v]$ We transfer $\alpha(s_u,s_v)/k$ to the $k$ states in $D'_{i}$ resulting from evictions from $s_v$. As in case 2. there are $c=c(s_u,s_v)$ states that result from evicting a non-common page with $s_u$ and the rest evict a common page. Each of the first $c$ states has $c-1$ non-common pages with $s_u$, while the rest have $c$ non-common pages. Hence, the contribution of these states to $\Delta(D_{i},D'_{i})$ is $\alpha(s_u,s_v)(kH_{c(s_u,s_v)}-1)$.
\item $[\sigma_i \notin s_u, \sigma_i \in s_v]$ This case is analogous to case 3. We transfer $\alpha(s_u,s_v)/k$ mass to $s_v \in D'$ from each of the $k$ states in $D$ that result from evictions from $s_u$. The contribution of these states is $\alpha(s_u,s_v)(kH_{c(s_u,s_v)}-1)$.
\end{enumerate}

Since in the cases above we account for all the probability mass of all possible states in $D_i$ and $D_i'$, the described mass transfer is a valid distance between the distributions, and its cost is:

\begin{eqnarray*}
\Delta(D_{i},D'_{i}) & \le   &\sum_{s_u,s_v |\sigma_i\in s_u,s_v} \alpha(s_u,s_v)kH_{c(s_u,s_v)}+\sum_{s_u,s_v |\sigma_i\notin s_u, \sigma_i \in s_v} \alpha(s_u,s_v)(kH_{c(s_u,s_v)}-1)\\
&  & + \sum_{s_u,s_v |\sigma_i\in s_u, \sigma_i \notin s_v} \alpha(s_u,s_v)(kH_{c(s_u,s_v)}-1) + \sum_{s_u,s_v |\sigma_i\notin s_u, \sigma_i \in s_v}  \alpha(s_u,s_v)(kH_{c(s_u,s_v)}-1)\\
& = & \Delta(D_{i-1},D'_{i-1}) - \sum_{s_u,s_v |\sigma_i \notin s_u \vee \sigma_i \notin s_v}\alpha(s_u,s_v)\\
& \le &  \Delta(D_{i-1},D'_{i-1}) - \sum_{s_u,s_v |\sigma_i \notin s_u}\alpha(s_u,s_v)
\label{eq:}
\end{eqnarray*}

Therefore, $\Delta(D_{i-1},D'_{i-1})-\Delta(D_{i},D'_{i})\ge \sum_{s_u,s_v |\sigma_i \notin s_u}\alpha(s_u,s_v)$.
On the other hand,  $E[M_{D_{i-1}}(\sigma_i)] - E[M_{D'_{i-1}}(\sigma_i)] \le E[M_{D_{i-1}}(\sigma_i)] = \sum_{s_u,s_v |\sigma_i \notin s_u}\alpha(s_u,s_v)$, and hence $E[M_{D_{i-1}}(\sigma_i)] - E[M_{D'_{i-1}}(\sigma_i)] \le \Delta(D_{i-1},D'_{i-1})-\Delta(D_{i},D'_{i})$. \qed
\end{proof}}{}
\end{FULL}

\begin{FORPROOFSONLYEXCLUDETHIS}

\nvsp
\nvsp
\subsection{Trading Competitiveness for Smoothness}
\nvsp

We have seen that none of the well-known randomized algorithms are particularly smooth.
\RANDOM is the only known randomized algorithm that is $(1, \delta c)$-smooth for some $c$.
However, it is neither smoother nor more competitive than \LRU, the smoothest deterministic algorithm. In this section we show that greater smoothness can be achieved at the expense of competitiveness. First, as an extreme example of this, we show that Evict-on-access (\EOA)~\cite{Cazorla13}---the policy that evicts each page with a probability of $\frac{1}{k}$ upon \emph{every} request, i.e., not only on faults but also on hits---beats the lower bounds of Theorems~\ref{thm:lowerboundrandomizeddemandpaging} and~\ref{thm:lowerboundrandomizedstronglycompetitive} and is strictly smoother than \opt. This policy is non-demand paging and it is obviously not competitive. We then introduce \smoothlru, a parameterized randomized algorithm that trades competitiveness for smoothness.

%Finally, we show that Evict-on-access (\EOA)---the policy that evicts a page chosen uniformly at random on \emph{every} request, i.e., not only on faults but also on hits---is as smooth as \opt for $k=1$ and strictly smoother for $k>1$. This policy, however, is obviously not competitive.
\end{FORPROOFSONLYEXCLUDETHIS}

\begin{restatable}[Smoothness of \EOA]{theorem}{thmeoasmoothness}
\label{thm:eoa}~\\
\indent$\EOA$ is $(1,\delta(1+\frac{k}{2k-1}))$-smooth. This is tight.
\end{restatable}
\nvsp

\begin{FULL}
\ifthenelse{\boolean{proof-thmeoasmoothness}}{%% proof of thmeoasmoothness
\begin{proof}
We show that \EOA is $(1,1+k/(2k-1),1)$-smooth. The theorem follows from Corollary~\ref{cor:1delta}. 
Let $\sigma$ be a sequence and let $\sigma_i$ denote the $i$-th request in $\sigma$. We denote by $d_i$ the \emph{reuse distance} of $\sigma_i$, i.e., the number of requests since the last request to this page in $\sigma$ (not including $\sigma_i$). If $\sigma_i$ is the first request to this page, then $d_i=\infty$. %This is similar to the notion of age that was
Consider a page $p$ in \EOA's cache. On a request for another page, $p$ is evicted with probability $1/k$ and thus the probability that $p$ is still in cache after $d$ requests to other pages is $(1-1/k)^d$. Hence, the probability of any request $\sigma_i$ being a hit is $(1-1/k)^{d_i}$ and therefore $\EOA(\sigma)=\sum_{i=1}^{|\sigma|} 1-(1-1/k)^{d_i}$.

Let $\sigma'$ be a sequence resulting from one change to $\sigma$. Assume that this change is an insertion of a new page $p$ that is not requested elsewhere in $\sigma$. The insertion of $p$ increases the reuse distance of all requests after $p$ whose previous request was before $p$. Let $d_i'$ denote the reuse distances of request $i$ in $\sigma'$. For simplicity of notation, assume that if $p$ is requested at index $j$ in $\sigma'$, then we add a request $\sigma_j$ to $\sigma$ for an empty page (which has distance $d_j=0$, does not affect other reuse distances nor requires any action from EOA). Hence, indices in both sequences correspond to the same pages (but for the request to $p$).

Let $S$ denote the set of indices of requests whose distances increase in $\sigma'$ compared to the ones they had in $\sigma$ after the insertion of $\sigma'_j=p$. That is, $S=\{i \mid i>j, d_i'\ge i-j\}$. Then $d_j'=\infty$, $d_i'=d_i+1$ if $i\in S$, and $d_i'=d_i$ otherwise.
Hence, 
\begin{eqnarray}
\EOA(\sigma')-\EOA(\sigma) &=& \left(\sum_{i=1}^{|\sigma'|} 1-\left(1-\frac{1}{k}\right)^{d_i'}\right)- \left(\sum_{i=1}^{|\sigma|} 1-\left(1-\frac{1}{k}\right)^{d_i}\right)\\
&=& 1+\sum_{i \in S}\left(1-\frac{1}{k}\right)^{d_i}-\left(1-\frac{1}{k}\right)^{d_i+1}\\
&=&1+\frac{1}{k}\sum_{i \in S}\left(1-\frac{1}{k}\right)^{d_i} 
\label{eq:eoa}
\end{eqnarray}

We now show that for all $S$ and valid distances $d_i$, $\sum_{i \in S}(1-1/k)^{d_i}\le k^2/(2k-1)$ and hence $\EOA(\sigma')-\EOA(\sigma)\le 1+k/(2k-1) \le 2$.

Let $i\in S$ and let $i'<j$ be the largest index such that $\sigma_{i'}=\sigma_i$. Then, $d_i=i-i'-2\ge 0$ (recall that $\sigma_j$ does not contribute to the reuse distance). It is convenient to represent $d_i$ as a pair $(a_i,b_i)$, where $a_i$ is the number of requests from $\sigma_{i'}$ until $\sigma_j$ and $b_i$ is the number of requests from $\sigma_{j}$ to $\sigma_i$. Thus, $d_i=a_i+b_i$. Let $\vec{d_S}=\{ (a_i,b_i) \}_{i=1}^{|S|}$ denote a configuration of reuse distances of requests with indices in  $S$, and let $h_{S,d}=\sum_{i \in S}(1-1/k)^{d_i}$. We claim that $h_{S,d}$ is maximal when, after the request $j$, all previous requests are requested  in reverse order. This is, $S=\{j+1,\ldots,2j-1\}$, and for all $d_i$, $a_i=b_i=i-j-1$. First, we argue that the maximum  of $h_{S,d}$ is attained when $S=\{j+1,\ldots,2j-1\}$, for some configuration $\vec{d_S}$. The size of $S$ can be at most $j-1$, as there cannot be more than $j-1$ requests whose previous requests are before $\sigma_j$. Now, assume that $S$ does not contain an index from $j+1$ to $2j-1$. Let $t_1$ be the smallest such index. Then, either there exists an index $t_0<j$ such that $\sigma_{t_0}$ is not requested after $\sigma_j$ or the next request to $\sigma_{t_0}$ is after $t_1$. In either case, we can modify $\sigma$ to obtain $S'$ (and the corresponding configuration $d'=d_{S'}$) that includes $t_1$ by making $\sigma_{t_1}=\sigma_{t_0}$, with $h_{S',d'}\ge h_{S,d}$. Suppose that $\sigma_{t_0}$ was not requested after $\sigma_{j}$. Then, making $\sigma_{t_1}=\sigma_{t_0}$ implies $h_{S',d'}= h_{S,d}+(1-1/k)^{t_1-t_0-2} \ge h_{S,d}$. Otherwise, $\sigma_{t_0}$ is requested again at index $t_2>t_1$. In this case, by making $\sigma_{t_1}=\sigma_{t_0}$, $h_{S',d'}= h_{S,d}-(1-1/k)^{t_2-t_0-2}+(1-1/k)^{t_1-t_0-2}$. Since $t_2>t_1$, $h_{S',d'} \ge h_{S,d}$. Thus, the largest value of $h_{S,d}$ is obtained for $S=\{j+1,\ldots,2j-1\}$. 
We now fix $S$ to be this set and show that $h_{S,d}$ is maximized when for all $i\in S$, $a_i=b_i=i-j-1$.

Let $\vec{d_S}=\{ (a_i,b_i) \}_{i=1}^{|S|}$ be a configuration of reuse distances.  We say that a configuration $\vec{d_S}$ is valid if for all pairs $i_1,i_2$ in $\vec{d_S}$ with $i_1 \ne i_2$  , $a_{i_1}\ne a_{i_2}$ and $b_{i_1} \ne b_{i_2}$. Since in a sequence only one request can start and end at each index, any configuration built from reuse distances of a subset $S$ of indices in an actual sequence is valid.

We now show that the maximum $h_{S,d}$ is attained when $a_i=b_i=i-j-1$ for all $i\in S$, i.e.,  for $\vec{d_S}=\{(0,0),(1,1),\ldots,(j-2,j-2)\}$ and $d_i=2i$ for $i=0,\ldots,j-2$. In this case, we say that $\vec{d_S}$ is a diagonal. 

Suppose, to the contrary, that $\vec{d_S}$ is a valid configuration that maximizes $h_{S,d}$ and is not a diagonal. Then, since $\vec{d_S}$ is valid there must exist at least two pairs $(a_{i_1},b_{i_1})$ and $(a_{i_2},b_{i_2})$ in $\vec{d_S}$ with $a_{i_1}\ne b_{i_1}$ and $a_{i_2}\ne b_{i_2}$ and, furthermore, $a_{i_1} > a_{i_2}$ and $b_{i_1} < b_{i_2}$. Let $\vec{d'_S}$ be a new configuration created by removing these two pairs from $\vec{d_S}$ and adding the pairs $(a_{i_1},b_{i_2})$ and $(a_{i_2},b_{i_1})$. It is easy to see that $\vec{d'_S}$ is also a valid configuration. In fact, $\vec{d'_S}$ corresponds to exchanging requests $\sigma_{j+b_1}$ and $\sigma_{j+b_2}$ in $\sigma$. 

We show that $h_{S,d'}> h_{S,d}$. Let $c=(1-1/k)$. Then $h_{S,d'}-h_{S,d}= c^{a_{i_2}+b_{i_1}}+c^{a_{i_1}+b_{i_2}}-(c^{a_{i_1}+b_{i_1}}+c^{a_{i_2}+b_{i_2}})$.
Since $c<1$, $a_{i_1} > a_{i_2}$, and $b_{i_1} < b_{i_2}$, then $c^{a_{i_1}}<c^{a_{i_2}}$ and $c^{b_{i_1}}>c^{b_{i_2}}$. Hence,  $c^{a_{i_2}}(c^{b_{i_1}}-c^{b_{i_2}})> c^{a_{i_1}}(c^{b_{i_1}}-c^{b_{i_2}})$ $\Rightarrow$ $c^{a_{i_2}+b_{i_1}}-c^{a_{i_2}+b_{i_2}}> c^{a_{i_1}+b_{i_1}}-c^{a_{i_1}+b_{i_2}}$ $\Rightarrow$ $c^{a_{i_2}+b_{i_1}}+c^{a_{i_1}+b_{i_2}} > c^{a_{i_1}+b_{i_1}}+c^{a_{i_2}+b_{i_2}} $. Therefore, $h_{S,d'}> h_{S,d}$, and hence $h_{S,d}$ is not maximal. Since for a fixed set $S$ there is a finite number of configurations $\vec{d_S}$, the maximum of $h_{S,d}$ exists and it must be attained when $\vec{d_S}$ is a diagonal.  
Then, for all $S$ and $d_S$, $h_{S,d}\le \sum_{i=0}^{j-2}(1-1/k)^{2i} \le \sum_{i=0}^{\infty}(1-1/k)^{2i} = k^2/(2k-1)$. %\le \sum_{i=0}^{\infty}(1-1/k)^{i}  = k$.

We have shown an upper bound on the number of extra misses of \EOA when $\sigma'$ is the result of an insertion of a new page $p$ into $\sigma$. 
It remains to argue that this is an upper bound when $p$ is not a new page, or when the change is a substitution or deletion. 

\begin{itemize}
\item \textbf{Insertion}: We compare the expected number of misses to the case of $\sigma'$ considered above. Suppose $\rho$ is a sequence resulting from inserting $\rho_j=p'$ in $\sigma$ that had been requested before but not after. Let $j'<j$ be the largest index with $\rho_{j'}=p'$. The miss probability of $\rho_j$ is $\epsilon < 1$, whereas the miss probability of $\sigma'_j$ is 1. Furthermore, the reuse distances of other pages are at least as large in $\sigma'$ as they are in $\rho$. Hence $\EOA(\sigma')>\EOA(\rho)$. Suppose now that the insertion is of a page that is requested later but not before. Let $j'>j$ be the smallest index such that $\rho_{j'}=\rho_j$. Then, the miss probability of both $\sigma'_j$ and $\rho_j$ is 1 and the miss probability of $\sigma_{j'}$ is larger than that of $\rho_{j'}$. In addition, the reuse distance of other pages are at least as large in $\sigma'$ as they are in $\rho$. Thus, $\EOA(\sigma')>\EOA(\rho)$. Finally, suppose that $\rho_{j}$ is requested both before and after (at index $j'$) and that the reuse distances of requests at $j$ and $j'$ are $d$ and $d'$, respectively. The expected number of misses of these requests is $2-c^d-c^{d'}$. On the other hand, the expected number of misses on these requests on $\sigma'$ is $2-c^{d+d'+1}$. Since $c<1$, $c^{d+d'}\le c^d$ and $c^{d+d'}\le c^{d'}$. Hence $c^{d+d'+1}=c\cdot c^{d+d'}\le c^d+c^{d'}$. Hence, $2-c^{d+d'+1}\ge $$2-c^d-c^{d'}$, and since the expected number of misses on the rest of the requests is equal for both sequences, $\EOA(\sigma')>\EOA(\rho)$. \todo{this can most likely be simplified}

\item \textbf{Deletion}: Let $\sigma'$ be the sequence resulting from deleting $\sigma_j$ from $\sigma$. We argue that $\EOA(\sigma')\le\EOA(\sigma)$. Suppose that $\sigma_j$ is not requested later, then the reuse distances of pages in $\sigma'$ are at most those on $\sigma$, and $\EOA(\sigma)$ includes the non-zero miss probability of $\sigma_j$, which is not present in $\EOA(\sigma')$. Hence, $\EOA(\sigma')\le\EOA(\sigma)$. Suppose now that $\sigma_j$ is requested later in $\sigma$, but not earlier. Let $j'>j$ be the smallest index such that $\sigma_{j'}=\sigma_j$. Then the expected number of misses of these pages in $\sigma$ is at least 1, while the probability of miss of $\sigma_j$ in $\sigma'$ is 1. Again, the reuse distance of other pages in $\sigma$ are at least those in $\sigma'$ and thus $\EOA(\sigma')\le\EOA(\sigma)$. Finally, assume that $\sigma_j$ is requested before and after in $\sigma$ (at index $j'$) and that the reuse distances of requests at $j$ and $j'$ are $d$ and $d'$. The expected number of misses of these requests in $\sigma$ is $m(\sigma)=2-c^d-c^{d'}$, while the expected number of misses of $\sigma_{j'}$ in $\sigma'$ is $m(\sigma')=1-c^{d+d'}$. We claim that $m(\sigma')\le m(\sigma)$. This is true if $c^d+c^{d'}\le 1+c^{d+d'}$. Let $b=d+d'$ and let $f(x)=c^x+c^{b-x}$, for $0 \le x \le b$. It is easy to verify that $f(x)$ is maximized at $x=0$ and $x=b$, and hence $f(x)\le 1+c^{d+d'}$. Once again, since the reuse distances of other pages in $\sigma$ are at most those in $\sigma$, it holds that $\EOA(\sigma')\le\EOA(\sigma)$.  
\item \textbf{Substitution}: Since a substitution is a deletion followed by an insertion, and a deletion cannot increase the expected number of misses, the expected number of misses due to a substitution cannot be larger than those created by an insertion.
\end{itemize}

\paragraph{\textbf{Lower bound}} For the lower bound we consider a sequence that realizes the analysis done above for the upper bound. Let $\rho=x_1x_2\ldots x_m$ be a sequence of $m$ requests, all to distinct pages, and let $\rho'$ be this sequence but reversed. Let $\sigma=(\rho \rho')^\delta$ and let $\sigma'=(\rho x \rho')^\delta$, where $x$ is a page not requested in $\rho$. Clearly, $\Delta(\sigma,\sigma')=\delta$. Let $d_i$ be the reuse distance of the $i$-th page in each repetition of $\rho'$ in $\sigma$. Then, this distance in $\sigma'$ equals $d_i+1$. The reuse distances in $\rho$ are equal in both sequences. Hence, the expected difference in number of misses in each repetition of $\rho$ is given by (\ref{eq:eoa}) above with $d_i=2i$. Then, for the entire sequence we have:
 
$$\EOA(\sigma')-\EOA(\sigma)=\delta\left(1+\frac{1}{k}\sum_{i=0}^{m-1}\left(1-\frac{1}{k}\right)^{2i}\right)=\delta\left(1+\frac{k}{2k-1}\left(1-\left(1-\frac{1}{k}\right)^{2m}\right)\right)$$
 
For any $\epsilon>0$, there exists $m$ such that $\EOA(\sigma')-\EOA(\sigma)>\delta(1+\frac{k}{2k-1}-\epsilon)$.
\qed
\end{proof}}{}
\end{FULL}

\nvsp
\nvsp
\subsubsection{\smoothlru}
\nvsp
\begin{FORPROOFSONLYEXCLUDETHIS}

%In this section, we introduce \smoothlru, a parameterized randomized algorithm that trades competitiveness for smoothness.
We now describe \smoothlru. The main idea of this algorithm is to smooth out the transition from the hit to the miss case.

\begin{FULL}
Recall the notion of age that is convenient in the analysis of \LRU:
\end{FULL}
\begin{SHORT}
The following notion of \emph{age} is convenient in the analysis of \LRU:
\end{SHORT}
The age of page $p$ is the number of distinct pages that have been requested since the previous request to $p$.
\LRU faults if and only if the requested page's age is greater than or equal to $k$, the size of the cache.
An additional request may increase the ages of $k$ cached pages by one.
At the next request to each of these pages, the page's age may thus increase from $k-1$ to $k$, and turn the request from a hit into a miss, resulting in $k$ additional misses.

By construction, under \smoothlru, the hit probability of a request decreases only gradually with increasing age.
The speed of the transition from definite hit to definite miss is controlled by a parameter~$i$, with $0 \leq i < k$.
\end{FORPROOFSONLYEXCLUDETHIS}
Under \smoothlru, the hit probability~$P(\textit{hit}_{\smoothlru_{k,i}}(a))$ of a request to a page with age~$a$ is:
\begin{SHORT}\vspace{-3mm}\end{SHORT}
\begin{equation}\label{eq:smoothlru}
P(\textit{hit}_{\smoothlru_{k,i}}(a)) = \begin{cases}
													1 & :  a  < k-i\\
									 				\frac{k+i-a}{2i+1} & : k-i \leq a < k+i\\
									 				0 & : a \geq k+i
									 			\end{cases}
\end{equation}%
where $k$ is the size of the cache.
\begin{FORPROOFSONLYEXCLUDETHIS}
%\begin{FULL}
Figure~\ref{fig:hitprobabilitiessmoothlru} illustrates this graphically in relation to \LRU for cache size $k=8$ and $i=4$.
%\end{FULL}
%For $i=0$, \smoothlru is identical to \LRU.
%\begin{SHORT}
%\newpage
%\thispagestyle{empty}
%\end{SHORT}

%\begin{figure}[p!]
%\begin{center}
	%\TikzHitProbabilitiesFig{}{.55\textwidth}{.30\textwidth}{
		%\AddCurve{lruprobabilities.txt}{color5}{triangle}{\LRU}
		%\AddCurve{smoothlruprobabilities.txt}{color6}{o}{\smoothlru}
	%}
%\end{center}
%\vspace{-4mm}
%\caption{Hit probabilities of \LRU and \smoothlru in terms of the age of the requested page.\label{fig:hitprobabilitiessmoothlru}}
%\end{figure}

%\begin{FULL}
\begin{figure}[t]
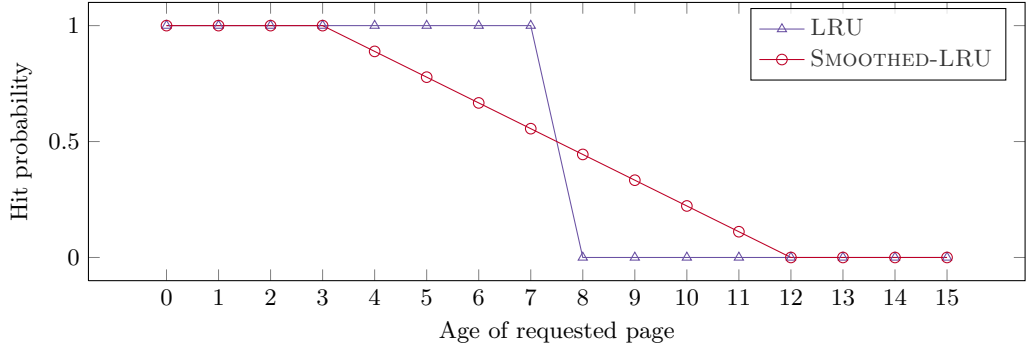

\begin{center}
	\TikzHitProbabilitiesFig{}{.85\textwidth}{.32\textwidth}{
		\AddCurve{lruprobabilities.txt}{color5}{triangle}{\LRU}
		\AddCurve{robustlruprobabilities.txt}{color6}{o}{\smoothlru}
	}
\end{center}
\begin{SHORT}\vspace{-8mm}\end{SHORT}
\caption{Hit probabilities of \LRU and \smoothlru in terms of the age of the requested page\label{fig:hitprobabilitiessmoothlru}}
\begin{SHORT}\vspace{-4mm}\end{SHORT}
\end{figure}
%\end{FULL}

\end{FORPROOFSONLYEXCLUDETHIS}

\begin{restatable}[Smoothness of \smoothlru]{theorem}{thmsmoothlrusmoothness}
\label{thm:smoothlrusmoothness}~\\
\indent$\smoothlru_{k,i}$ is $(1,\delta(\frac{k+i}{2i+1}+1))$-smooth. This is tight.
\end{restatable}
%\begin{SHORT}The proof is similar to that for \LRU.%commented so save space, also lru proof is not in short version right now.
%The key difference is that, in contrast to \LRU, an age increase may only increase the miss probability of a page by~$\frac{1}{2i+1}$.\end{SHORT}
\begin{FULL}
\ifthenelse{\boolean{proof-thmsmoothlrusmoothness}}{%% proof of thmsmoothlrusmoothness
\begin{proof}
The proof of the upper bound is similar to that for \LRU. The key difference is that, in contrast to \LRU, an age increase may only increase the miss probability of a page by $\frac{1}{2i+1}$.
We show that $\smoothlru_{k,i}$ is $(1,\frac{k+i}{2i+1}+1,1)$-smooth. 
Corollary~\ref{cor:1delta} then implies that $\smoothlru_{k,i}$ is $(1,\delta(\frac{k+i}{2i+1}+1))$-smooth. 

Let us first consider how the insertion of one request may affect ages and the expected number of faults.
By definition, the age of any page is only affected from the point of insertion up to its next request.
Only the hit probability of the next request to a page may thus change due to an additional request. 

Under $\smoothlru_{k,i}$, at most $k+i$ pages have a non-zero hit probability at any time.
Only subsequent requests to these pages may increase the expected number of misses.
By construction, increasing the age of a request by one may only decrease the hit probability by $\frac{1}{2i+1}$.
As the inserted request itself may also introduce a fault, the overall number of faults may thus increase by at most~$\frac{k+i}{2i+1}+1$.

Substitutions are similar to insertions: they may increase the ages of at most $k+i$ pages, and the substituted request itself may introduce one additional fault.
The deletion of a request to page $p$ does not increase the ages of other pages.
Only the next request to $p$ may turn from a hit into a miss.\looseness=-1

For tightness, consider the two sequences $\sigma = 1, 2, \dots, k+i, 1, 2, \dots, k+i, y$ and $\sigma' = 1, 2, \dots, k+i, x, 1, 2, \dots, k+i, y$, with $\Delta(\sigma', \sigma) = 1$ if $y > x > k+i$.
%The sequence~$\sigma$ incurs an expected number of $k+i + (k+i)\cdot \frac{2i}{2i+1}+1$ faults, whereas $\sigma'$ incurs $k+i + 1 + k+i+1$ faults.
The difference between the expected number of faults on the two sequences is exactly $\frac{k+i}{2i+1}+1$.
For $\delta > 1$, consider the sequences $\sigma_\delta$ and $\sigma'_\delta$ obtained by concatenating $\delta$ copies of $\sigma$ and $\sigma'$, respectively.\qed
\end{proof}}{}
\end{FULL}

\begin{FORPROOFSONLYEXCLUDETHIS}
For $i=0$, $\smoothlru$ is identical to $\LRU$ and $(1,\delta(k+1))$-smooth.
At the other extreme, for $i=k-1$, $\smoothlru$ is $(1, 2\delta)$-smooth, like the optimal offline algorithm.
However, for larger $i$, \smoothlru is less competitive than \LRU:\looseness=-1
\end{FORPROOFSONLYEXCLUDETHIS}

\begin{restatable}[Competitiveness of \smoothlru]{lemma}{lemcompetitivenesssmoothlru}
For any sequence~$\sigma$ and $l \leq k-i$,
\nvsp\nvsp\nvsp
\[\smoothlru_{k,i}(\sigma) \leq \frac{k-i}{k-i-l+1}\cdot \opt_l(\sigma) + l,\]
where $\opt_l(\sigma)$ denotes the number of faults of the optimal offline algorithm processing $\sigma$ on a fast memory of size $l$.
For $l > k-i$ and any $\alpha$ and $\beta$ there is a sequence $\sigma$, such that
$\smoothlru_{k,i}(\sigma) > \alpha\cdot \opt_l(\sigma) + \beta.$
\end{restatable}

\begin{FULL}
\ifthenelse{\boolean{proof-lemcompetitivenesssmoothlru}}{%% proof of lemcompetitivenesssmoothlru
\begin{proof}
Clearly, $\smoothlru_{k,i}(\sigma) \leq \LRU_{k-i}(\sigma)$ for any sequence $\sigma$, as $\smoothlru_{k,i}$ caches all pages younger than $k$ with probability one.
From Sleator and Tarjan~\cite{Sleator85}, we know that $\LRU_{k-i}(\sigma) \leq \frac{k-i}{k-i-l+1}\cdot \opt_l(\sigma) + l$.

For the second part of the theorem consider the sequence $\sigma_n = (1, \dots, l)^n$, which contains $l$ distinct pages.
The optimal offline algorithm misses exactly $l$ times on this sequence independently of $n$.
For $k-i < l$, on the other hand, $\smoothlru_{k,i}$ has a non-zero miss probability of at least $\frac{1}{2i+1}$ on every request.
For every $\alpha$ and $\beta$ there is an $n$ such that $\smoothlru_{k,i}(\sigma_n) > \alpha\cdot \opt_l(\sigma_n) + \beta$.
\qed
\end{proof}}{}
\end{FULL}

So far we have analyzed \smoothlru based on the hit probabilities given in (\ref{eq:smoothlru}).
We have yet to show that a randomized algorithm satisfying (\ref{eq:smoothlru}) can be realized.
\begin{SHORT}In the full version of the paper~\cite{smoothness_full}\end{SHORT}\begin{FULL}In the following\end{FULL}, we construct a probability distribution on the set of all deterministic algorithms using a fast memory of size~$k$ that satisfies (\ref{eq:smoothlru}).
This is commonly referred to as a \emph{mixed strategy}.
\begin{FULL}

First, we decompose an instance of \smoothlru into $i+1$ instances of a simpler algorithm called \steplru.
Then we show how \steplru can be realized as a mixed strategy.
Like \smoothlru, \steplru is parameterized by $i$, and it exhibits the following hit probabilities in terms of the age of a requested page:
\begin{equation}\label{eq:steplru}
P(\textit{hit}_{\steplru_{k,i}}(a)) = \begin{cases}
													1 & :  a  < k-i\\
									 				\frac{1}{2} & : k-i \leq a < k+i\\
									 				0 & : a \geq k+i
									 			\end{cases}
\end{equation}

\begin{restatable}[Decomposition of \smoothlru in terms of \steplru]{lemma}{lemdecompositionsmoothlru}
\label{lem:decompositionsmoothlru}
For all ages $a$,
%\[\smoothlru_{k,i}(\sigma) = \frac{1}{2i+1}\left(\steplru_{k,0}(\sigma) + \sum_{j=1}^i 2\cdot\steplru_{k,j}(\sigma)\right),\]
\[P(\textit{hit}_{\smoothlru_{k,i}}(a)) = \frac{1}{2i+1}\left(1\cdot P(\textit{hit}_{\steplru_{k,0}}(a)) + \sum_{j=1}^i 2\cdot P(\textit{hit}_{\steplru_{k,j}}(a))\right).\]
\end{restatable}
\ifthenelse{\boolean{proof-lemdecompositionsmoothlru}}{%% proof of lemdecompositionsmoothlru
\begin{proof}
We perform a case distinction on the age $a$:
\begin{enumerate}
	\item For $a < k-i$, all of the above summands are one and the equality holds.
	\item For $a \geq k+i$, all summands are zero and the equality holds as well.
	\item For $k-i \leq a < k$, we get 
		\begin{align*}
			 & \frac{1}{2i+1}\left(1\cdot P(\textit{hit}_{\steplru_{k,0}}(a)) + \sum_{j=1}^i 2\cdot P(\textit{hit}_{\steplru_{k,j}}(a))\right)\\
			 =~& \frac{1}{2i+1}\left(1 + \sum_{j=1}^{k-a-1} 2\cdot 1 + \sum_{j=k-a}^i 2\cdot \frac{1}{2}\right) = \frac{1}{2i+1}\left(1 + 2\cdot(k-a-1) + 1\cdot(i-(k-a)+1)\right)\\
			 =~& \frac{k+i-a}{2i+1} = P(\textit{hit}_{\smoothlru_{k,i}}(a))
		\end{align*}
	\item For $k \leq a < k+i$, we get
		\begin{align*}
			 & \frac{1}{2i+1}\left(1\cdot P(\textit{hit}_{\steplru_{k,0}}(a)) + \sum_{j=1}^i 2\cdot P(\textit{hit}_{\steplru_{k,j}}(a))\right)\\
			 =~& \frac{1}{2i+1}\left(1\cdot 0 + \sum_{j=1}^{a-k} 2\cdot 0 + \sum_{j=a-k+1}^i 2\cdot \frac{1}{2}\right) = \frac{1}{2i+1}\left(0 + 0 + 1\cdot(i-(a-k+1)+1)\right)\\
			 =~& \frac{k+i-a}{2i+1} = P(\textit{hit}_{\smoothlru_{k,i}}(a))
		\end{align*}\qed
\end{enumerate}
\end{proof}}{}

As a consequence, we can realize \smoothlru as a mixed strategy if we can realize \steplru as a mixed strategy.

While the hit probabilities $P(\textit{hit}_{\steplru_{k,i}}(a))$ do not fully define \steplru, by linearity of expectation they are sufficient to determine the expected number of faults on any sequence~$\sigma$, which we denote by $\steplru_{k,i}(\sigma)$. 

\begin{restatable}[\steplru as a mixed strategy]{proposition}{thmsteplrumixed}
\label{thm:steplrumixed}
There is a probability distribution $d : {\cal A} \rightarrow \mathbb{R}$ over a finite set of deterministic paging algorithms ${\cal A}$ using a fast memory of size~$k$, such that for all sequences $\sigma$,
\[\steplru_{k,i}(\sigma) = \sum_{A \in {\cal A}} d(A)\cdot A(\sigma).\]
\end{restatable}
\ifthenelse{\boolean{proof-thmsteplrumixed}}{%% proof of thmsteplrumixed
\begin{proof}
Consider the following deterministic paging algorithm, called \detsteplru, which is parameterized by $k$ and $i$:
\detsteplru always caches the $k-i$ youngest pages.
In addition, it caches $i$ of the $2i$ pages whose ages are between $k-i$ and $k+i-1$.\looseness=-1

Upon a miss to a page of age $a \geq k+i$, \detsteplru replaces the page of age $k+i-1$ if it is in the fast memory. \todo{this requires additional memory compared with a regular cache of the same associativity}
Otherwise, it replaces the page of age $k-i-1$, which is guaranteed to be cached before the request.
Upon a miss to a page of age $a < k+i$, \detsteplru always replaces the page of age $k-i-1$.

Instead of starting with an empty cache, the fast memory is initially filled with $k$ ``dummy'' pages, which may not be requested later on.
The first $k-i$ of these ``dummy'' pages get assigned ages $0$ to $k-i-1$.
A further parameter, $D \subseteq \{k-i, \dots, k+i-1\}$ with $|D|=i$, controls the assignment of ages to the remaining $i$ ``dummy'' pages.
We denote by $\detsteplru_{k,i,D}$ the algorithm that arises when the dummy pages initially assume the ages specified in $D$.

We argue that $\steplru_{k,i}$ results from the uniform distribution over the set of deterministic algorithms ${\cal A} = \{\detsteplru_{k,i,D} \mid D \subseteq \{k-i, \dots, k+i-1\}, |D|=i\}$:
\[\steplru_{k,i}(\sigma) = \sum_{A \in {\cal A}} \frac{1}{{2i \choose i}}\cdot A(\sigma).\]

To see this, consider two arbitrary algorithms $A_1, A_2 \in {\cal A}$ defined by $D_1, D_2$ with $D_1 \neq D_2$.
We claim that after processing an arbitrary sequence $\sigma$, the set of ages of cached pages differ between $A_1$ and $A_2$.
We prove this by induction on the length of $\sigma$:
Clearly this holds for $|\sigma| = 0$ as $D_1 \neq D_2$.
Assume the statement holds before a request.
We perform a case distinction on the age $a$ of the requested page:
\begin{enumerate}
 	\item If $a < k-i$, both caches hit and the set of cached ages does not change in either of the two caches.
 	\item If $a \geq k+i$, then both caches must miss. If one of the two caches stores age $k+i-1$ and the other does not, then one replaces the page with age $k-i-1$ and the other does not. If neither of the caches stores age $k+i-1$, then they both replace the page with age $k-i-1$, which they have in common. If they both store age $k+i-1$, they also replace a common page.
 	\item If $k-i \leq a < k+i$, there are three cases to consider:
 		\begin{enumerate}
 			\item Both caches hit: then there is no change in the set of cached pages.
 			\item Both caches miss: then both replace a common page, the page with age $k-i-1$.
 			\item If one cache misses and the other hits, then one replaces the page with age $k-i-1$ and the other does not.
 		\end{enumerate}
\end{enumerate}

Initially, all of the $2i \choose i$ algorithms in ${\cal A}$ differ from each other regarding the set of ages of cached pages.
Based on the reasoning above they continue to differ from each other after an arbitrary request sequence.
As there are exactly $2i \choose i$ possibilities of choosing $i$ of the $2i$ pages with age $k-i$ to $k+i-1$, each of the possibilities is covered by exactly one algorithm in ${\cal A}$ at any point in time.
A page with age $a$ between $k-i$ and $k+i-1$ is contained in exactly half of these possibilities, and thus the  hit probability is exactly $\frac{1}{2}$ in the uniform distribution over ${\cal A}$.\qed
\end{proof}}{}

\begin{restatable}[\smoothlru as a mixed strategy]{corollary}{corsmoothlrumixed}
There is a probability distribution $d : {\cal A} \rightarrow \mathbb{R}$ over a finite set of deterministic paging algorithms ${\cal A}$ using a fast memory of size~$k$, such that for all sequences $\sigma$,
\[\smoothlru_{k,i}(\sigma) = \sum_{A \in {\cal A}} d(A)\cdot A(\sigma).\]
\end{restatable}
\ifthenelse{\boolean{proof-corsmoothlrumixed}}{%% proof of corsmoothlrumixed
\begin{proof}
This follows immediately from Lemma~\ref{lem:decompositionsmoothlru} and Proposition~\ref{thm:steplrumixed}.\qed
\end{proof}}{}
\end{FULL}

\begin{FORPROOFSONLYEXCLUDETHIS}
\subsection{A Competitive and Smooth Randomized Paging Algorithm: \lrurandom}
\label{sec:lrurandom}

In this section we introduce and analyze \lrurandom, a competitive randomized algorithm that is smoother than any competitive deterministic algorithm.
\lrurandom orders the pages in the fast memory by their recency of use; like \LRU.
Upon a miss, \lrurandom evicts older pages with a higher probability than younger pages.
More precisely, the $i^{th}$ oldest page in the cache is evicted with probability $\frac{1}{i\cdot H_k}$. %, where $H_k = \sum_{i=1}^k \frac{1}{k}$ is the $k^{th}$ harmonic number.
By construction the eviction probabilities sum up to 1: $\sum_{i=1}^k \frac{1}{i\cdot H_k} = \frac{1}{H_k}\cdot \sum_{i=1}^k \frac{1}{i} = 1$.
\lrurandom is \emph{not} demand paging: if the cache is not yet entirely filled, it may still evict cached pages according to the probabilities mentioned above.
\todo{create a separate definition for lrurandom?}

\lrurandom is at least as competitive as strongly-competitive deterministic algorithms:
\end{FORPROOFSONLYEXCLUDETHIS}
\begin{restatable}[Competitiveness of \lrurandom]{theorem}{thmcompetitivenesslrurandom}
For any sequence~$\sigma$, 
\[\lrurandom(\sigma) \leq k \cdot \opt(\sigma).\]
\label{thm:lrurandomcomp}
\end{restatable}
\begin{FULL}
\ifthenelse{\boolean{proof-thmcompetitivenesslrurandom}}{%% proof of thmcompetitivenesslrurandom
\begin{proof}\todo{why is there such a large distance between proof and theorem?}
We actually prove a stronger statement, namely that \lrurandom is $k$-competitive against any \emph{adaptive online adversary}~\cite{Motwani95}.
Our proof is based on a potential argument.

Let $S_\adv$ and $S_\lrur$ be the set of pages contained in the adversary's and \lrurandom's fast memory, respectively.
Further, let $age(p)$ be the age of page $p \in S_\lrur$, i.e., $age(p)$ is~$0$ for the most-recently-used page and $k-1$ for the least-recently-used one among those pages that are in $S_\lrur$.
Based on $age(p)$, we define $s(p) = k - age(p)$.
In other words, $s(p)$ is $1$ for the oldest cached page, and $k$ for the youngest, most-recently-used.
Using these notions we define the following potential function:
 \[\Phi = H_k\cdot \sum_{p \in S_{\lrur}\setminus S_{\adv}} \frac{s(p)}{H_{s(p)}}.\]
We will show that for any page $x$ and any decision of the adversary to evict a page from its memory, we have
\begin{equation}\label{eq:potentiallrurandom}
\lrurandom(x) + \Delta\Phi(x) \leq k\cdot \adv(x),
\end{equation}
where $\lrurandom(x)$ and $\adv(x)$ denote the cost of the request, and $\Delta\Phi(x)$ is the expected change in the potential function.
Note that the potential function is initially zero, given that both caches are initially empty.
Further it is never negative.
From this and (\ref{eq:potentiallrurandom}) the $k$-competitiveness of \lrurandom against an adaptive online adversary follows.
To prove (\ref{eq:potentiallrurandom}), we distinguish four cases upon a request to page $x$:
\begin{enumerate}
	\item \lrurandom hits and \adv hits. Then, $\lrurandom(x) = \adv(x) = 0$. The request may not decrease the ages of pages in $S_\lrur \setminus S_\adv$ and so the potential may not increase, as $\frac{s}{H_s}$ is monotone in $s$.
	\item \lrurandom hits and \adv misses. As $\lrurandom(x) = 0$ and $\adv(x) = 1$, we have to show that $\Delta\Phi(x) \leq k$. 
		The contribution of each page $p \in S_\lrur \setminus S_\adv$ to the potential drops or stays the same, as the ages of these pages may not decrease.		
		The potential may only increase if \adv chooses to evict a page in $S_\lrur \cap S_\adv$.	
		The maximal increase is achieved by evicting the youngest such page $p$.
		After the request, $p$'s age is at least $1$, as it was not the requested page.
		Therefore it contributes at most $H_k\cdot \frac{k-1}{H_{k-1}}$ to the potential, which is
		\[H_k\cdot \frac{k-1}{H_{k-1}} = \left(H_{k-1} + \frac{1}{k}\right)\cdot \frac{k-1}{H_{k-1}} = k-1 + \frac{k-1}{k\cdot H_{k-1}} < k.\]
	\item \lrurandom misses and \adv hits. Then, we have to show that the potential reduces by at least $1$ in expectation. 
		Again, the contribution of no page $p \in S_\lrur \setminus S_\adv$ may increase.
		Further, as \adv may not evict a page, no new page may contribute to the potential.
		We show that the contribution of each page $p \in S_\lrur \setminus S_\adv$ drops by at least 1 in expectation.
		There are three possible cases for a page $p$ with $s(p) = s$:
		\begin{enumerate}
			\item A younger page is replaced, and $p$'s contribution to the potential does not change.
				This happens with probability $\sum_{i=s+1}^k \frac{1}{i\cdot H_k} = 1 - \frac{H_s}{H_k}$.
			\item Page $p$ gets replaced. This happens with probability $\frac{1}{s\cdot H_k}$ and it reduces the potential by~$\frac{H_k\cdot s}{H_s}$.
			\item An older page is replaced, and $p$'s age increases by one.
				This happens with probability $\sum_{i=1}^{s-1} \frac{1}{i\cdot H_k} = \frac{H_{s-1}}{H_k}$ and it reduces the potential by~$H_k\cdot \left(\frac{s}{H_{s}}-\frac{s-1}{H_s-1}\right) = H_k\cdot\frac{H_s-1}{H_sH_{s-1}}.$
		\end{enumerate} 
		So the expected change in potential due to page $p$ is
		 \[\frac{1}{s\cdot H_k}\cdot \frac{-H_k\cdot s}{H_s} + \frac{H_{s-1}}{H_k}\cdot H_k\cdot\frac{1-H_s}{H_sH_{s-1}} = -\frac{1}{H_s} + \frac{1-H_s}{H_s} = -1.\]
	
	\item \lrurandom misses and \adv misses. If, before the request, $S_\lrur \setminus S_\adv \neq \emptyset$, then we can combine the arguments from cases 2 and 3 to show that the potential increases by at most $k-1$.
		This does not cover the case where $S_\lrur = S_\adv$.
		In this case, the potential is increased maximally if the adversary chooses to evict the most-recently-used page.
		If \lrurandom replaces a different page, the potential increases by $H_k\cdot \frac{k-1}{H_{k-1}}$. 
		However, with probability $\frac{1}{k\cdot H_k}$, \lrurandom also replaces the most-recently-used page (in which case the potential remains the same).
		The expected change in potential is thus bounded by 
		\[\left(1-\frac{1}{k H_k}\right) H_k \frac{k-1}{H_{k-1}} = \frac{kH_k-1}{kH_k} H_k \frac{k-1}{H_{k-1}} = \frac{(kH_k-1)(k-1)}{H_{k-1}k} \leq \frac{kH_{k-1}(k-1)}{H_{k-1}k} = k-1.\]
\end{enumerate}
\qed
\end{proof}}{}
\end{FULL}
\begin{SHORT}\vspace{-7mm}\end{SHORT}

\begin{FORPROOFSONLYEXCLUDETHIS}
The proof of Theorem~\ref{thm:lrurandomcomp} applies to an adaptive online adversary. An analysis for an oblivious adversary might yield a lower competitive ratio.

For $k=2$, we also show that \lrurandom is $(1,\delta c)$-smooth, where $c$ is less than $k+1$, which is the best possible among deterministic, demand-paging or competitive algorithms.
Specifically, $c$ is $1+11/6=2.8\bar{3}$. 
%We also show that \lrurandom is $(1,\delta c)$-smooth, where $c$ is less than $k+1$, which is the best possible among deterministic, demand-paging or competitive algorithms.
%Specifically, $c$ is $1+11/6=2.8\bar{3}$, which is less than $3=k+1$, for $k=2$. 
Although our proof technique does not scale beyond $k=2$, we conjecture that this algorithm is in fact smoother than $(1, \delta(k+1))$ for all $k$.
\end{FORPROOFSONLYEXCLUDETHIS}

\begin{restatable}[Smoothness of \lrurandom]{theorem}{thmlrurandomsmoothness}~\\
\indent Let $k=2$. $\lrurandom$ is $(1,\frac{17}{6}\delta)$-smooth.
\end{restatable}

\begin{FULL}
\ifthenelse{\boolean{proof-thmlrurandomsmoothness}}{%% proof of thmlrurandomsmoothness
\begin{proof}
Similar to the proof for $\RAND$, we will look at the distances between state distributions. In this case, however, we will only consider the distances between singleton distributions and distributions that result from requests to a page.  The distances between states shown in Table~\ref{tab:lrurandomdistances} satisfy the following property. Let $s$ and $s'$ be two cache states and let $f_p(s,s')=[p \notin s'] - [p \notin s]$ be the difference between the number of faults when accessing $p$ with caches $s$ and $s'$. Let $d(s,s')$ denote the distance between $s$ and $s'$. Let $D_p(s)$ denote the resulting distribution when requesting $p$ to state $s$. Then, all pairs of states $s,s'$  in Table~\ref{tab:lrurandomdistances} satisfy \todo{mention how we know this?}:

\begin{equation}
d(s,s') \ge \max_p\{f_p(s,s') + \Delta(D_p(s),D_p(s'))\},
\label{eq:distanceineq}
\end{equation}
 
\noindent where $\Delta(D,D')$ is the distance between distributions defined as the minimum cost to transfer the probability mass of $D$ to $D'$. For this definition, the cost to transfer mass between two states $r\in D$ and $r'\in D'$ equals $d(r,r')$.
Note that we only consider cache states that are full. See the discussion at the beginning of the proof of Theorem~\ref{thm:rand} for a justification.

Let $\lrur_{D}(\sigma)$ denote the expected number of faults of $\lrurandom(\sigma)$ when starting from a probability state distribution $D$, and let $\{s\}$ denote the singleton distribution with state $s$. Let $m(\sigma,s,s')=\lrur_{\{s'\}}(\sigma)-\lrur_{\{s\}}(\sigma)$, and let $m(s,s')=\max_\sigma m(\sigma,s,s')$.
We claim that for any pair of states $s$ and $s'$ that satisfy (\ref{eq:distanceineq}), $m(s,s')\le d(s,s')$. 

We prove this by induction on the number of page requests $n$. Let $n=1$ and let $s$ and $s'$ be two states. Let $m^n(s,s')$ be $m(s,s')$ restricted to sequences of length $n$. Then, $m^1(s,s')=\max_p\{f_p(s,s')\}$ and since $\Delta(D,D')\ge 0$ for any pair of distributions and $d(s,s')$ satisfies (\ref{eq:distanceineq}), then $d(s,s')\ge m^1(s,s')$. Now assume that the claim is true for sequences of length $n$. We prove that it holds for sequences of length $n+1$. Let $\sigma$ be a sequence of length $n$ and let $p\cdot\sigma$ be the concatenation of page $p$ and $\sigma$. Then for a pair of states $s$ and $s'$, 

\begin{eqnarray}
m(p\cdot\sigma,s,s') & = & f_p(s,s') + \lrur_{D_p(s')}(\sigma)-\lrur_{D_p(s)}(\sigma)\\
										 & = & f_p(s,s') + \sum_{s_i\in D_p(s')}\lrur_{\{s_i\}}(\sigma)w'(s_i)-\sum_{s_i\in D_p(s)}\lrur_{\{s_i\}}(\sigma)w(s_i)
\label{eq:induction}
\end{eqnarray}

In the last equation, %$m(s_i)=\max_\sigma\lrur_{\{s_i\}}(\sigma)$ 
and $w(s_i)$ and $w'(s_i)$ are the probability of state $s_i$ in $D_p(s)$ and $D_p(s')$. We rewrite (\ref{eq:induction}) by pairing states in both distributions and assigning a weight $\alpha(s_i,s_j)$  to each pair $(s_i,s_j)$:

\begin{eqnarray}
m(p\cdot\sigma,s,s') & = & f_p(s,s') +  \sum_{s_i\in D_p(s'),s_j\in D_p(s)}(\lrur_{\{s_i\}}(\sigma)-\lrur_{\{s_j\}}(\sigma))\alpha(s_i,s_j)\\
										 & \le & f_p(s,s') +  \sum_{s_i\in D_p(s'),s_j\in D_p(s)}m^n(s_i,s_j)\alpha(s_i,s_j)
\label{eq:induction2}
\end{eqnarray}

Above, $\alpha$ is an assignment that satisfies $\forall s_i\in D_p(s'), w'(s_i) = \sum_{s_j\in D_p(s)} \alpha(s_i,s_j)$ and $\forall s_j\in D_p(s), w(s_j) = \sum_{s_i\in D_p(s')} \alpha(s_i,s_j)$. $\alpha$ defines a transfer of mass from distribution $D_p(s')$ to $D_p(s)$. We pick $\alpha$ to be the assignment of weight of minimum cost when the cost of transferring mass between states $s_i$ and $s_j$ is $d(s_i,s_j)$. Since by the inductive hypothesis $m^n(s_i,s_j) \le d(s_i,s_j)$, we have
\begin{eqnarray}
m(p\cdot\sigma,s,s') & \le & f_p(s,s') +  \sum_{s_i\in D_p(s'),s_j\in D_p(s)}d(s_i,s_j)\alpha(s_i,s_j)\\
										 & = & f_p(s,s') + \Delta(D_p(s),D_p(s'))	
\label{eq:induction3}
\end{eqnarray}

Therefore, $m^{n+1}(s,s') \le \max_p\{f_p(s,s') + \Delta(D_p(s),D_p(s'))\} \le d(s,s')$, which proves the claim.

Now we are ready to prove the theorem. 
We will consider a pair of initial states and will show an upper bound on the expected difference between number of misses that can be reached after an insertion, deletion, or substitution.  Let $m_{p,q}(s,s') = \max_{\sigma}\{\lrur_{\{s'\}}(p\cdot\sigma)-\lrur_{\{s\}}(q\cdot\sigma)$\}.
Let $m_p(s,s')=m_{p,p}(s,s')$
 and let $m(s,s')=m_{\bot,\bot}(s,s')$, where $\bot$ is an empty request.

We will prove that $m_{p,q}(s,s)\le 17/6$. If $p=q$, $m_{p,q}(s,s)=0$. Assume $p\ne q$. Let $\bot$ be an empty request. We consider the insertion ($p=\bot$), deletion ($q=\bot$), and substitution cases, and all possible non-redundant requests $p,q$. Let $s=[0\ 1]$ be a cache state, where the left page is the most recently used one. For some pairs of initial states, when the distance between states shown in Table~\ref{tab:lrurandomdistances} is a good enough upper bound on the expected number of misses we use that value.

The values in Table~\ref{tab:lrurandomdistances} are the least-fixed point of the monotone function induced by (\ref{eq:distanceineq}) and have been computed using a Kleene iteration starting from the bottom element $d_\bot(s, s') = 0$.

\begin{table}[!t]
\begin{center}
\caption{Distances between pairs of states. In each entry $d(s,s')$ is an upper bound on the worst-case expected difference between faults when executing $\lrurandom$ starting with $s$ and $s'$.\label{tab:lrurandomdistances}}
\begin{tabular}{c  c | c}
$s$ & $s'$ & $d(s,s')$\\\hline
$[0 \ 1]$ & $[0 \ 1]$ & 0\\
$[0 \ 1]$ & $[0 \ 2]$ & $3/2$\\
$[0 \ 1]$ & $[1 \ 0]$ & $1/2$\\
$[0 \ 1]$ & $[1 \ 2]$ & $3/2$\\
$[0 \ 1]$ & $[2 \ 0]$ & 2\\
$[0 \ 1]$ & $[2 \ 1]$ & 2\\
\end{tabular}
\end{center}
\vspace{-1mm}
\end{table}

\begin{itemize}
\item Insertion: 
\begin{enumerate}
	\item $m_{\bot,1}([0\ 1],[0\ 1]) = m([0 \ 1],[1\ 0]) \le 1/2$
	\item $m_{\bot,2}([0\ 1],[0\ 1]) \le 1 + \max_p\{\frac{2}{3}m_p([0\ 1],[2\ 0]) +\frac{1}{3}m_p([0\ 1],[2\ 1])\}$.
	
	Let $p\in \{0,1,2,3\}$. We compute $m_p([0\ 1],[2\ 0])$:
	\begin{enumerate}
	\item $m_0([0\ 1],[2\ 0]) \le 0+ m([0\ 1],[0\ 2]) = \frac{3}{2}$ 
	\item $m_1([0\ 1],[2\ 0]) \le 1+ \frac{2}{3}m([1\ 0],[1\ 2]) +\frac{1}{3}m([1\ 0],[1\ 0]) \le 1+ \frac{2}{3}\cdot\frac{3}{2}+0 = 2$
	\item $m_2([0\ 1],[2\ 0]) \le -1+ \frac{2}{3}m([2\ 0],[2\ 0]) +\frac{1}{3}m([2\ 1],[2\ 0]) \le -1+ 0+ \frac{1}{3}\cdot\frac{3}{2} = \frac{1}{2}$ 
	\item $m_3([0\ 1],[2\ 0]) \le 0+ \frac{1}{3}m([3\ 0],[3\ 0]) +\frac{1}{3}m([3\ 1],[3\ 2])+\frac{1}{3}m([3\ 1],[3\ 2]) \le 0+ \frac{1}{3}\cdot\frac{3}{2}+\frac{1}{3}\cdot\frac{3}{2} = 1$ 
	\end{enumerate}
	Note that in the last inequality above we can choose to pair any of the states of the resulting distribution to compute an upper bound on the expected number of misses. 
	
	We compute $m_p([0\ 1],[2\ 1])$:
	\begin{enumerate}
	\item $m_0([0\ 1],[2\ 1]) \le 1+ \frac{2}{3}m([0\ 1],[0\ 2]) +\frac{1}{3}m([0\ 1],[0\ 1]) \le 1+ \frac{2}{3}\cdot\frac{3}{2}+0 = 2$ 
	\item $m_1([0\ 1],[2\ 1]) \le 0+ m([1\ 0],[1\ 2]) = \frac{3}{2}$ 
	\item $m_2([0\ 1],[2\ 1]) \le -1+ \frac{2}{3}m([2\ 0],[2\ 1]) +\frac{1}{3}m([2\ 1],[2\ 1]) \le -1+ \frac{2}{3}\cdot\frac{3}{2}+0 = 0$ 
	\item $m_3([0\ 1],[2\ 1]) \le 0+ \frac{2}{3}m([3\ 0],[3\ 2]) +\frac{1}{3}m([3\ 1],[3\ 1]) \le \frac{2}{3}\cdot\frac{3}{2}+0 = 1$ 
	\end{enumerate}
	Plugging in the results in $m_{\bot,2}([0\ 1],[0\ 1])$ we obtain:
	
	$m_{\bot,2}([0\ 1],[0\ 1])\le 1 + \frac{2}{3}m_1([0\ 1],[2\ 0]) +\frac{1}{3}m_1([0\ 1],[2\ 1])=\frac{17}{6}$
\end{enumerate}

\item Deletion:

\begin{enumerate}
	\item $m_{1,\bot}([0\ 1],[0\ 1]) = m([1 \ 0],[0\ 1]) \le 1/2$
	\item $m_{2,\bot}([0\ 1],[0\ 1]) \le -1 + \frac{2}{3}m([2\ 0],[0\ 1]) +\frac{1}{3}m([2\ 1],[0\ 1])=-1 +\frac{2}{3}\cdot\frac{3}{2} + \frac{1}{3}\cdot 2 = \frac{2}{3}$
\end{enumerate}

\item Substitution:

\begin{enumerate}
	\item $m_{0,1}([0\ 1],[0\ 1]) = m_{\bot,1}([0\ 1],[0\ 1]) \le 1/2$ 
	\item $m_{0,2}([0\ 1],[0\ 1]) = m_{\bot,2}([0\ 1],[0\ 1]) \le 17/6$ 
	\item $m_{1,0}([0\ 1],[0\ 1]) = m_{\bot,1}([0\ 1],[0\ 1]) \le 1/2$ 
	\item $m_{1,2}([0\ 1],[0\ 1]) \le 1 + \max_p\{\frac{2}{3}m_p([1\ 0],[2\ 0]) +\frac{1}{3}m_p([1\ 0],[2\ 1])\}$.
	
Let $p\in \{0,1,2,3\}$. We compute $m_p([1\ 0],[2\ 0])$:
	\begin{enumerate}
	\item $m_0([1\ 0],[2\ 0]) \le \frac{3}{2}$ (equivalent to $m_1([0\ 1],[2\ 1])$ in Insertion)  
	\item $m_1([1\ 0],[2\ 0]) \le  2$  (equivalent to $m_0([0\ 1],[2\ 1])$ in Insertion)
	\item $m_2([1\ 0],[2\ 0]) \le 0$ (equivalent to $m_2([0\ 1],[2\ 1])$ in Insertion)
	\item $m_3([1\ 0],[2\ 0])  \le 1$ (equivalent to $m_3([0\ 1],[2\ 1])$ in Insertion)
	\end{enumerate}
	We compute $m_p([1\ 0],[2\ 1])$:
	\begin{enumerate}
	\item $m_0([1\ 0],[2\ 1]) \le 2$ (equivalent to $m_1([0\ 1],[2\ 0])$ in Insertion)  
	\item $m_1([1\ 0],[2\ 1]) \le  \frac{3}{2}$  (equivalent to $m_0([0\ 1],[2\ 0])$ in Insertion)
	\item $m_2([1\ 0],[2\ 1]) \le -\frac{1}{2}$ (equivalent to $m_2([0\ 1],[2\ 0])$ in Insertion)
	\item $m_3([1\ 0],[2\ 1])  \le 1$ (equivalent to $m_3([0\ 1],[2\ 0])$ in Insertion)
	\end{enumerate}
		Therefore, $m_{1,2}([0\ 1],[0\ 1]) \le 1 + \frac{2}{3}m_1([1\ 0],[2\ 0]) +\frac{1}{3}m_1([1\ 0],[2\ 1])=\frac{17}{6}$.
	\item $m_{2,1}([0\ 1],[0\ 1]) = m_{2,\bot}([0\ 1],[0\ 1]) + m_{\bot,1}([0\ 1],[0\ 1]) \le \frac{2}{3} +\frac{1}{2} = \frac{7}{6}$
	\item $m_{2,3}([0\ 1],[0\ 1]) \le 0 + \frac{2}{3}m([2\ 0],[3\ 0]) +\frac{1}{3}m([2\ 1],[3\ 1])=\frac{2}{3}\cdot 2 + \frac{1}{3}\cdot 2 = 2$
\end{enumerate}

\end{itemize}
The maximum of all the upper bounds derived above is $17/6$ and hence $\lrurandom(\sigma')-\lrurandom(\sigma)\le 17/6=2.8\bar{3}$ for any $\sigma$ and $\sigma'$ with $\Delta(\sigma,\sigma')=1$. Corollary~\ref{cor:1delta} implies the theorem.
\qed 
\end{proof}}{}
\end{FULL}

%\begin{SHORT}\vspace{-2mm}\end{SHORT}
\begin{FORPROOFSONLYEXCLUDETHIS}
\nvsp\nvsp\nvsp\nvsp
\begin{restatable}[Smoothness of \lrurandom]{conj}{conlrurandomsmoothness}~\\
\indent $\lrurandom$ is $(1,\Theta(H_k^2)\delta)$-smooth.
\end{restatable}
\begin{SHORT}\vspace{-1mm}\end{SHORT}
\todo{discuss this conjecture}
%The conjecture is based on \lrurandom's behavior on a pair of sequences that we believe maximize the difference in misses for a single perturbation.
%\begin{SHORT}\vspace{-2mm}\end{SHORT}

\begin{figure}[t]
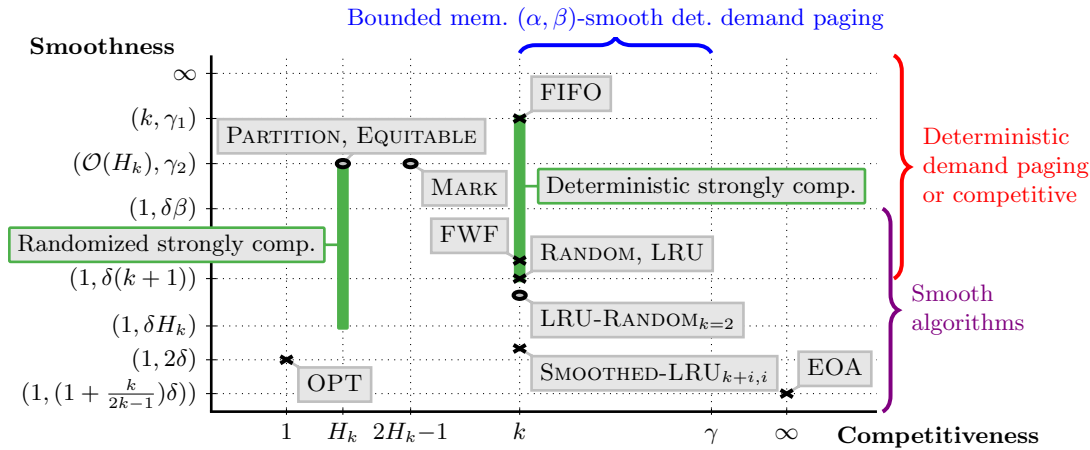

\begin{center}
\begin{SHORT}
\scalebox{0.7925}{\robustnessvscompetitiveness}
\end{SHORT}
\begin{FULL}
\robustnessvscompetitiveness
\end{FULL}
\begin{SHORT}\vspace{-4mm}\end{SHORT}
\caption{Schematic view of the smoothness and competitiveness landscape. Crosses indicate tight results, whereas ellipses indicate upper bounds. Braces denote upper and lower bounds on the smoothness or competitiveness of classes of algorithms. For simplicity of exposition, $\gamma_1$ and $\gamma_2$ are left unspecified; $\gamma$ can be chosen arbitrarily. More precise statements are provided in the respective theorems.}\label{fig:landscape}
\begin{SHORT}\vspace{-8mm}\end{SHORT}
\end{center}
\end{figure}

\nvsp\nvsp\nvsp\nvsp
\section{Discussion}
\nvsp\nvsp\nvsp\nvsp

We have determined fundamental limits on the smoothness of deterministic and randomized paging algorithms.
No deterministic competitive algorithm can be smoother than $(1, \delta(k+1))$-smooth.
Under the restriction to bounded-memory algorithms, which is natural for hardware implementations of caches, smoothness implies competitiveness.
\begin{SHORT}
\LRU is strongly competitive, and it matches the lower bound for deterministic competitive algorithms.
\end{SHORT}
\begin{FULL}
\LRU is strongly competitive, and it matches the lower bound for deterministic competitive algorithms, while \FIFO matches the upper bound.
\end{FULL}
There is no trade-off between smoothness and competitiveness for deterministic algorithms.

In contrast, among randomized algorithms, we have identified \smoothlru, an algorithm that is very smooth, but not competitive.
In particular, it is smoother than any strongly-competitive randomized algorithm may be.
The well-known randomized algorithms \MARK, \partition, and \equitable are not smooth.
It is an open question, whether there is a randomized ``\LRU sibling'' that is both strongly-competitive and $(1, \delta H_k)$-smooth.
With \lrurandom we introduce a randomized algorithm that is at least as competitive as any deterministic algorithm, yet provably smoother, at least for $k=2$. 
Its exact smoothness remains open.
Figure~\ref{fig:landscape} schematically illustrates many of our results.

\subsubsection*{Acknowledgments.}
\begin{SHORT}
\nvsp\nvsp
This work was partially supported by the DFG as part of the SFB/TR 14 AVACS.
\nvsp\nvsp
\end{SHORT}
\begin{FULL}
This work was partially supported by the German Research Council (DFG) as part of the Transregional Collaborative Research Center ``Automatic Verification and Analysis of Complex Systems'' (SFB/TR 14 AVACS).
\end{FULL}

\end{FORPROOFSONLYEXCLUDETHIS}

\begin{SHORT}\vspace{-1mm}\end{SHORT}
\bibliographystyle{splncs}
\bibliography{cache}

\begin{CUSTOM}
\ifthenelse{\boolean{atleastoneomittedproof}}
{
%\newpage
\appendix
\section*{Appendix}

\ifthenelse{\boolean{proof-lemonetodelta}}{}
{ %else
\lemonetodelta*
%% proof of lemonetodelta
\begin{proof}
By induction on $\delta$. The case $\delta=1$ is trivial. Assume the hypothesis is true for $1< \delta \le h$.  Let $\sigma_{h+1}$ and $\sigma$ be any pair of sequences such that $\Delta(\sigma,\sigma_{h+1})=h+1$. Then there exists a sequence $\sigma_h$ such that $\Delta(\sigma,\sigma_h)=h$ and  $\Delta(\sigma_h,\sigma_{h+1})=1$. Since $A$ is  $(\alpha,\beta,1)$-smooth, then 
$A(\sigma_{h+1})\le \alpha A(\sigma_h)+\beta$.  By the inductive hypothesis,
 $A(\sigma_h)\le \alpha^h A(\sigma)+\beta\sum_{i=0}^{h-1}\alpha^i$. Therefore,  $A(\sigma_{h+1})\le \alpha(\alpha^h A(\sigma)+\beta\sum_{i=0}^{h-1}\alpha^i)+\beta=\alpha^{h+1} A(\sigma)+\beta\sum_{i=0}^{h}\alpha^i$, and thus $A$ is  $(\alpha^\delta,\beta\sum_{i=0}^{\delta-1}\alpha^i)$-smooth.\looseness=-1\qed
%For any pair of sequences $\sigma$ and $\sigma_\delta$, if $\Delta(\sigma, \sigma_\delta)=\delta$ then there exist sequences $\sigma_1,\sigma_2,\ldots,\sigma_{\delta-1}$ such that $\Delta(\sigma,\sigma_{1})=1$, $\Delta(\sigma_{\delta-1},\sigma_{\delta})=1$, and  $\Delta(\sigma_i,\sigma_{i+1})=1$ for all $1\le i < \delta-1$. Since $A$ is  $(\alpha,\beta,1)$-smooth then by composition of the inequalities 
\end{proof}
}

\ifthenelse{\boolean{proof-thmlowerbounddemandpaging}}{}
{ %else
\thmlowerbounddemandpaging*
%% proof of thmlowerbounddemandpaging
\begin{proof}
Let $A$ be any deterministic, demand-paging algorithm. 
Using $k+1$ distinct pages, we can construct a sequence $\sigma_A(\delta)$ of length $k+\delta(k+1)$ such that $A$ faults on every request: first request the $k+1$ distinct pages in any order; then arbitrarily extend the sequence by requesting the page that $A$ has just evicted. 
Let $p$ be the page that occurs least frequently in $\sigma_A(\delta)$.
By removing all requests to $p$ from $\sigma_A(\delta)$, we obtain a sequence $\sigma'_A(\delta)$ that consists of $k$ distinct pages only.
By assumption $A$ is demand paging. 
Thus, $A$ incurs only $k$ page faults on the entire sequence.
Assume for a contradiction that $A$ is $(1, \delta(k+1-\epsilon))$-smooth for some $\epsilon > 0$.
Then, we have by definition:
\[A(\sigma_A(\delta)) \leq 1\cdot A(\sigma'_A(\delta)) + \Delta(\sigma'_A(\delta), \sigma_A(\delta))\cdot(k+1-\epsilon).\]
Clearly, $p$ occurs at most $\left\lfloor \frac{k+\delta(k+1)}{k+1} \right\rfloor = \delta$ times in $\sigma_A(\delta)$.
So $\Delta(\sigma'_A(\delta), \sigma_A(\delta)) \leq \delta$, and we get:
\begin{align*}
 k+\delta(k+1) &\leq  k + \delta\cdot(k+1-\epsilon)\\
\Leftrightarrow \epsilon &\leq 0, 
\end{align*}
which contradicts the assumption that $\epsilon > 0$.\qed
\end{proof}
}

\ifthenelse{\boolean{proof-thmlowerboundcompetitive}}{}
{ %else
\thmlowerboundcompetitive*
%% proof of thmlowerboundcompetitive
\begin{proof}
Let $A$ be any $c$-competitive deterministic online paging algorithm. The proof is essentially the same as the one for Theorem~\ref{thm:lowerbounddemandpaging} with the number of faults of $A$ on $\sigma'_A(\delta)$ being at most $ck + \beta$, for some constant $\beta$. This follows from the competitiveness of $A$ and the fact that $OPT$ makes at most $k$ faults on $\sigma'_A(\delta)$.
	Assuming for a contradiction that $A$ is $(1, \delta(k+1-\epsilon))$-smooth for some $\epsilon$, we get
\begin{align*}
	A(\sigma_A(\delta)) & \leq 1\cdot A(\sigma'_A(\delta)) + \Delta(\sigma'_A(\delta), \sigma_A(\delta))\cdot(k+1-\epsilon)\\
	\Rightarrow k+\delta(k+1) & \leq ck + \beta + \delta\cdot(k+1-\epsilon)\\
	\Rightarrow \delta\epsilon & \leq (c-1)k+\beta
\end{align*}
For any $\epsilon > 0$, there is a $\delta_m$ such that any $\delta > \delta_m$ contradicts the above inequality.
By a slight generalization of Corollary~\ref{cor:1delta} this implies that the algorithm is not $(1,\delta(k+1-\epsilon),\delta)$-smooth for any $\delta$ and $\epsilon > 0$.
\qed
\end{proof}
}

\ifthenelse{\boolean{proof-thmoptsmoothness}}{}
{ %else
\thmoptsmoothness*

}

\ifthenelse{\boolean{proof-thmcompetitivesmoothness}}{}
{ %else
\thmcompetitivesmoothness*
\begin{proof}
Let $\sigma'$ be a sequence such that $\Delta(\sigma,\sigma')=\delta$. By Theorem~\ref{thm:optsmoothness}, $\opt(\sigma')\le \opt(\sigma)+2\delta$. Therefore, $A(\sigma') \le c\cdot (\opt(\sigma) +2\delta) +\beta \le  c\cdot A(\sigma) +2\delta c +\beta$.\qed
\end{proof}
}

\ifthenelse{\boolean{proof-thmboundedmemorysmooth}}{}
{ %else
\thmboundedmemorysmooth*
%% proof of thmboundedmemorysmooth
\begin{proof}
Assume algorithm~$A$ is non-competitive.
Then, there is no bound on the number of misses in a single $k$-phase for $A$:  otherwise, if $r$ is a bound on the number of misses of $A$ in every phase, then $A$ is competitive with competitive ratio $c \leq r$.

Let $n$ be the number of states of $A$.
Within a $k$-phase, a demand-paging algorithm can reach at most $(2k)^k$ different configurations: each of the $k$ slots can either contain one of the $k$ ``old'' pages cached at the start of the $k$-phase, or one of the up to $k$ ``new'' pages requested within the phase. \todo{an exact bound assumes no other information held by the cache, but we say that there might be some more finite amount of additional state}
Let $\sigma$ be a sequence of minimal length that ends on a phase in which $A$ misses more than $n\cdot (2k)^k$ times.
 By the pigeon-hole principle, $A$ must assume the same state and configuration pair twice within that phase.
 Due to the minimality of the sequence, $A$ must fault at least once between those two occurrences. 
 By repeating the sequence of requests between the two occurrences, we can thus pump up the sequence and the number of faults arbitrarily without increasing the number of phases.
By removing the finite prefix of the sequence that comprises all but the final phase, we can construct a sequence $\sigma'$ containing at most $k$ distinct pages.
Any demand-paging algorithm, in particular $A$, will fault at most $k$ times on this sequence. 
The edit distance between $\sigma'$ and $\sigma$ is finite, but the difference in faults is unbounded.
This shows that $A$ is not $(\alpha, \beta)$-smooth for any $\alpha$ and $\beta$.\looseness=-1\qed
\end{proof}
}

\ifthenelse{\boolean{proof-thmlrusmoothness}}{}
{ %else
\thmlrusmoothness*
\begin{proof}
We show that $\LRU$ is $(1,k+1,1)$-smooth. 
Corollary~\ref{cor:1delta} then immediately implies that $\LRU$ is $(1,\delta(k+1))$-smooth. 
Tightness follows from Theorem~\ref{thm:lowerbounddemandpaging} as $\LRU$ is demand paging.
%% proof of thmlrusmoothness
\todo{write this down more elegantly}
To analyze $\LRU$, it is convenient to introduce the notion of \emph{age}. 
The age of page $p$ is the number of distinct pages that have been requested since the previous request to $p$. 
Before their first request, all pages have age $\infty$.
A request to page $p$ results in a fault if and only if $p$'s age is greater than or equal to $k$, the size of the cache.
Finite ages are unique, i.e., no two pages have the same age less than $\infty$.
At any time at most $k$ pages are cached, and at most $k$ pages have an age less than $k$.

Let us now consider how the insertion of one request may affect ages and the expected number of faults.
By definition, the age of any page is only affected from the point of insertion up to its next request.
Only the next request to a page may thus turn from a hit into a miss. 
At any time at most $k$ pages have an age less than $k$. 
So at most $k$ requests may turn from hits into misses.
As the inserted request itself may also introduce a fault, the overall number of faults may thus increase by at most $k+1$.

Substitutions are similar to insertions: they turn at most $k$ succeeding hits into misses, and the substituted request itself may introduce one additional fault.
The deletion of a request to page $p$ does not increase the ages of other pages.
Only the next request to $p$ may turn from a hit into a miss.\looseness=-1 \qed
\end{proof}
}

\ifthenelse{\boolean{proof-lemsuffix}}{}
{ %else
\begin{restatable}{proposition}{lemsuffix}
\label{prop:suffix}
For a sequence $\sigma$, let $\Phi(\sigma)$ denote the number of phases in its $k$-phase partition.
Let $\sigma$ be a sequence, let $\rho$ be a suffix of $\sigma$, and let $\ell$ and $\ell'$ denote the number of distinct pages in the last phase of $\sigma$ and~$\rho$, respectively. Then $\Phi(\rho)\le \Phi(\sigma)$. Furthermore, if $\Phi(\rho)= \Phi(\sigma)$ then $\ell'\le \ell$. 
\end{restatable}
%% proof of lemsuffix
\begin{proof}
Let $i_j$ and $i_j'$ denote the indices in $\sigma$ of the first request of the $j^{th}$ phase in $\sigma$ and $\rho$, respectively, 
%with $i_{\Phi(\sigma)+1}-1$ and $i_{\Phi(\rho)+1}-1$ being the indices of the last pages of $\sigma$ and $\rho$. 
with $i_j=|\sigma|+1$ for $j>\Phi(\sigma)$ and $i_j'=|\rho'|+1$ for $j>\Phi(\rho')$.
Then, for all $j$ it holds that $i_j\le i_j'$. We prove this by induction on $j$. The case $j=1$ is trivially true as $i_1=1$ and $i_1'\ge 1$ since $\rho$ is a suffix of $\sigma$. Suppose that the hypothesis holds for $1 < j \le n$. It is easy to see that it holds for $j=n+1$: since there are at most $k$ distinct pages between $i_j$ and $i_{j+1}-1$, inclusive, and by the inductive hypothesis $i_j'\ge i_j$, then the request that ends the $j^{th}$ phase in $\rho$ cannot be earlier than $i_{j+1}$, and hence $i_{j+1}'\ge i_{j+1}$. Since this is true for all phases including the last one, then $\Phi(\rho)\le \Phi(\sigma)$.
Now, assume that $\Phi(\rho)=\Phi(\sigma)=m$. Then $i_{m}' < |\rho'|+1$ and by the proof above $i_{m}\le i_{m}'$, which implies that $\ell'\le \ell$.\qed
 \end{proof}
}

\ifthenelse{\boolean{proof-lemphases}}{}
{ %else
\lemphases*

}

\ifthenelse{\boolean{proof-thmfwfsmoothness}}{}
{ %else
\thmfwfsmoothness*
%% proof of thmfwfsmoothness
\begin{proof}
Let $\sigma$ and $\sigma'$ be two sequences such that $\Delta(\sigma,\sigma')=1$. Let $\Phi(\sigma)$ (resp. $\Phi(\sigma')$) be the number of phases in the $k$-phase partition of $\sigma$ (resp. $\sigma'$), and let $\ell$ (resp. $\ell'$) be the number distinct pages in the last phase of the partition of $\sigma$ (resp. $\sigma'$). 
FWF misses exactly $k$ times in any phase of a sequence, except possibly for the last one, in which it misses a number of times equal to the number of distinct pages in the phase. Then, $\FWF(\sigma)=k \cdot (\Phi(\sigma)-1) + \ell$, and  $\FWF(\sigma')=k \cdot (\Phi(\sigma')-1) + \ell'$.
By Lemma~\ref{prop:phases}, if $\Phi(\sigma')=\Phi(\sigma)+2$, $\ell'\le \ell$, and thus $\FWF(\sigma')\le k(\Phi(\sigma)+2-1) + \ell=\FWF(\sigma)+2k$. Otherwise, if $\Phi(\sigma')\le \Phi(\sigma)+1$, then $\FWF(\sigma')\le k(\Phi(\sigma)+1-1) + \ell'=\FWF(\sigma)-\ell+k+\ell'$. Since $\ell\ge 0$ and $\ell'\le k$, $\FWF(\sigma')\le \FWF(\sigma)+2k$. The upper bound in the lemma follows by Corollary~\ref{cor:1delta}.
To see that this upper bound is tight, let $\sigma = (x_1,\ldots,x_k)^{2\delta+1}$, where $x_i\ne x_j$ for all $i\ne j$. Let $\sigma' = x_1,\ldots,x_k(x_{k+1},x_1,\ldots,x_k,x_1,\ldots,x_k)^\delta$, where $x_{k+1}\ne x_i$ for all $i\le k$. Thus, $\Delta(\sigma,\sigma')=\delta$. Clearly $\FWF(\sigma)=k$, while $\FWF(\sigma')=k+2\delta k$, and hence $\FWF(\sigma')=\FWF(\sigma)+2\delta k$.\qed
\end{proof}
}

\ifthenelse{\boolean{proof-thmfifosmoothness}}{}
{ %else
\thmfifosmoothness*

}

\ifthenelse{\boolean{proof-thmlowerboundrandomizeddemangpaging}}{}
{ %else
\thmlowerboundrandomizeddemangpaging*

}

\ifthenelse{\boolean{proof-thmlowerboundstronglycompetitive}}{}
{ %else
\thmlowerboundstronglycompetitive*

}

\ifthenelse{\boolean{proof-thmlowerboundpartitionequitable}}{}
{ %else
\thmlowerboundpartitionequitable*

}

\ifthenelse{\boolean{proof-thmlowerboundmark}}{}
{ %else
\thmlowerboundmark*

}

\ifthenelse{\boolean{proof-thmrandomsmoothness}}{}
{ %else
\thmrandomsmoothness*

}

\ifthenelse{\boolean{proof-thmeoasmoothness}}{}
{ %else
\thmeoasmoothness*

}

\ifthenelse{\boolean{proof-thmsmoothlrusmoothness}}{}
{ %else
\thmsmoothlrusmoothness*

}

\ifthenelse{\boolean{proof-lemcompetitivenesssmoothlru}}{}
{ %else
\lemcompetitivenesssmoothlru*
%% proof of lemcompetitivenesssmoothlru
\begin{proof}
Clearly, $\smoothlru_{k,i}(\sigma) \leq \LRU_{k-i}(\sigma)$ for any sequence $\sigma$, as $\smoothlru_{k,i}$ caches all pages younger than $k$ with probability one.
From Sleator and Tarjan~\cite{Sleator85}, we know that $\LRU_{k-i}(\sigma) \leq \frac{k-i}{k-i-l+1}\cdot \opt_l(\sigma) + l$.

For the second part of the theorem consider the sequence $\sigma_n = (1, \dots, l)^n$, which contains $l$ distinct pages.
The optimal offline algorithm misses exactly $l$ times on this sequence independently of $n$.
For $k-i < l$, on the other hand, $\smoothlru_{k,i}$ has a non-zero miss probability of at least $\frac{1}{2i+1}$ on every request.
For every $\alpha$ and $\beta$ there is an $n$ such that $\smoothlru_{k,i}(\sigma_n) > \alpha\cdot \opt_l(\sigma_n) + \beta$.
\qed
\end{proof}
}

\ifthenelse{\boolean{proof-lemdecompositionsmoothlru}}{}
{ %else
\lemdecompositionsmoothlru*
%% proof of lemdecompositionsmoothlru
\begin{proof}
We perform a case distinction on the age $a$:
\begin{enumerate}
	\item For $a < k-i$, all of the above summands are one and the equality holds.
	\item For $a \geq k+i$, all summands are zero and the equality holds as well.
	\item For $k-i \leq a < k$, we get 
		\begin{align*}
			 & \frac{1}{2i+1}\left(1\cdot P(\textit{hit}_{\steplru_{k,0}}(a)) + \sum_{j=1}^i 2\cdot P(\textit{hit}_{\steplru_{k,j}}(a))\right)\\
			 =~& \frac{1}{2i+1}\left(1 + \sum_{j=1}^{k-a-1} 2\cdot 1 + \sum_{j=k-a}^i 2\cdot \frac{1}{2}\right) = \frac{1}{2i+1}\left(1 + 2\cdot(k-a-1) + 1\cdot(i-(k-a)+1)\right)\\
			 =~& \frac{k+i-a}{2i+1} = P(\textit{hit}_{\smoothlru_{k,i}}(a))
		\end{align*}
	\item For $k \leq a < k+i$, we get
		\begin{align*}
			 & \frac{1}{2i+1}\left(1\cdot P(\textit{hit}_{\steplru_{k,0}}(a)) + \sum_{j=1}^i 2\cdot P(\textit{hit}_{\steplru_{k,j}}(a))\right)\\
			 =~& \frac{1}{2i+1}\left(1\cdot 0 + \sum_{j=1}^{a-k} 2\cdot 0 + \sum_{j=a-k+1}^i 2\cdot \frac{1}{2}\right) = \frac{1}{2i+1}\left(0 + 0 + 1\cdot(i-(a-k+1)+1)\right)\\
			 =~& \frac{k+i-a}{2i+1} = P(\textit{hit}_{\smoothlru_{k,i}}(a))
		\end{align*}\qed
\end{enumerate}
\end{proof}
}

\ifthenelse{\boolean{proof-thmsteplrumixed}}{}
{ %else
\thmsteplrumixed*

}

\ifthenelse{\boolean{proof-corsmoothlrumixed}}{}
{ %else
\corsmoothlrumixed*
%% proof of corsmoothlrumixed
\begin{proof}
This follows immediately from Lemma~\ref{lem:decompositionsmoothlru} and Proposition~\ref{thm:steplrumixed}.\qed
\end{proof}
}

\ifthenelse{\boolean{proof-thmcompetitivenesslrurandom}}{}
{ %else
\thmcompetitivenesslrurandom*

}

\ifthenelse{\boolean{proof-thmlrurandomsmoothness}}{}
{ %else
\thmlrurandomsmoothness*

}

}
{}
\end{CUSTOM}

\end{document}